\documentclass[11pt]{article}
\usepackage[latin9]{inputenc}
\usepackage{float}
\usepackage{bm}
\usepackage{amsmath}
\usepackage{amsthm}
\usepackage{graphicx}
\usepackage{authblk}

\makeatletter

\providecommand{\tabularnewline}{\\}

\theoremstyle{plain}
\newtheorem{thm}{\protect\theoremname}
\theoremstyle{definition}

\theoremstyle{plain}
\newtheorem{prop}[thm]{\protect\propositionname}
\theoremstyle{remark}
\newtheorem{rem}[thm]{\protect\remarkname}
\theoremstyle{plain}
\newtheorem{lem}[thm]{\protect\lemmaname}

\providecommand{\definitionname}{Definition}
\providecommand{\propositionname}{Proposition}
\providecommand{\remarkname}{Remark}
\providecommand{\theoremname}{Theorem}
\providecommand{\lemmaname}{Lemma}

\usepackage{amsmath,amssymb,amsfonts,amsthm}
\usepackage[normalem]{ulem}
\usepackage{soul}
\usepackage{moreverb,rotating,graphics}
\usepackage[square,numbers]{natbib}
\usepackage{hyperref}
\hypersetup{
	colorlinks,
	linkcolor={blue!50!blue},
	citecolor={red},
	urlcolor={blue!50!blue}
}
\usepackage{float}
\usepackage{appendix}
\usepackage{epsfig}
\usepackage{tikz}
\usepackage{graphicx}
\usepackage{lipsum}
\usepackage{enumerate}
\usepackage{import}
\usepackage[normalem]{ulem}
\usepackage{color}
\usepackage{caption}
\usepackage{subfigure}
\usepackage[ruled,vlined]{algorithm2e}
\usepackage{bm}
\usepackage[absolute,overlay]{textpos}
\usepackage{array}
\usepackage{multicol}
\usepackage{mathrsfs}
\usepackage{tablefootnote}
\usepackage{minitoc}
\usepackage{cases}
\allowdisplaybreaks

\usepackage[Conny]{fncychap}
\usepackage[a4paper]{geometry}
\theoremstyle{plain}
\newtheorem{theorem}{Theorem}[section]
\numberwithin{equation}{section}
\newtheorem{lemma}[thm]{Lemma}
\newtheorem{proposition}[thm]{Proposition}

\newtheorem{corollary}[thm]{Corollary}
\newtheorem{assumption}[thm]{Assumption}

\theoremstyle{definition}

\theoremstyle{remark}
\newtheorem{remark}{Remark}

\def\sqw{\hbox{\rlap{\leavevmode\raise.3ex\hbox{$\sqcap$}}$%
\sqcup$}}
\def\cqfd{\ifmmode\sqw\else{\ifhmode\unskip\fi\nobreak\hfil
\penalty50\hskip1em\null\nobreak\hfil\sqw
\parfillskip=0pt\finalhyphendemerits=0\endgraf}\fi}

\hypersetup{
    colorlinks,
    linkcolor={blue!50!blue},
    citecolor={red},
    urlcolor={blue!50!blue}
}
\renewcommand{\eqref}[1]{(\ref{#1})}

\usepackage{todonotes}

\newcommand{\1}{\mathbf{1}}

\newcommand{\Error}{\mathrm{Error}}
\newcommand{\Hol}{\mathrm{Hol}}

\newcommand{\R}{\mathbb{R}}

\newcommand{\dd}{\mathrm d}

\newcommand{\cA}{\mathcal{A}}

\newcommand{\cC}{\mathcal{C}}

\newcommand{\cE}{\mathcal{E}}

\newcommand{\cN}{\mathcal{N}}
\newcommand{\cO}{\mathcal{O}}

\newcommand{\bD}{{\mathbb D}}
\newcommand{\bE}{{\mathbb E}}

\newcommand{\bN}{{\mathbb N}}

\newcommand{\bP}{{\mathbb P}}

\newcommand{\bR}{{\mathbb R}}

\newcommand{\be}{\begin{equation}}
\newcommand{\ee}{\end{equation}}

\newcommand{\mino}{<}

\newcommand{\call}{\mathrm{call}}

\newcommand{\VIX}{\mathrm{VIX}}
\newcommand{\ve}{\varepsilon}

\newcommand{\mup}{\mu_{\mathrm{P}}}
\newcommand{\sigmap}{\sigma_{\mathrm{P}}}
\newcommand{\BS}{\mathrm{BS}}
\newcommand{\Pcall}{P^{\mathrm{call}}}
\newcommand{\Pput}{P^{\mathrm{put}}}
\newcommand{\PF}{P^{\mathrm{F}}}
\newcommand{\Deltacall}{\Delta^{\rm{call}}}
\newcommand{\Deltaput}{\Delta^{\rm{put}}}

\newcommand{\Kdot}{K^{\cdot}}

\newcommand{\Ydot}[1]{\ensuremath{ 
		Y_{#1}^{\cdot}
}}
\newcommand{\xidot}[1]{\ensuremath{ 
		\xi_{#1}^{\cdot}
}}
\newcommand{\SP}{\rm{S\&P}500}

\newcommand{\egl}{\overset{d}{=}}

\DeclareMathAlphabet\mathbfcal{OMS}{cmsy}{b}{n}

\newcommand{\Var}{\mathrm {Var}}

\usepackage[framemethod=tikz]{mdframed}
\usetikzlibrary{decorations.pathmorphing,calc}
\newcommand\wavydecor{%
	\draw[decoration={coil,aspect=0.1,segment length=5pt,amplitude=1.0pt},decorate,line width=1.5pt,black]
	(O|-P) -- (O);
}
\newmdenv[
hidealllines=true,
innerleftmargin=10pt,
innerrightmargin=0pt,
innertopmargin=0pt,
innerbottommargin=0pt,
leftmargin=-10pt,
skipabove=.5\baselineskip,
skipbelow=.5\baselineskip,
singleextra={\wavydecor},
firstextra={\wavydecor},
secondextra={\wavydecor},
middleextra={\wavydecor}
]{done}

\makeatother

\begin{document}
\title{\textbf{Weak approximations and VIX option price expansions in
forward variance curve models}}
\author[1,2]{F.\ Bourgey\thanks{(Corresponding author) florian.bourgey@polytechnique.edu, bourgeyflorian@gmail.com}}
\author[1]{S.\ De Marco\thanks{stefano.de-marco@polytechnique.edu}}
\author[1]{E.\ Gobet\thanks{emmanuel.gobet@polytechnique.edu}}
\affil[1]{Centre de Math\'ematiques Appliqu\'ees (CMAP), CNRS, Ecole Polytechnique, Institut Polytechnique de Paris. Route de Saclay, 91128 Palaiseau Cedex, France. }
\affil[2]{Bloomberg L.P., Quantitative Research, 3 Queen Victoria St, London
	EC4N 4TQ, UK.}
\maketitle
\begin{abstract}
\noindent	We provide explicit approximation formulas for VIX futures and options in forward variance models, with particular emphasis 
	on the family of so-called Bergomi models: the one-factor Bergomi model [Bergomi, Smile dynamics II, Risk, 2005], the rough Bergomi 
	model [Bayer, Friz, and Gatheral, Pricing under rough volatility, Quantitative Finance, 16(6):887-904, 2016], and an enhanced version of the rough model that can generate realistic positive skew for VIX smiles -- introduced simultaneously by De Marco [Bachelier World Congress, 2018] and Guyon 
	[Bachelier World Congress, 2018] on the lines of [Bergomi, Smile dynamics III, Risk, 2008], that we refer to as ``mixed rough Bergomi model''. 
	Following the methodology set up in [Gobet and Miri, Weak approximation of averaged diffusion processes. Stochastic Process.\ Appl., 124(1):475-504, 2014],  we derive weak approximations for the law of the VIX, leading to option price approximations under the form of explicit combinations of Black-Scholes prices and greeks.
As new contributions, we cope with the fractional integration kernel appearing in rough models and treat the case of non-smooth payoffs, so to encompass
VIX futures, call and put options. We stress  that our approach does not rely on small-time asymptotics nor small-parameter (such as small volatility-of-volatility) asymptotics, and can therefore be applied to any option maturity and a wide range of parameter configurations. Our results are illustrated by several numerical experiments and calibration tests to VIX market data.
\end{abstract}

\bigskip

\section{Introduction}

\label{sec:vix:expansion:intro}

The $\VIX$ volatility index at a given time $T$ is, by definition, 
the implied volatility of a $30$-day log-contract on the $\SP$ index starting at $T$. 
Introduced in $1993$, the $\VIX$ is quoted
by the Chicago Board Options Exchange \cite{exchange2009cboe} and it is computed in practice by static replication 
of the log-contract from market prices of listed vanilla options on the $\SP$. In $2004$, futures contracts on the $\VIX$ started trading, and later
on, in $2006$, options on the $\VIX$ appeared. 
Futures and options on the $\VIX$ can be used as risk-management tools to hedge
the volatility exposure of more complex options portfolios and have become extremely popular volatility derivatives.

An efficient yet parsimonious way of modeling the joint dynamics of an asset  price $S_t$ (here the $\SP$) and its implied volatility is precisely to 
target the implied variances of log-contracts (usually simply referred to as forward variances), which have the appealing feature of being driftless under 
the pricing measure, see \cite{dupire1993model}, \cite{bergomi2004smile}, \cite{buehler2006}.
A well-established practice, in the spirit of forward rate modeling, is to model instantaneous forward variances $(\xi_{t}^{u})_{t \le u}$, defined by 
$\xi_{t}^{u}=\frac{\dd}{\dd u}((u-t)\hat \sigma_t(u)^2)$, where $\hat \sigma_t(u)^2$ denotes the implied variance of a log-contract with maturity $u$, observed at time $t$.

\paragraph{A class of exponential models for instantaneous forward variances.}
An effective class of models, encompassing the celebrated model of Bergomi \cite{bergomi2005smile} and the so-called rough 
Bergomi model \cite{bayer2016pricing}, is obtained assuming that the process $\left(\xi_{t}^{u}\right)_{t \le u}$ solves the following
stochastic differential equation 
\begin{equation}
\xi_{t}^{u}
= \xi_{0}^{u} + \int_0^t \xi_{s}^{u} \, K^{u}\left(s\right)\dd W_{s},
\qquad t \le u\,,
\label{eq:sde:forward:variance}
\end{equation}
where (postponing precise assumptions to section \ref{sec:vix:expansion:}) $K$ is a deterministic $L^2$ kernel and $\xi_{0}^{u}$ a given initial variance curve.
In practice, the curve $u\mapsto\xi_{0}^{u}$ can be computed from observed option prices.
The unique solution to \eqref{eq:sde:forward:variance} is of course the log-normal process $\xi_{t}^{u}
= \xi_{0}^{u} \, e^{\int_0^t K^{u}\left(s\right)\dd W_{s} - \frac 12 \int_0^t K^{u} \left(s\right)^2 \dd s}$.
Equivalently, the instantaneous log-forward variance $X_{t}^{u}:=\ln\left(\xi_{t}^{u}\right)$
solves 
\begin{equation}
\dd X_{t}^{u}=-\frac{1}{2}K^{u}\left(t\right)^{2}\dd t+K^{u}\left(t\right)\dd W_{t}, 
\qquad t \le u\,.
\label{eq:sde:log:forward:variance}
\end{equation}

\emph{Examples of deterministic kernels.} 
The one-factor Bergomi model \cite{bergomi2005smile} corresponds
to an exponential kernel of the form 
\begin{equation} 
K^{u}(t)=\omega \, e^{-k(u-t)},
\qquad\omega>0, \ k \geq 0 \,m
\label{eq:def:one:factor:bergomi}
\end{equation}
where the parameter $k$ corresponds to a mean reversion speed, and $\omega$ to the volatility of forward variances.
It is well known, see again \cite{bergomi2005smile}, that the choice of the exponential kernel leads to a one-dimensional Markovian representation for the forward variance curve:  $\xi_{t}^{u}=\xi_{0}^{u}f^{u}(t,Z_{t})$ where $Z$ is the Ornstein--Uhlenbeck (OU) process $\dd Z_{t}=-kZ_{t}\dd t+\dd W_{t}$
and $f^{u}$ an explicit deterministic function (see section \ref{sec:vix:expansion:numerical:tests}
for details). From a pricing and calibration point of view, such a representation is extremely convenient as it only involves a single Gaussian random variable
($n$-factors extensions of the model \eqref{eq:sde:forward:variance}-\eqref{eq:def:one:factor:bergomi} are of course possible and have been considered in \cite{bergomi2005smile}, leading to low-dimensional Markovian representations of the variance curve in terms of $n$ OU processes).

Introduced recently in \cite{bayer2016pricing}, the rough Bergomi model corresponds to a power kernel of the form 
\begin{equation}
K^{u}(t)=\eta(u-t)^{H-\frac{1}{2}},
\qquad\eta>0, \ H\in(0,1/2) \,.
\label{eq:def:rough:bergomi:model}
\end{equation}
The parameter $H$ encodes the kernel decay, while $\eta$ tunes the volatility of forward variances.
As opposed to \eqref{eq:def:one:factor:bergomi}, the kernel now explodes as $u-t\to0$, and the curve $u\mapsto\xi_{t}^{u}$ 
does not admit a finite-dimensional Markovian representation anymore; on the other side, the model \eqref{eq:sde:forward:variance}-\eqref{eq:def:one:factor:bergomi} is able to provide a parsimonious fit to the term structures of implied volatilities and implied volatility skews observed on the $\SP$ market \cite{alos2007short, fukasawa2011asymptotic, bayer2016pricing}. 

\paragraph{Objectives.}
In the instantaneous forward variance framework above, the VIX index	 at time $T$ is given by
\begin{equation}
\VIX_{T}^{2}
= \frac{1}{\Delta}\int_{T}^{T+\Delta}\xi_{T}^{u} \, \dd u
= \frac{1}{\Delta}\int_{T}^{T+\Delta}e^{X_{T}^{u}} \, \dd u \,.
\label{eq:def:VIX2}
\end{equation}
In this work, we are concerned with the pricing of VIX options $\varphi(\VIX_{T}^{2})$. For example, a call (resp.\ put)
option on the $\VIX$ with strike $\kappa>0$ corresponds to $\varphi(x)=(\sqrt{x}-\kappa)_{+}$
(resp.\ to $\varphi(x)=(\kappa-\sqrt{x})_{+}$), while VIX futures correspond to $\varphi(x)=\sqrt{x}$.
Under \eqref{eq:sde:forward:variance}, the price at time $t=0$ of such an option or futures contract is given by 
\begin{align}
\bE\left[\varphi\left(\VIX_{T}^{2}\right)\right]
=
\bE\left[\varphi\left(\frac{1}{\Delta}\int_{T}^{T+\Delta}\xi_{0}^{u} \, 
e^{-\frac{1}{2}\int_{0}^{T}K^{u}\left(t\right)^{2}\dd t+\int_{0}^{T}K^{u} \left(t\right)\dd W_{t}}
\dd u\right)\right].
\label{eq:price:vix:option}
\end{align}
Despite the simplicity of the forward variance model \eqref{eq:sde:forward:variance}, the option pricing problem \eqref{eq:price:vix:option} is not trivial.
The random variable $\VIX_{T}^{2}$ is given by a continuous sum of correlated log-normal random variables, and as such, its distribution is not explicit and cannot be simulated exactly.
Of course, we see an analogy with Asian option pricing, even if the problem here is structurally different, for the integration in \eqref{eq:def:VIX2}  takes place with respect to the maturity dimension of the forward variance curve, as opposed to the running time variable of a Markov process  as in Asian option payoffs $\varphi(\int_0^T S_t \, \dd t)$.
The expectation \eqref{eq:price:vix:option} can eventually be approximated by coupling a discretization scheme with Monte Carlo simulation, see \cite{jacquier_martini_muguruza_2018, horvath2018volatility} or \cite{bourgey2021multilevel} for the implementation of a multilevel scheme, and  asymptotic formulas for short maturity $T$ can also be derived, as in \cite{alos2022smile, lacombe2021asymptotics}.
Here, we explore an alternative approach based on analytical approximations 
of the form
\be \label{e:formal_expansion}
\bE\left[\varphi\left(\VIX_{T}^{2}\right)\right]={\rm Main\,term}+{\rm Correction\,terms}+\mathrm{Error}.
\ee
Both the main and the correction terms will be easily computable using simple log-normal distributions. In addition, we aim at providing error bounds in terms of the kernel's characteristics and the length of the time-window $\Delta$, covering the case of non-smooth payoffs $\varphi$.

\paragraph{Comparison with the literature on option price expansions.}

\textit{Instantaneous volatility (or first-generation stochastic volatility) models.} Using Fourier-based techniques, semi-analytical formulas are derived in \cite{zhu2012analytical}
for VIX futures in a Heston-type stochastic volatility model with jumps, and in \cite{goutte2017} for option prices in a regime-switching
Heston model. Short-time expansions are provided in \cite{zhao2018} for $\VIX$ futures and options in stochastic volatility models
including Heston, mean-reverting CEV, and $3/2$ models, and in \cite{barletta2019}
under the assumption that $\VIX_{t}=\phi(t,Y_{t})$ for some Markov process $Y$. Assuming multi-scale volatility modeling, some expansions for
$\VIX$ and $\SP$ derivatives are provided in \cite{fouque_saporito2018}, but their accuracy seems to deteriorate for short maturities (below four months), for which $\VIX$ derivatives are the most liquid. In all these works,
the $\VIX$ model is based on a finite-dimensional Markov process, which does not encompass all models of the form \eqref{eq:sde:forward:variance}.

\noindent \textit{Forward variance (or second-generation stochastic volatility) models.}
A recent work close in spirit to ours is \cite{guyon2020vix}, where formal expansions in terms of powers of the volatility-of-volatility parameter are established 
for $\VIX$ futures and power payoffs
in the one-factor and two-factor Bergomi models, along with their extensions to the mixed Bergomi model \cite{bergomi2008smile}.
These expansions are shown to provide accurate approximations for a wide range of model parameters, covering typical values of calibrated parameters in the Equity and FX markets, 
and also beyond -- in particular, even for large values of the volatility-of-volatility parameter.
With respect to \cite{guyon2020vix}, we also deal with the case of $\VIX$ call and put options, 
we cover more general kernel functions, and provide error estimates for our expansions.
Concerning specifically the rough Bergomi model \eqref{eq:def:rough:bergomi:model}, several recent papers have tackled the problem of VIX derivatives pricing in this framework.
In \cite{jacquier_martini_muguruza_2018}, instead of expansion formulas, upper and lower bounds for $\VIX$ futures are provided, and in \cite{lacombe2021asymptotics}, large deviation theory is applied to derive small-maturity asymptotic formulas (covering one- and multi-factor mixed rough Bergomi models).
In \cite{alos2022smile}, the authors specifically focus on the short-maturity at-the-money implied volatility level and skew. 
Exploiting representations from Malliavin calculus, the authors theoretically confirm the capability of  mixed log-normal models to generate a positive VIX skew \cite[section 3.2]{alos2022smile} (as previously announced by the numerical tests in \cite{demarco2018pres} and \cite{guyon2018pres}), and derive short-term asymptotic formulas for at-the-money values.
For instance, in \cite[Example 21]{alos2022smile}  the short-term limit of the at-the-money VIX implied volatility skew is provided for the mixed rough Bergomi model,
which we will consider in section \ref{sec:mixed_model}.


\paragraph{More details on our contributions.}
Our approach
follows a different path with respect to small-maturity or small-parameter asymptotics, consisting in taking advantage of a structural property of the VIX: the relatively short time-window $\Delta$ over which forward variances are integrated in \eqref{eq:def:VIX2} (recall that $\Delta \approx \frac{1}{12}$, when measuring time in years).
The first step is to replace the arithmetic average of exponentials in \eqref{eq:price:vix:option} with their geometric average, 
in the spirit of the work in \cite{kemna1990pricing} for Asian options, which then serves as a central point 
for deriving asymptotic expansions.
The mathematical analysis, although close to \cite{gobet2014weak}
about averaged diffusion processes, is significantly different: first,
we deal with forward curve processes that do not have a Markovian representation; second, we cover the case of payoffs 
$\varphi$ that can fail to be smooth and are only $\frac{1}{2}$-H\"{o}lder;
third, we also deal with model mixtures, that is, mixed Bergomi and rough Bergomi models.
The terms in the resulting expansion \eqref{e:formal_expansion} will be given by a Black--Scholes price along with explicit Black--Scholes Greeks, see Theorem \ref{thm:expansion:plain:model} and Theorem \ref{thm:expansion:mixed:model} for more details.
Thanks to the integration-by-parts formula of Malliavin calculus, we can prove that the error term of our approximation formula is of order $\Delta^{3(d_{1}\wedge\frac{d_{2}}{2})}$
where the constants $d_{1}$ and $d_{2}$ depend on deterministic $L^{p}$ estimates related to the
kernel $K$ (see \eqref{eq:assu:drift} and \eqref{eq:assu:diffusion}).
As main examples, we cover the one-factor standard Bergomi model \eqref{eq:def:one:factor:bergomi} and the rough Bergomi model \eqref{eq:def:rough:bergomi:model} and show that in such cases the error is $\mathcal{O}(\Delta^{3})$, resp.\ $\mathcal{O}(\Delta^{3H})$.
We illustrate these results with several numerical tests on option prices and implied volatilities, showing that the approximation formulas provide very accurate results (relative errors are smaller than $2\%$ in all our tests) for a wide
range of model parameters, see section \ref{sec:vix:expansion:numerical:tests}.
Given the documented inability of exponential models of the form \eqref{eq:sde:forward:variance} to generate realistic market $\VIX$ smiles, we establish a similar expansion formula in the so-called mixed Bergomi and mixed rough Bergomi models, proving that the error term is still $\mathcal{O}(\Delta^{3})$ (resp.\ $\mathcal{O}(\Delta^{3H})$)
for smooth payoffs, see section \ref{sec:mixed_model}.
Finally, some numerical tests on market data confirm that our approximations can be used for fast and efficient calibration of the mixed models to  $\VIX$ smiles (see section \ref{sec:vix:expansion:numerical:tests:mixed:rough:bergomi}).


\paragraph{Notations.}
In most of our explicit formulas and proofs, we find it convenient to factor out the dependence with respect to the initial forward variance curve $u \mapsto \xi_0^u$.
In order to do so while still keeping a compact formulation, we introduce the following probability measures on the interval $\cA:=[T,T+\Delta]$
\begin{equation}
\nu\left(\dd u\right):=\frac{\dd u}{\Delta},
\quad
\nu_{0}\left(\dd u\right):= 
\frac{\xi_{0}^{u}}{ \frac 1 \Delta \int_T^{T+\Delta} \xi_{0}^{u} \dd u} \, \frac{\dd u} \Delta\,.
\label{eq:def:probability:measures}
\end{equation}
We will denote 
\be \label{eq:def:integral:operators}
\begin{aligned}
\nu(f) =  \int_T^{T+\Delta} f(u) \frac{\dd u} \Delta =
\int_\cA f(u) \, \nu(\dd u)
\\
\nu_0(f) =  \int_T^{T+\Delta} f(u) \nu_0(\dd u) = \int_\cA f(u) \, \nu_0(\dd u)
\end{aligned}
\ee
the means of integrable functions $f$ with respect to the measures $\nu$ and $\nu_0$.
Note that, using the notation above, we have 
\[
\bE\left[ \VIX_{T}^{2} \right]
= \frac 1 \Delta \int_T^{T+\Delta} \xi_{0}^{u} \dd u
= \nu(\xi_0^{\cdot})
\]
and yet $\nu_0(\dd u) =  \frac{\xi_0^u}{\bE\left[ \VIX_{T}^{2} \right]} \nu(\dd u)$.
Of course, when the initial forward variance $u \mapsto \xi_0^u$ is constant, we have $\nu=\nu_{0}$.
Finally, we set
\begin{equation}
Y_{T}^{u}
:= X_{T}^{u}-X_{0}^{u}
= -\frac{1}{2}\int_{0}^{T}K^{u}\left(t\right)^{2}\dd t+\int_{0}^{T}K^{u}\left(t\right)\dd W_{t}.\label{eq:def:YT}
\end{equation}

We denote $\Vert \cdot \Vert_p$ the $L^p$ norm for random variables.
In our error estimates and proofs, we will denote $C$ as a generic positive constant that 
may change from line to line and may depend on the model and option parameters, but which is in any case independent of $\Delta$, of the curve $u\mapsto\xi_{0}^{u}$, and the payoff $\varphi$. For two
non-negative real numbers $x$ and $y$, $x\leq_{c}y$ stands for $x \leq Cy$.
We denote $\Phi$ the cumulative distribution function of the standard normal distribution. 

\paragraph{Acknowledgments,}The authors gratefully acknowledge financial support from the research projects Chaire Risques Financiers (\'{E}cole Polytechnique, Fondation du Risque and Soci\'{e}t\'{e} G\'{e}n\'{e}rale) and Chaire Stress Test, Risk Management and Financial Steering (\'{E}cole Polytechnique, Fondation de l'\'{E}cole Polytechnique and BNP Paribas).
We thank Julien Guyon, Martino Grasselli, and Mathieu Rosenbaum for feedback and stimulating discussions on the subject of this article.

\section{Exponential forward variance models\label{sec:vix:expansion:}}

\begin{assumption}
\label{assu:xi0}
The initial instantaneous forward variance curve $u \mapsto\xi_{0}^{u}$ is positive, bounded, and bounded away from zero.
\end{assumption}

\begin{assumption} \label{ass:2}
	The kernel $K^{\cdot}(\cdot)$ in \eqref{eq:sde:forward:variance} is such that $\int_0^T K^u(t)^2 \dd t \mino \infty$  for every $u \in [T, T + \overline \Delta]$, for some $\overline \Delta \ge 1$.
Moreover, for any $p>0$, there exists a positive constant $C_{p}$ such that
	
		\begin{equation}\label{eq:simple:assu}
			\frac 1{\Delta}\int_T^{T+\Delta} 
			e^{p\int_{0}^{T}K^{u}\left(t\right)^{2}\dd t} 
			\, \dd u
			\leq C_{p}
		\end{equation}
for all $\Delta \le \overline \Delta$.
\end{assumption}

Assumption \ref{ass:2} is a mild technical condition, which essentially means that the moments of $\xi_{T}^{u}$ are integrable over $[T, T+\Delta]$.
As a consequence of Assumption \ref{ass:2}, all the moments of the random variable $\VIX_{T}^{2}$ are also finite, by Jensen's inequality.
Using the notation we introduced in \eqref{eq:def:probability:measures} for the measure $\nu_0$, under Assumption \ref{assu:xi0} the condition \eqref{eq:simple:assu} is equivalent to 
\begin{equation}
\int_T^{T+\Delta} e^{p\int_{0}^{T}K^{u}\left(t\right)^{2}\dd t}\, \nu_{0}\left(\dd u\right)\leq C'_{p}
\label{eq:assu:kernel:exp:integral:u}
\end{equation}
for some constant $C'_{p}$ and  all $\Delta \le \overline \Delta$;  we will apply Assumption \ref{ass:2} under the  form \eqref{eq:assu:kernel:exp:integral:u} in our estimates and proofs.
It is easy to check that Assumption \ref{ass:2} is satisfied by the one-factor Bergomi model \eqref{eq:def:one:factor:bergomi} and the rough Bergomi model \eqref{eq:def:rough:bergomi:model}.

\subsection{Proxy for the mean of exponentials\label{subsec:proxy}}

Recalling that $Y_{T}^{u}=X_{T}^{u}-X_{0}^{u}$ from \eqref{eq:def:YT}, the random variable $\VIX_{T}^{2}$ in \eqref{eq:def:VIX2} can be rewritten as 
\begin{equation}
\begin{aligned}
\VIX_{T}^{2} &= 
\frac 1 \Delta \int_T^{T+\Delta} \xi_{0}^{u} \, e^{Y_T^u} \dd u
=  \bE[\VIX_{T}^{2}] 
\int_T^{T+\Delta} \frac{\xi_{0}^{u}}{\bE[\VIX_{T}^{2}]} \, e^{Y_T^u} \, \frac{\dd u} \Delta
\\
&= \bE[\VIX_{T}^{2}]  \, \nu_{0}\bigl(e^{\Ydot T}\bigr) 
= \nu\left(\xidot0\right)\nu_{0}\bigl(e^{\Ydot T}\bigr) \,,
\label{eq:def:true:quantity}
\end{aligned}
\end{equation}
where we have used the definitions  \eqref{eq:def:probability:measures} and \eqref{eq:def:integral:operators} for the measures $\nu, \nu_0$ and their integral means in the last identity.
As addressed in the Introduction, the starting point of our analysis is to approximate the arithmetic mean of exponentials $\nu_{0}\bigl(e^{\Ydot T}\bigr) = \int_T^{T+\Delta} e^{Y^u_T} \nu_0(\dd u)$ with their geometric mean $e^{\nu_{0}\left(\Ydot T\right)}$, which has the appealing property of being log-normal (with explicit mean and variance parameters, given in Proposition \ref{prop:proxies:characteristics} below).
More precisely, we set 
\begin{equation}
\VIX_{T,{\rm P}}^{2}
:= \bE[\VIX_{T}^{2}]  \, 
e^{ \frac 1 \Delta \int_T^{T+\Delta} \frac{\xi_{0}^{u}}{\bE[\VIX_{T}^{2}]} \, e^{Y_T^u} \dd u } 
= \nu\left(\xidot0\right) e^{\nu_{0}\left(\Ydot T\right)},
\label{eq:def:proxy}
\end{equation}
where the subscript ${\rm P}$ stands for proxy.
Log-normal approximations of the VIX random variable in exponential forward variance models  have already been exploited by several authors to derive coarse approximations of VIX 
futures and options prices, as in \cite{bayer2016pricing, jacquier_martini_muguruza_2018, horvath2018volatility}.
In this work, we are precisely going to quantify the difference between option prices on the true VIX \eqref{eq:def:true:quantity} and the corresponding prices computed on the log-normal approximation \eqref{eq:def:proxy}.

In order to work out a representation of the difference $\VIX_{T}^{2} - \VIX_{T,{\rm P}}^{2}$, we introduce the interpolation
\begin{align}
I\left(\ve\right) & :=\nu\left(\xidot0\right)\int_{\cA}e^{\nu_{0}\left(\Ydot T\right)+\ve\left(Y_{T}^{u}-\nu_{0}\left(\Ydot T\right)\right)}\nu_{0}\left(\dd u\right),\quad\ve\in[0,1],\label{eq:def:I:eps}
\end{align}
which is such that $I\left(0\right)=\VIX_{T,{\rm P}}^{2}$ and $I\left(1\right)=\VIX_{T}^{2}$.
Under Assumptions \ref{assu:xi0} and \ref{ass:2}, it is easy to see that the map $\ve\mapsto I\left(\ve\right)$ is smooth almost surely, with $n$th
derivative given by 
\begin{align}
I^{\left(n\right)}\left(\ve\right) & =\nu\left(\xidot0\right)\int_{\cA}\left(Y_{T}^{u}-\nu_{0}\left(\Ydot T\right)\right)^{n}e^{\nu_{0}\left(\Ydot T\right)+\ve\left(Y_{T}^{u}-\nu_{0}\left(\Ydot T\right)\right)}\nu_{0}\left(\dd u\right).\label{eq:I:nth:derivative}
\end{align}
Noticing that $I^{\left(1\right)}\left(0\right)=0$, an application of Taylor's theorem
with integral remainder yields 
\begin{equation}
\begin{aligned}
\VIX_{T}^{2}-\VIX_{T,{\rm P}}^{2}
&= I\left(1\right)-I\left(0\right)
\\
&=\int_{0}^{1}\left(1-\ve\right)I^{\left(2\right)}\left(\ve\right)\dd\ve
=\frac{I^{(2)}\left(0\right)}{2}+\int_{0}^{1}\frac{\left(1-\ve\right)^{2}}{2}I^{(3)}\left(\ve\right)\dd\ve.\label{eq:taylor:I1:I0}
\end{aligned}
\end{equation}
The representations of $ I\left(1\right)-I\left(0\right)$ in the second line of \eqref{eq:taylor:I1:I0} will allow us to quantify the difference
between $\VIX_{T}^{2}$ and $\VIX_{T,{\rm P}}^{2}$, and
to derive our expansions for the expectations of  functions of $\VIX_{T}^{2}$.

As mentioned above, the keystone of our approach is the lognormal
property of $\VIX_{T,{\rm P}}^{2}$. 

\begin{proposition} \label{prop:proxies:characteristics} The proxy
$\VIX_{T,{\rm P}}^{2}$ is lognormal, that is
\[
\ln\left(\VIX_{T,{\rm P}}^{2}\right)\egl\cN\left(\mup,\sigmap^{2}\right),
\]
where the mean and variance parameters are given by
\begin{equation}
\sigmap^{2} := \int_{0}^{T}\nu_{0}\bigl(\Kdot\left(t\right)\bigr)^{2}\dd t \,,
\quad
\mup := 
\ln(\nu(\xidot0)) -\frac{1}{2}\int_{0}^{T}\nu_{0}\bigl(\Kdot\left(t\right)^{2}\bigr)\dd t \,.
\label{eq:mean:variance:proxy}
\end{equation}
\end{proposition} 
\begin{proof}
According to \eqref{eq:def:proxy}, $\ln\left(\VIX_{T,{\rm P}}^{2}\right)-\ln\left(\nu\left(\xidot0\right)\right)$ is equal to
\[
-\frac{1}{2}\int_{\cA}\biggl(\int_{0}^{T}K^{u}\left(t\right)^{2}\dd t\biggr)\nu_{0}(\dd u)+\int_{\cA}\biggl(\int_{0}^{T}K^{u}\left(t\right)\dd W_{t}\biggr)\nu_{0}(\dd u),
\]
and we conclude using stochastic Fubini's theorem (see e.g.\ \cite[Lemma 1.1]{gobet2014weak}). 
\end{proof}

Incidentally, in light of Proposition 3, we note that $\bE[\VIX_{T,{\rm P}}^{2}] \neq \bE[ \VIX_{T}^{2}]$.

\subsection{Strong error estimates between $\VIX_{T}^{2}$ and its proxy $\VIX_{T,{\rm P}}^{2}$}

To estimate the $L^{p}$ norm of the difference $\VIX_{T}^{2}-\VIX_{T,{\rm P}}^{2}$, we need some estimates for the deterministic $L^{p}$ norm (over $\mathcal A = [T, T+\Delta]$) of the difference between the diffusion coefficient $K^u(t)$ (resp.\ the drift coefficient $(K^u(t))^2$) of the log-forward variance $\log (\xi^u_t)$ and its integral average $\nu_0(K^\cdot(t))$ (resp.\ $\nu_0(K^\cdot(t)^2)$).
We require that these deterministic $L^p$ norms go to zero as $\Delta$ goes to zero with certain rates $d_1$ and $d_2$, see conditions \eqref{eq:assu:drift} and \eqref{eq:assu:diffusion} in Assumption \ref{assu:kernel}.
In specific examples -- the Bergomi model \eqref{eq:def:one:factor:bergomi} and the rough Bergomi model \eqref{eq:def:rough:bergomi:model} -- we are able to check this assumption and to estimate the corresponding rates $d_1$ and $d_2$, see section \ref{sec:flat:forward:variance}.

We further require some conditions such that the proxy random variable is not degenerate, in the sense of \eqref{eq:assu:limit:mean:proxy} and \eqref{eq:assu:limit:variance:proxy}.

\begin{assumption}\label{assu:kernel} For any $p>0$, there exist
positive constants $d_{1},d_{2},C$ such that 
\begin{align}
\Gamma_{\Delta,T,p}:=\Bigl(\int_{\cA}\left|\int_{0}^{T}\left[K^{u}\left(t\right)^{2}-\nu_{0}\left(K^{\cdot}\left(t\right)^{2}\right)\right]\dd t\right|^{p}\nu_{0}\left(\dd u\right)\Bigr)^{\frac{1}{p}} & \leq C\Delta^{d_{1}},\label{eq:assu:drift}
\\
\Lambda_{\Delta,T,p}:=\Bigl(\int_{\cA}\left|\int_{0}^{T}\left[K^{u}\left(t\right)-\nu_{0}\left(K^{\cdot}\left(t\right)\right)\right]^{2}\dd t\right|^{p}\nu_{0}\left(\dd u\right)\Bigr)^{\frac{1}{p}} & \leq C\Delta^{d_{2}},\label{eq:assu:diffusion}
\\
\sup_{\Delta} |\mup| & \leq C,
\label{eq:assu:limit:mean:proxy}
\\
\frac{1}{C}\leq\inf_{\Delta}\sigmap\leq\sup_{\Delta}\sigmap & \leq C.\label{eq:assu:limit:variance:proxy}
\end{align}
\end{assumption}

\begin{proposition} \label{prop:estimate:YT:meanYT}
Under Assumptions  \ref{assu:xi0}, \ref{ass:2}, and \ref{assu:kernel}, for any $p,q\geq1$ we have 
\begin{equation}
\int_{\cA}\left\Vert Y_{T}^{u}-\nu_{0}\left(\Ydot T\right)\right\Vert _{p}^{q}\nu_{0}\left(\dd u\right)\leq_{c}\Delta^{(d_{1}\wedge\frac{d_{2}}{2})q}.\label{eq:Lpq:Y}
\end{equation}
\end{proposition}

\begin{theorem} \label{thm:estimate:nue:nueproxy}
Under Assumptions \ref{assu:xi0}, \ref{ass:2}, and \ref{assu:kernel}, for any $p\geq1$ and $n\in\bN$ we have
\begin{equation}
\sup_{\ve\in[0,1]}\left\Vert I^{\left(n\right)}\left(\ve\right)\right\Vert _{p}\leq_{c}\nu\left(\xidot0\right)\Delta^{(d_{1}\wedge\frac{d_{2}}{2})n} \,.
\label{eq:Lp:I:nth:derivative}
\end{equation}
Consequently, in view of \eqref{eq:taylor:I1:I0},
\begin{equation}
\left\Vert
\VIX_{T}^{2}-\VIX_{T,{\rm P}}^{2}
\right\Vert _{p}\leq_{c}\nu\left(\xidot0\right)\Delta^{2d_{1}\wedge d_{2}} \,.
\label{eq:Lp:exp:exp:proxy}
\end{equation}
\end{theorem}

\subsection{General price expansion\label{subsec:price:expansion}}

In light of our discussion in section \ref{subsec:proxy}, the leading order term in the approximation of the VIX option price $\bE[ \varphi(\VIX_{T}^{2})]$ will be given by the price on the proxy $\bE\left[\varphi\left(\VIX_{T,{\rm P}}^{2}\right)\right]$.
Assuming for a moment that the payoff function $\varphi$ is smooth, a Taylor expansion around the point $\VIX_{T,{\rm P}}^{2}$
yields
\be \label{e:Taylor_smooth_varphi}
\bE\bigl[\varphi\left(\VIX_{T}^{2}\right)\bigr]
=
\bE\bigl[\varphi\left(\VIX_{T,{\rm P}}^{2}\right)\bigr]
+ \bE\Bigl[\varphi^{\prime}\left(\VIX_{T,{\rm P}}^{2}\right)\left(\VIX_{T}^{2}-\VIX_{T,{\rm P}}^{2}\right)\Bigr] 
+ E_0
\ee
where the remainder $E_0 = \int_{0}^{1}\left(1-\lambda\right)\bE\bigl[\varphi^{\prime \prime}\left(\lambda\VIX_{T}^{2}+\left(1-\lambda\right)\VIX_{T,{\rm P}}^{2}\right)\left(\VIX_{T}^{2}-\VIX_{T,{\rm P}}^{2}\right)^{2}\bigr]\dd\lambda$ will be treated as an error term.
The difference $\VIX_{T}^{2}-\VIX_{T,{\rm P}}^{2}$ can be expanded using \eqref{eq:taylor:I1:I0}:
recalling from \eqref{eq:I:nth:derivative} the expression of the derivative $I^{(2)}$, we get
\[
\begin{aligned}
\VIX_{T}^{2}-\VIX_{T,{\rm P}}^{2}
&=
\frac 12 I^{(2)}\left(0\right)
+\int_{0}^{1}\frac{\left(1-\ve\right)^{2}}{2}I^{(3)}\left(\ve\right)\dd\ve
\\
&= \frac{1}{2} \VIX_{T,{\rm P}}^{2} \int_{\cA}\left(Y_{T}^{u}-\nu_{0}\left(\Ydot T\right)\right)^{2}\nu_{0}\left(\dd u\right)
+ \int_{0}^{1}\frac{\left(1-\ve\right)^{2}}{2}I^{(3)}\left(\ve\right)\dd\ve \,,
\end{aligned}
\]
so that the second expectation on the right-hand side of \eqref{e:Taylor_smooth_varphi} can eventually be written as
\be \label{e:first_term_representation}
\bE\Bigl[\varphi^{\prime}\left(\VIX_{T,{\rm P}}^{2}\right)\left(\VIX_{T}^{2}-\VIX_{T,{\rm P}}^{2}\right)\Bigr] 
= 
\bE\Bigl[\varphi^{\prime}\left(\VIX_{T,{\rm P}}^{2}\right)
\frac{1}{2} \VIX_{T,{\rm P}}^{2} \int_{\cA}\left(Y_{T}^{u}-\nu_{0}\left(\Ydot T\right)\right)^{2}\nu_{0}\left(\dd u\right) \Bigr]
+ E_1 
\ee
where $E_1 = \bE\bigl[\varphi^{\prime}\left(\VIX_{T,{\rm P}}^{2}\right)  \int_{0}^{1}\frac{\left(1-\ve\right)^{2}}{2}I^{(3)}\left(\ve\right)\dd\ve \bigr] $ will be treated as a second error term.
The random variable multiplying $\varphi^{\prime}\left(\VIX_{T,{\rm P}}^{2}\right)$ inside \eqref{e:first_term_representation} can be interpreted as the random weight appearing after the application of an integration-by-parts formula to higher-order derivatives of $\varphi$ -- which means that, in its turn, the expectation on the right-hand side of \eqref{e:first_term_representation} can be rewritten in terms of a combination of derivatives of the form $\partial^i_{\ve} \, \bE\bigl[\varphi^{\prime}\left(\VIX_{T,{\rm P}}^{2} \, e^{\ve} \right) \bigr]|_{\ve = 0}$.
The important property of such higher-order derivatives of the expectation $\bE\bigl[\varphi \left(\VIX_{T,{\rm P}}^{2} \right) \bigr]$ is to be explicit -- they are Black-Scholes Greeks.
The final expression of the expansion \eqref{e:Taylor_smooth_varphi} will therefore contain a combination of a Black--Scholes price and some of its partial derivatives.

The details of the approach sketched above will be given in section \ref{sec:proof_main_thm}; 
here we state the final expression we obtain for the expansion \eqref{e:Taylor_smooth_varphi} after the integration-by-parts procedure, 
see Theorem \ref{thm:expansion:plain:model} below, which is our main result in this section.

We will make use of the following coefficients
$\left(\gamma_{i}\right)_{i=1,2,3}$ :
\be \label{eq:def:gamma:coeffs}
\begin{aligned}
\gamma_{1} & :=\frac{1}{8}\int_{\cA}\Bigl(\int_{0}^{T}\left[K^{u}\left(t\right)^{2}-\nu_{0}(\Kdot\left(t\right)^{2})\right]\dd t\Bigr)^{2}\nu_{0}\left(\dd u\right)+\frac{1}{2}\int_{\cA}\Bigl(\int_{0}^{T}\left[K^{u}\left(t\right)-\nu_{0}\left(\Kdot\left(t\right)\right)\right]^{2}\dd t\Bigr)\nu_{0}\left(\dd u\right),
\\
\gamma_{2} & :=-\frac{1}{2}\int_{\cA}\Bigl(\int_{0}^{T}\nu_{0}\left(\Kdot\left(t\right)\right)\left[K^{u}\left(t\right)-\nu_{0}\left(\Kdot\left(t\right)\right)\right]\dd t\Bigr)\Bigl(\int_{0}^{T}\left[K^{u}\left(t\right)^{2}-\nu_{0}(\Kdot\left(t\right)^{2})\right]\dd t\Bigr)\nu_{0}\left(\dd u\right),
\\
\gamma_{3} & :=\frac{1}{2}\int_{\cA}\Bigl(\int_{0}^{T}\nu_{0}\left(\Kdot\left(t\right)\right)\left[K^{u}\left(t\right)-\nu_{0}\left(\Kdot\left(t\right)\right)\right]\dd t\Bigr)^{2}\nu_{0}\left(\dd u\right).
\end{aligned}
\ee

\begin{theorem}[Option price approximation] \label{thm:expansion:plain:model}
Let Assumptions \ref{assu:xi0}, \ref{ass:2}, and \ref{assu:kernel} be in force, and let	$\varphi: \R \to \R$ be a $\theta$--H\"{o}lder continuous function for some $\theta\in(0,1]$.
The price
of the $\VIX$ option with payoff $\varphi(\VIX_{T}^{2})$ is given by
\begin{equation}
\bE\left[\varphi\left(\VIX_{T}^{2}\right)\right]
=
\bE\left[\varphi\left(\VIX_{T,{\rm P}}^{2}\right)\right]
+ \sum_{i=1}^{3}\gamma_{i} \left. \partial_{\ve}^{i} \, \bE\left[\varphi\left(\VIX_{T,{\rm P}}^{2} \, e^{\ve}\right)\right]\right|_{\ve=0}
+ {\mathscr{E}_{\varphi}},
\label{eq:expansion:plain:model}
\end{equation}
where $\mathscr{E}_{\varphi}$ is an error term such that $\left|\mathscr{E}_{\varphi}\right|\leq_{c}\Delta^{3(d_{1}\wedge\frac{d_{2}}{2})}$.
\end{theorem}

\begin{remark}
VIX futures correspond to $\varphi(x) = \sqrt{x}$, VIX put options to $\varphi(x) = (\kappa - \sqrt{x})^+$ and call options to to $\varphi(x) = (\sqrt{x} - \kappa)^+$.
In all these cases, the function $\varphi$ is $\frac 12$--H\"{o}lder ($\varphi$ is Lipschitz in the case of put options with strictly positive strike  $\kappa$).
\end{remark}

\begin{remark} Although the payoff $\varphi$ may fail to be smooth, condition \eqref{eq:assu:limit:variance:proxy} ensures that the lognormal proxy $\VIX_{T,{\rm P}}^{2}$ is not degenerate, with the effect of regularizing the map $\ve\mapsto\bE\left[\varphi\left(\VIX_{T,{\rm P}}^{2}e^{\ve}\right)\right]$, so that the derivatives $\partial_{\ve}^{i} \, \bE\left[\varphi\left(\VIX_{T,{\rm P}}^{2}e^{\ve}\right)\right]$ are well-defined. 
\end{remark}

\begin{remark} \label{rem:coeffs_gamma}
The adimensional coefficients $(\gamma_{i})_{i\in\{1,2,3\}}$ are defined by deterministic integrals with respect to time variables. They depend on the option maturity $T$, on the time window $\Delta$, and on the model parameters $\xi_0^{\cdot}$ and $K$, but not on the option payoff -- which means that, in the case of call and put options, they can be evaluated once for all strikes.
When the initial forward variance curve $u \mapsto \xi_0^u$ is constant over the VIX time window $(T,T+\Delta)$ (which is a standard choice that can usually be made in practice), the $\gamma_i$'s have analytical closed-form expressions in the Bergomi model \eqref{eq:def:one:factor:bergomi}, see Proposition \ref{prop:charact:onefactor:flat}.
In the rough Bergomi model \eqref{eq:def:rough:bergomi:model}, the $\gamma_i$'s  do not seem to admit a closed-form representation even when the initial variance curve $\xi_0^u$ is constant, but in this case, their dependence with respect to the remaining model parameters $\eta$ and $H$ can be simplified, see Remark \ref{rem:gamma_i_rough_Bergomi}  for more details.
For the general case of a non-constant initial curve $\xi_{0}^{u}$, one has to appeal to deterministic quadrature to approximate the integrals with respect to the measure $\nu_{0}(\dd u) = \frac{\xi_{0}^{u}}{ \frac 1 \Delta \int_T^{T+\Delta} \xi_{0}^{u} \dd u} \, \frac{\dd u} \Delta$.
\end{remark}

\begin{remark}(Limiting case: constant kernel)
When the kernel $K$ is constant, which corresponds to $H = \frac{1}{2}$ in the rough Bergomi model and to $k=0$ in the Bergomi model, 
we have $\VIX_T^2 = \VIX_{T, \rm P}^2$. Correspondingly, in this case
$\gamma_i=0$ for every $i \in \{1,2,3\}$, and \eqref{eq:expansion:plain:model} holds  with zero error term $\mathscr{E}_{\varphi} = 0$.
\end{remark}

As a direct corollary of Theorem \ref{thm:expansion:plain:model}, we obtain expansion formulas for the price of calls, puts, and futures on $\VIX_{T}$, for which $\varphi$
is $\frac{1}{2}$-H\"{o}lder continuous. 
The resulting expression will contain a combination
of the following Black--Scholes prices and Greeks :
\begin{align*}
C_{\text{BS}}(x,y,\sigma) & :=\bE_{Z\sim\cN(0,1)}\Bigl[(xe^{-\frac{\sigma^{2}}{2}+\sigma Z}-y)_{+}\Bigr]=x\Phi\Bigl(\frac{1}{\sigma}\ln\Bigl(\frac{x}{y}\Bigr)+\frac{\sigma}{2}\Bigr)-y\Phi\Bigl(\frac{1}{\sigma}\ln\Bigl(\frac{x}{y}\Bigr)-\frac{\sigma}{2}\Bigr),\\
\Deltacall_{\BS}(x,y,\sigma) & :=\partial_{x}C_{\text{BS}}(x,y,\sigma)=\Phi\Bigl(\frac{1}{\sigma}\ln\left(\frac{x}{y}\right)+\frac{\sigma}{2}\Bigr),\\
\Gamma_{{\rm BS}}\left(x,y,\sigma\right) & :=\partial_{x}^{2}C_{\text{BS}}\left(x,y,\sigma\right)=\frac{\Phi^{\prime}\Bigl(\frac{1}{\sigma}\ln\left(\frac{x}{y}\right)+\frac{\sigma}{2}\Bigr)}{x\sigma},
\\
{\rm Speed}_{{\rm BS}}\left(x,y,\sigma\right) & :=\partial_{x}^{3}C_{\text{BS}}\left(x,y,\sigma\right)=-\frac{\Gamma_{{\rm BS}}\left(x,y,\sigma\right)}{x}\Bigl(\frac{1}{\sigma^{2}}\ln\left(\frac{x}{y}\right)+\frac{3}{2}\Bigr),
\end{align*}
where $x,y,\sigma>0$.

We note that the resulting expansion formulas have the appealing property of satisfying put-call parity (see Corollary \ref{cor:plain:model:call:put:future} below); as a consequence, implied volatilities computed either from the expansion \eqref{eq:expansion:plain:model} for call options or from the same expansion for put options will coincide (when using a Black-Scholes formula with forward value equal to the VIX futures price given by the same expansion \eqref{eq:expansion:plain:model}).

\begin{corollary} \label{cor:plain:model:call:put:future} 
Let Assumptions \ref{assu:xi0}, \ref{ass:2} and \ref{assu:kernel} be in force.
For every $i\in\{0,1,2,3\}$, let
$P_{i}^{{\rm {O}}}:=\left.\partial_{\ve}^{i}\bE\left[\varphi\left(\VIX_{T,{\rm P}}^{2}e^{\ve}\right)\right]\right|_{\ve=0}$
where ${\rm {O=call}}$ when $\varphi(x)=(\sqrt{x}-\kappa)_{+},$
${\rm {O=F}}$ when $\varphi(x)=\sqrt{x},$ and ${\rm {O=put}}$ when
$\varphi(x)=(\kappa-\sqrt{x})_{+}$, for some given strike $\kappa>0$.
 Then, the expansion \eqref{eq:expansion:plain:model}
holds for\\
 $\rhd$ VIX call options with 
\begin{align*}
\Pcall_{0} & =C_{\BS}\bigl(e^{\frac{\mup}{2}+\frac{\sigmap^{2}}{8}},\kappa,\frac{\sigmap}{2}\bigr),\quad\Pcall_{1}=\frac{1}{2}e^{\frac{\mup}{2}+\frac{\sigmap^{2}}{8}}\Deltacall_{\BS}\bigl(e^{\frac{\mup}{2}+\frac{\sigmap^{2}}{8}},\kappa,\frac{\sigmap}{2}\bigr),\\
\Pcall_{2} & =\frac{\Pcall_{1}}{2}+\frac{e^{\mup+\frac{\sigmap^{2}}{4}}}{4}\Gamma_{\BS}\bigl(e^{\frac{\mup}{2}+\frac{\sigmap^{2}}{8}},\kappa,\frac{\sigmap}{2}\bigr),\\
\Pcall_{3} & =-\frac{\Pcall_{1}}{2}+\frac{3\Pcall_{2}}{2}+\frac{1}{8}e^{\frac{3\mup}{2}+\frac{3\sigmap^{2}}{8}}{\rm Speed}_{\BS}\bigl(e^{\frac{\mup}{2}+\frac{\sigmap^{2}}{8}},\kappa,\frac{\sigmap}{2}\bigr).
\end{align*}
$\rhd$ VIX futures with 
\[
\PF_{i}=2^{-i}e^{\frac{\mup}{2}+\frac{\sigmap^{2}}{8}}\text{ for }i\in\{0,1,2,3\},
\]
$\rhd$ VIX put options with 
\[
\Pput_{0}:=\Pcall_{0}-\PF_{0}+\kappa,\quad\Pput_{i}=\Pcall_{i}-\PF_{i}\text{ for }i\in\{1,2,3\}.
\]
In particular, 
note that put--call parity holds for the truncated expansion \eqref{eq:expansion:plain:model}, that is
\begin{equation}
\Pcall_{0}-\Pput_{0}
+ \sum_{i=1}^{3}\gamma_{i}\bigl(\Pcall_{i}-\Pput_{i}\bigr)
=
\PF_{0} 
+ \sum_{i=1}^{3} \gamma_{i} \PF_{i}
-\kappa.\label{eq:put:call:parity:expansion}
\end{equation}
\end{corollary}

\subsection{Verification of  Assumption \ref{assu:kernel} and evaluation of the expansion in our main examples} \label{sec:flat:forward:variance}

In this section, we assume that
the initial instantaneous forward variance curve $u \mapsto \xi_{0}^{u}$
is constant over the interval $(T, T+\Delta)$.
In this case, $\nu_{0}(\dd u) = \nu(\dd u) = \frac{\dd u}\Delta$.

\subsubsection{Bergomi model}

We consider the one-factor Bergomi model \eqref{eq:def:one:factor:bergomi}, where the convolution kernel is $K^{u}(t)=\omega e^{-k(u-t)}$ with $\omega>0$ and $k \geq 0$. 
In Proposition \ref{prop:charact:onefactor:flat}, we
establish explicit formulas for the proxy's mean and variance \eqref{eq:mean:variance:proxy},
and work out the exact asymptotics of the $L^{p}$ norms \eqref{eq:assu:drift}
and \eqref{eq:assu:diffusion} as $\Delta \to 0$.
This will allow us to see that Assumption \ref{assu:kernel} is verified for the Bergomi model.
Moreover, we provide closed-form expressions for the coefficients
$\gamma_{i}$ appearing in the expansion \eqref{eq:expansion:plain:model}.

\begin{proposition} \label{prop:charact:onefactor:flat}
Consider the one-factor
Bergomi model \eqref{eq:def:one:factor:bergomi} and recall the coefficients $\Gamma_{\Delta,T,p}$ and $\Lambda_{\Delta,T,p}$ defined in Assumption \ref{assu:kernel}.
If $k = 0$, then $\Gamma_{\Delta,T,p}  = \Lambda_{\Delta,T,p} \equiv 0$.
Otherwise if $k > 0$, for every $p>0$ we have
\begin{align}
\Gamma_{\Delta,T,p} & \underset{\Delta\to0}{\sim}\frac{\omega^{2}(1-e^{-2kT})}{2(1+p)^{\frac{1}{p}}}\Delta,
\qquad
\Lambda_{\Delta,T,p}\underset{\Delta\to0}{\sim}\frac{\omega^{2}k(1-e^{-2kT})}{8(1+2p)^{\frac{1}{p}}}\Delta^{2}.\label{eq:expansion:asymptotic:bergomi}
\end{align}
The mean and variance \eqref{eq:mean:variance:proxy} of the proxy $\VIX_{T,{\rm P}}^{2}$ are given by
\[
\mup=X_{0}-\frac{\omega^{2}}{8k^{2}\Delta}\left(1-e^{-2kT}\right)\left(1-e^{-2k\Delta}\right),
\qquad
\sigmap^{2}=\frac{\omega^{2}}{2k^{3}\Delta^{2}}\left(1-e^{-2kT}\right)\left(1-e^{-k\Delta}\right)^{2} \,.
\]
Furthermore, the coefficients
$\left(\gamma_{i}\right)_{i\in\{1,2,3\}}$  defined in \eqref{eq:def:gamma:coeffs} have the following closed-form expressions
\begin{align*}
\gamma_{1} & =\frac{\omega^{4}}{128k^{4}\Delta^{2}}\Big(-1+k\Delta\frac{1+e^{-2k\Delta}}{1-e^{-2k\Delta}}\Big)\left(1-e^{-2kT}\right)^{2}\left(1-e^{-2k\Delta}\right)^{2}\\
 & \qquad\qquad\qquad+\frac{\omega^{2}}{8k^{3}\Delta^{2}}\left((2+k\Delta)e^{-k\Delta}-2+k\Delta\right)\left(1-e^{-2kT}\right)\left(1-e^{-k\Delta}\right),\\
\gamma_{2} & =-\frac{\omega^{4}}{48\Delta^{3}k^{5}}\left(1-e^{-k\Delta}\right)^{2}\left(2k\Delta e^{-k\Delta}+2k\Delta+e^{-2k\Delta}(2k\Delta+3)-3\right)\left(1-e^{-2kT}\right)^{2},\\
\gamma_{3} & =\frac{\omega^{4}}{16k^{6}\Delta^{4}}\left(1-e^{-2kT}\right)^{2}\left(1-e^{-k\Delta}\right)^{3}\left(k\Delta-2+(2+k\Delta)e^{-k\Delta}\right).
\end{align*}
\end{proposition}

In light of \eqref{eq:expansion:asymptotic:bergomi} 
and of the expressions of $\mup$ and $\sigmap$ in
Proposition \ref{prop:charact:onefactor:flat}, we have the following result.

\begin{corollary}\label{cor:check:onefactor:bergomi} 
In the Bergomi model \eqref{eq:def:one:factor:bergomi} with parameter $k >0$, Assumption \ref{assu:kernel} holds with $d_{1}=1$ and $d_{2}=2$. Consequently, Theorem \ref{thm:expansion:plain:model}
holds and the error term ${\mathscr{E}_{\varphi}}$ in \eqref{eq:expansion:plain:model} is $\cO(\Delta^{3}).$
\end{corollary}

\subsubsection{The rough Bergomi model}

We now consider the rough Bergomi model \eqref{eq:def:rough:bergomi:model}. Recall that in this case, we have the fractional kernel $K^{u}(t)=\eta(u-t)^{H-\frac{1}{2}}$, with $\eta>0$ and $H\in(0,1)$. 

\begin{proposition} \label{prop:charact:rough:flat}
Consider the rough
Bergomi model \eqref{eq:def:rough:bergomi:model} and recall the coefficients $\Gamma_{\Delta,T,p}$ and $\Lambda_{\Delta,T,p}$ defined in Assumption \ref{assu:kernel}.
If $H = \frac{1}{2}$, then $\Gamma_{\Delta,T,p}  = \Lambda_{\Delta,T,p} \equiv 0$.
Otherwise if $H\in(0,1)\setminus\{\frac{1}{2}\}$, for every $p>0$ we have
\[
\Gamma_{\Delta,T,p}\underset{\Delta\to0}{\sim}\eta^{2}\begin{cases}
\frac{1}{2H}\left(\int_{0}^{1}\left|\frac{1}{2H+1}-y^{2H}\right|^{p}\dd y\right)^{\frac{1}{p}}\Delta^{2H} & \text{if }H\in\big(0,\frac{1}{2}\big),\\
\frac{T^{2H-1}}{2\left(1+p\right)^{\frac{1}{p}}}\Delta & \text{if }H\in\big(\frac{1}{2},1\big),
\end{cases}
\]
and
\[
\Lambda_{\Delta,T,p}\underset{\Delta\to0}{\sim}\eta^{2}f_{{\rm diff}}\left(H,p\right)\Delta^{2H},
\]
where $f_{{\rm diff}}\left(H,p\right):=\frac{1}{\left(H+\frac{1}{2}\right)^{2}}\Bigl(\int_{u=0}^{1}\Bigl|\int_{s=0}^{\infty}\left((1+s)^{H+\frac{1}{2}}-s^{H+\frac{1}{2}}-\left(H+\frac{1}{2}\right)(u+s)^{H-\frac{1}{2}}\right)^{2}\dd s\Bigr|^{p}\dd u\Bigr)^{\frac{1}{p}}.$
The proxy's mean and variance \eqref{eq:mean:variance:proxy} are given by
\begin{align*}
\mup & =X_{0}-\frac{\eta^{2}\bigl[(T+\Delta)^{2H+1}-\Delta^{2H+1}-T^{2H+1}\bigr]}{4H(2H+1)\Delta}\xrightarrow[\Delta\to0]{}X_{0}-\frac{\eta^{2}T^{2H}}{4H}\in\bR,
\\
\sigmap^{2} & =\frac{\eta^{2}}{\Delta^{2}\left(H+\frac{1}{2}\right)^{2}} \biggl[\frac{\left(T+\Delta\right)^{2H+2}+T^{2H+2}-\Delta^{2H+2}}{2H+2}-2\int_{0}^{T}\left(t+\Delta\right)^{H+\frac{1}{2}}t^{H+\frac{1}{2}}\dd t\biggr]
\\
 & \xrightarrow[\Delta\to0]{}\frac{\eta^{2}T^{2H}}{2H}\in(0,\infty) \,.
\end{align*}
Consequently, estimates \eqref{eq:assu:limit:mean:proxy} and \eqref{eq:assu:limit:variance:proxy}
hold for the rough Bergomi model.
\end{proposition}

As a direct corollary of Proposition \ref{prop:charact:rough:flat}, we have the following result.

\begin{corollary}\label{cor:check:rough:bergomi} Assume $H \in (0,1) \setminus \{\frac 12\}$. In the rough Bergomi
model, Assumption \ref{assu:kernel} holds with $d_{1}=1\wedge2H$
and $d_{2}=2H$. Consequently, Theorem \ref{thm:expansion:plain:model}
holds and the error term ${\mathscr{E}_{\varphi}}$ in the expansion \eqref{eq:expansion:plain:model} is $\cO(\Delta^{3H}).$
\end{corollary}

\begin{rem}
The proxy's variance can also be rewritten as 
\begin{multline}
\sigmap^{2}=\frac{\eta^{2}}{\Delta^{2}(H+\frac{1}{2})^{2}}\biggl[\frac{(T+\Delta)^{2H+2}+T^{2H+2}-\Delta^{2H+2}}{2H+2}\\
-\frac{4}{2H+3}\Delta^{H+\frac{1}{2}}T^{H+\frac{3}{2}}{}_{2}F_{1}\Bigl(-H-\frac{1}{2},H+\frac{3}{2};H+\frac{5}{2};-\frac{T}{\Delta}\Bigr)\biggr],\label{eq:var:proxy:hyper}
\end{multline}
using the identity \cite[3.197.8]{gradshteyn2014table}, where
$_{2}F_{1}$ is the Gaussian hypergeometric function (see \cite[Chapter 15]{olver2010nist}).
\end{rem}

\begin{rem} \label{rem:gamma_i_rough_Bergomi}
In the rough Bergomi model, the coefficients $(\gamma_{i})_{i\in\{1,2,3\}}$ defined in \eqref{eq:def:gamma:coeffs} do not seem to admit a closed-form representation as in the standard Bergomi model.
Nevertheless, note that their dependence with respect to the volatility-of-variance parameter $\eta$ is particularly simple: we have  
$\gamma_1 = \eta^{4} \int_T^{T+\Delta} g_H(u,t) \frac{\dd u}{\Delta}  + \eta^2 \int_T^{T+\Delta} \int_0^T f_{1,H}(u,t) \dd t \, \frac {\dd u}{\Delta}$,
$\gamma_2 = \eta^{4} \int_T^{T+\Delta} g_H(u,t) \int_0^T f_{2,H}(u,t) \dd t \, \frac {\dd u}{\Delta}$,
$\gamma_3 = \eta^{4} \int_T^{T+\Delta}(\int_0^T f_{2,H}(u,t) \dd t)^2 \frac {\dd u}{\Delta}$,
where $f_{1,H}$, $f_{2,H}$ and $g_H$ are explicit functions
(see \cite[Chapter 6]{bourgey2020stochastic} for detailed expressions).
For a given VIX maturity $T$, the integrals of the functions $f_{1,H}$, $f_{2,H}$ and $g_H$ can be evaluated and tabulated once for several values of $H$ over a grid in the interval $(0,1)$, and then simply looked up within a pricing or calibration procedure.
The dependence with respect to the parameter $\eta$ is explicit.
Recall, as pointed out in Remark \ref{rem:coeffs_gamma},  that the $(\gamma_{i})_{i\in\{1,2,3\}}$ do
not depend on the option payoff $\varphi$, and can therefore be evaluated beforehand for all call and put strikes.
\end{rem}

%

\subsection{Numerical tests of option price formulas}
\label{sec:vix:expansion:numerical:tests}

In this section, we test our approximation formulas for VIX futures, call, and put options. First of all, let us explain how we compute our reference  prices.

\begin{remark}(Computation of the reference prices)\label{rem:benchmark:prices}
For the standard Bergomi model, recall we have the Markovian representation of forward variances
	\[
	\VIX_{T}^{2} = \frac{1}{\Delta}\int_{T}^{T+\Delta} \xi_{0}^{u} \, f^u \left(T,X_{T}\right)\dd u, 
	\]
where 
\[
	f^{u}(T,x)=\exp\bigl(\omega e^{-k(u-T)}x-\frac{1}{2}\omega^{2}  e^{-2k(u-T)} \Var(X_T) \bigr)
\]
and $(X_{t})_{t\geq 0}$ is the Ornstein--Uhlenbeck process $X_{t} = -k \int_0^t X_{s}\dd s + W_{s}$, with $\Var(X_T) =  \frac{1 - e^{-2kT}}{2k} \1_{k>0} + T \1_{k=0}$.
Consequently, in order to price an option on $\VIX_T$ in the one-factor Bergomi model, we can (and do) rely on a two-dimensional deterministic quadrature: we couple a Gauss--Legendre scheme for the integration with respect to the time parameter $u$, and a Gauss--Hermite scheme for the space dimension (see \cite{bergomi2005smile, guyon2020vix} for more details on the implementation).

For the rough Bergomi model, the Markovian representation is not available anymore and  we rely on a discretization of the variance curve process $u \mapsto \xi_T^u$. We approximate the integral $\VIX_{T}^{2}=\frac{1}{\Delta}\int_{T}^{T+\Delta}\xi_{T}^{u}\dd u$ using a rectangle scheme $\VIX_{T}^{2,n} = \frac 1 n \sum_{i=1}^{n} \xi_T^{u_i}$ over a regular grid $(u_i)_{i=1,\dots,n}$.
The random vector $(\xi_T^{u_i})_i$ can be simulated exactly,
so that we eventually approximate VIX futures and options prices via their  empirical means over i.i.d.\ samples of the discretized variable $\VIX_{T}^{2,n}$. 
The related weak error is known to behave as $O(\frac 1 n)$ (independently of the value of $H$), see \cite{horvath2018volatility} and \cite{bourgey2021multilevel}.
We can considerably reduce the variance of the estimators using an efficient control variate that is nothing but the discretized version of our lognormal proxy, see	again \cite{bourgey2021multilevel} for details.
\end{remark}

We consider three different VIX maturities $T \in \{1,3,6 \mbox{ months}\}$, and set the initial instantaneous forward variance to $\xi_{0}=0.235^{2}$ (so that $X_{0}=\ln\left(\xi_{0}\right)\approx-2.896$). 
Note we could have considered a piece-wise constant initial variance curve, constant 
over $(T,T+\Delta)$ for every $T$, instead of a flat one -- this will actually be the case in our calibration tests in section \ref{sec:mixed_model}.
To keep track of the error sign, we consider the signed relative error, computed as
$\frac{\text{Approximation Price} - \text{Reference Price}}{\text{Reference Price}}\times100.$

\subsubsection{Numerical tests for the rough Bergomi model}

We consider options struck at $\kappa=0.2$. For the reference price, we used $M=10^{6}$ Monte Carlo samples and $n=300$ discretization points to construct the estimators outlined in Remark \ref{rem:benchmark:prices}. 
We compare the accuracy of the expansions for various values of $\eta$ and $H$.
In each case, we plot the reference price along with its $95\%$ Monte Carlo confidence interval.

\paragraph{Quality of the approximation for different values of the vol-of-variance $\eta$.}

In Figure \ref{fig:vix:eta}, we set $H=0.1$, $\Delta=\frac{1}{12}$,
and choose $10$ evenly-spaced values of $\eta$ ranging from $0.1$ to $1.5.$
We observe that the smaller the $\eta$, the more accurate our expansion formulas. Yet, even for large values of $\eta$ (see Table \ref{tab:calib:params:mixed_rough} for typical values of $\eta$ obtained from the calibration of market $\VIX$ smiles),	our approximations are extremely accurate (almost indistinguishable from the reference prices in the left plots in Figure \ref{fig:vix:eta}): absolute relative errors are less than $0.5\%$ for
the futures contract, $0.3\%$ for the call option, and $1.4\%$ for the put option.

\begin{figure}[H]
\includegraphics[scale=0.5]{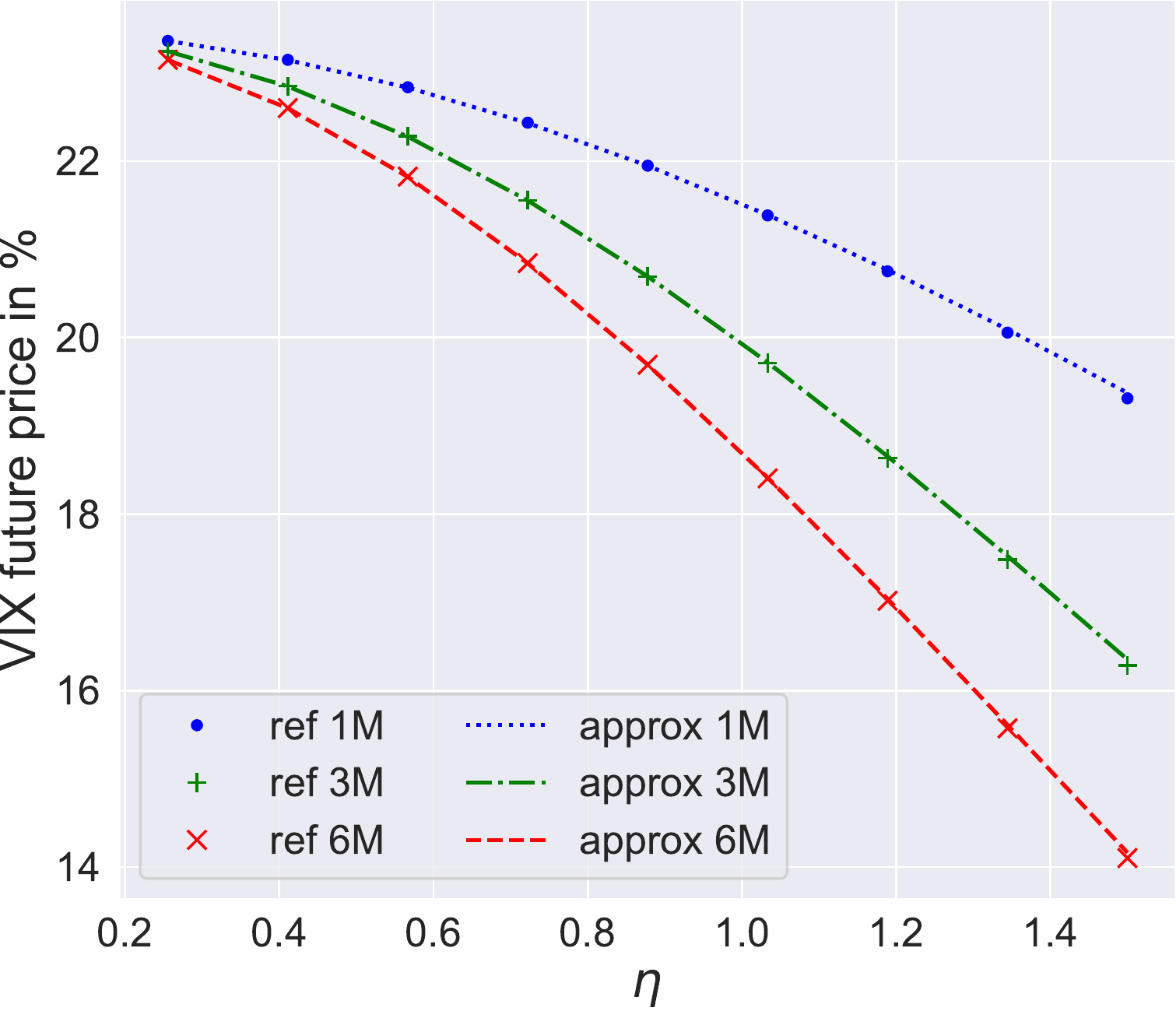}
\includegraphics[scale=0.5]{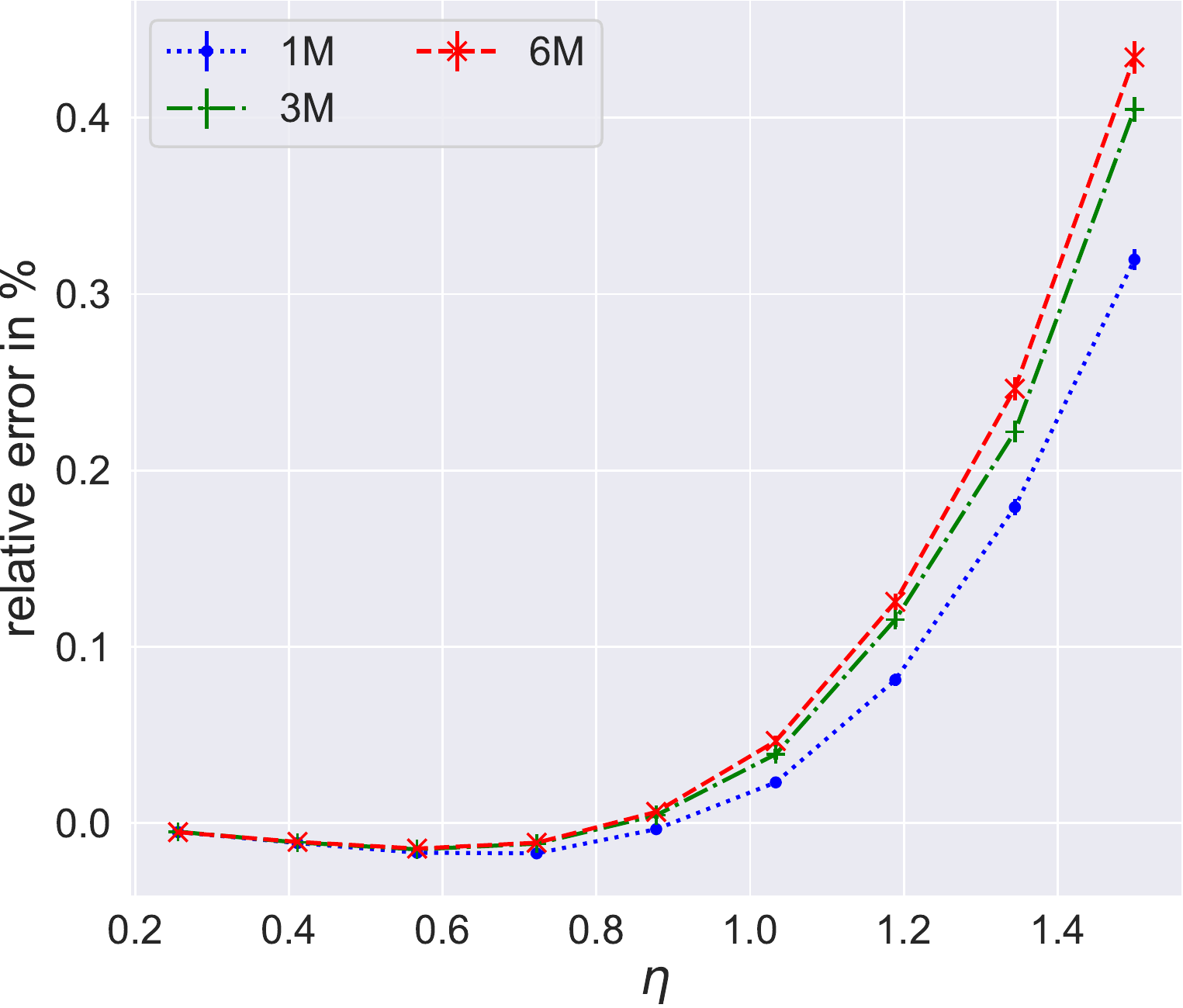}
\\
 \includegraphics[scale=0.5]{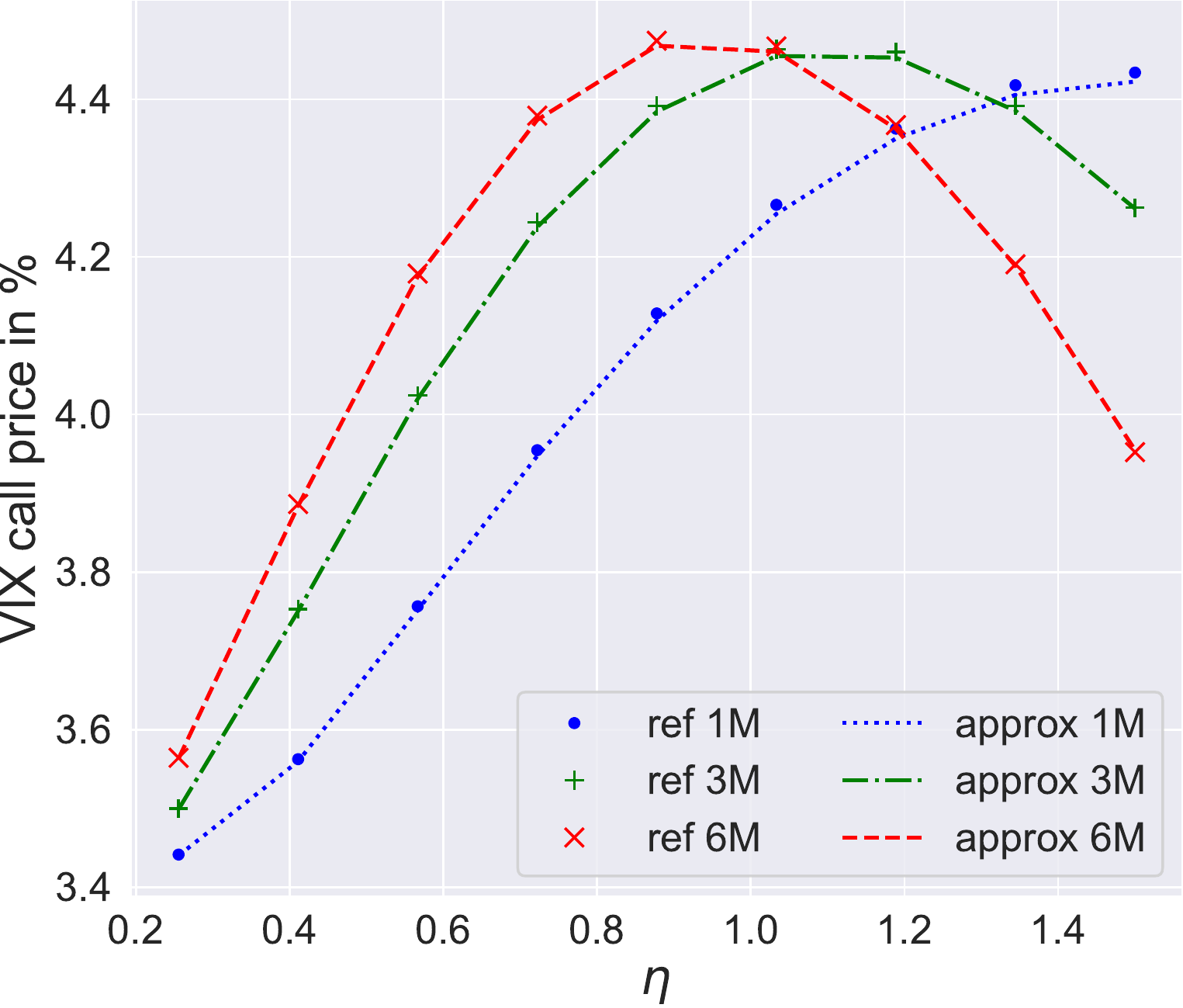}
 \includegraphics[scale=0.5]{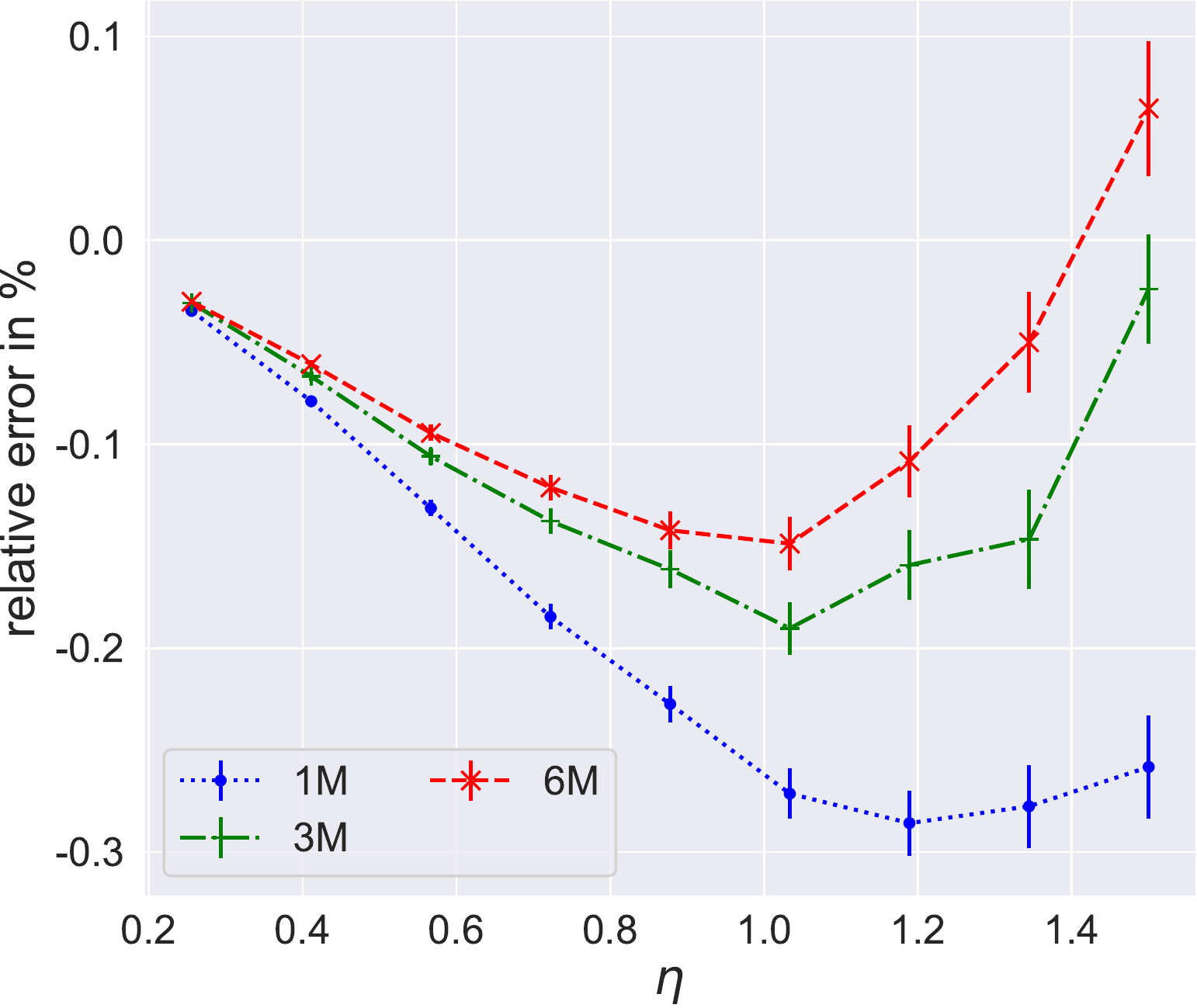}
 \\
 \includegraphics[scale=0.5]{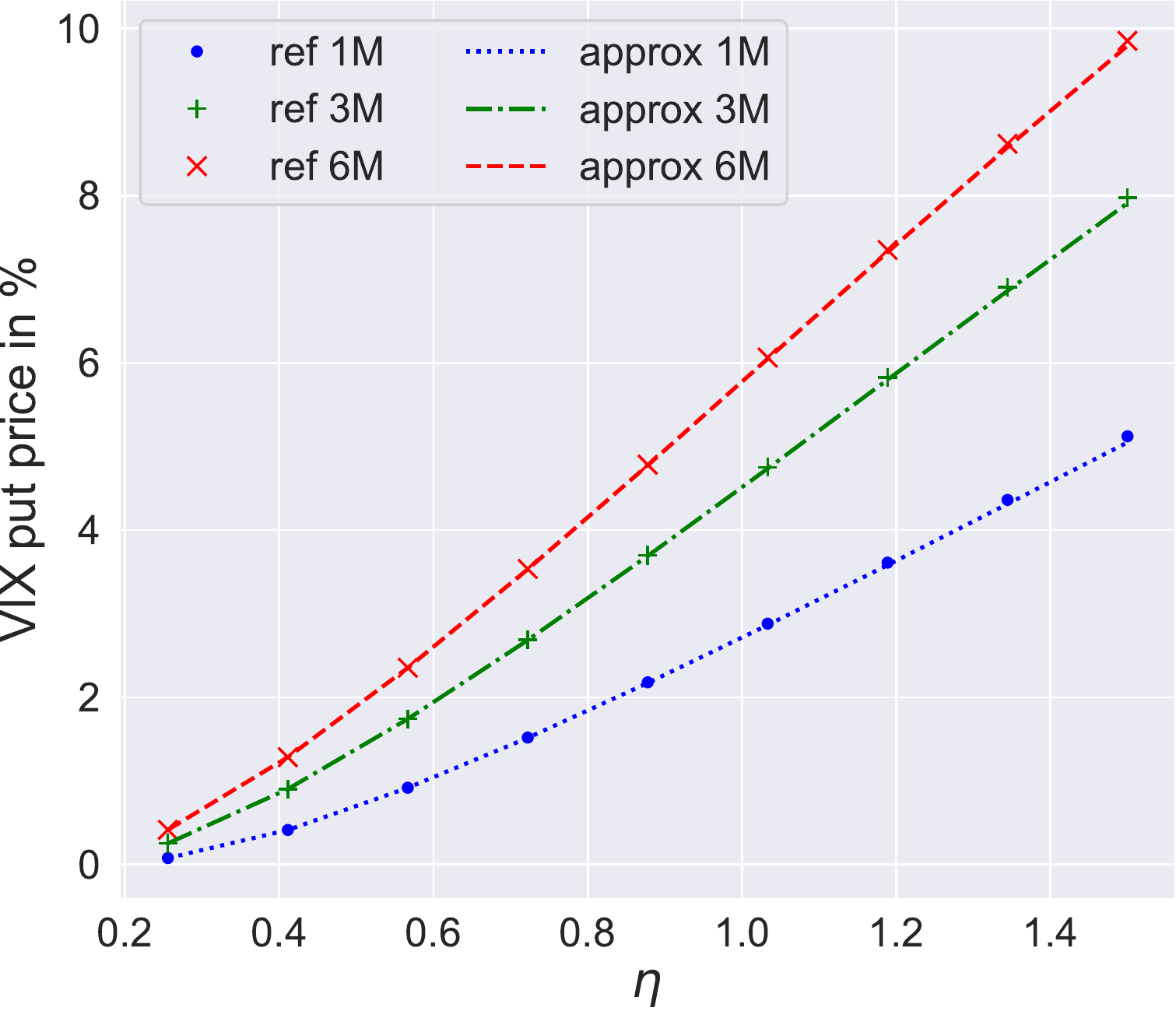}
 \includegraphics[scale=0.5]{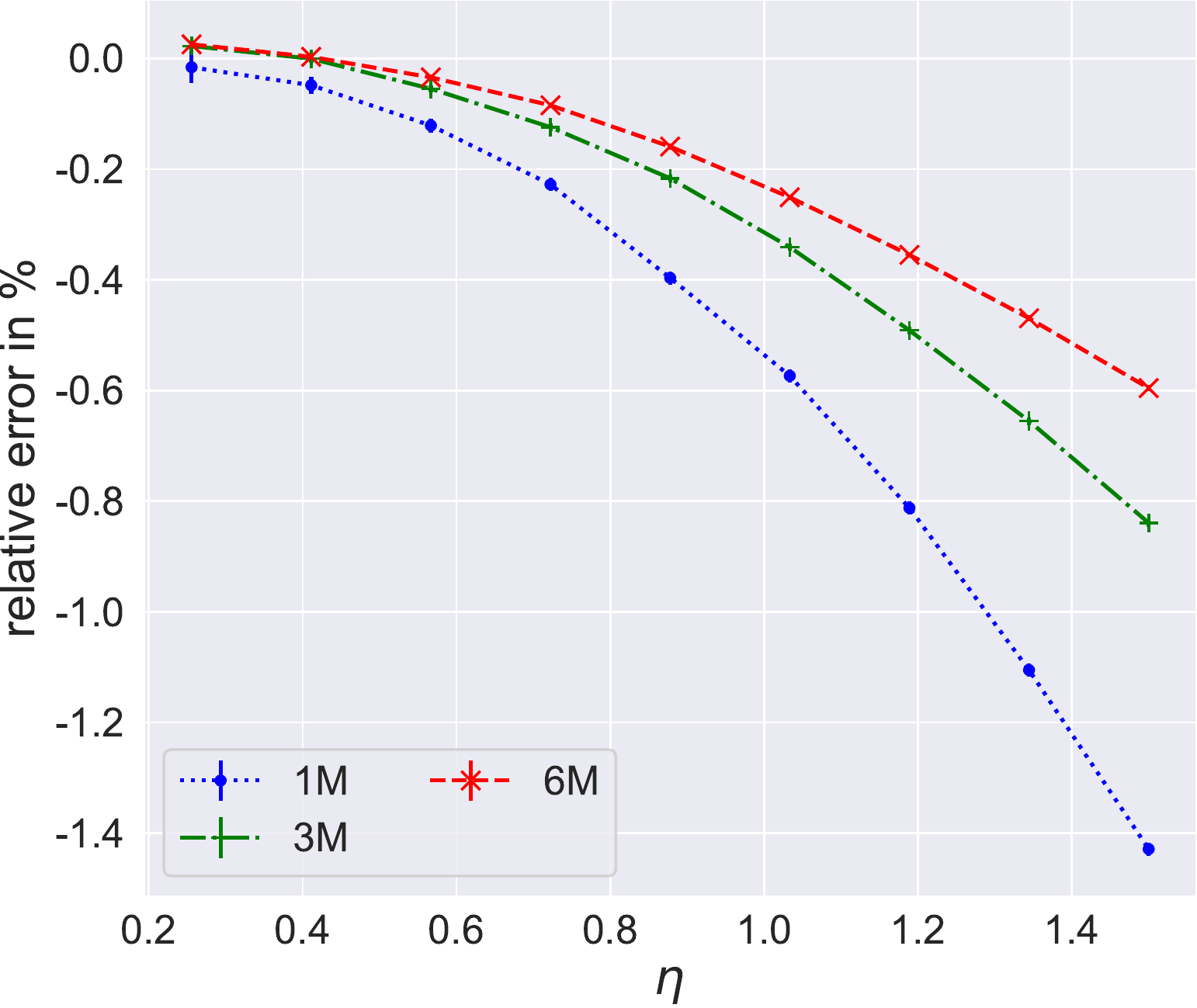}
 \caption{
$\VIX$ futures, call, and put option prices 
for different values of the volatility-of-variance parameter $\eta$ in the rough Bergomi model \eqref{eq:def:rough:bergomi:model}.
\emph{Left}: 
benchmark prices obtained with the Monte Carlo procedure described in Remark \ref{rem:benchmark:prices} along 
with our explicit expansions from Theorem \ref{thm:expansion:plain:model}. 
\emph{Right}: 
relative error in $\%$ 
between the benchmark prices and the expansions.
}
\label{fig:vix:eta} 
\end{figure}

\paragraph{Quality of the approximation for different values of $H$.}
In Figure \ref{fig:vix:H}, we set $\eta=1$, $\Delta=\frac{1}{12}$, 
and choose $10$ evenly-spaced values of $H$ ranging from $0.05$ to $0.4$. Since the closer $H$ to zero, the more singular the kernel, we expect the errors of our expansions to be a decreasing function of $H$,
and this is indeed what we observe in Figure \ref{fig:vix:H}.

\begin{figure}[htp]
\includegraphics[scale=0.5]{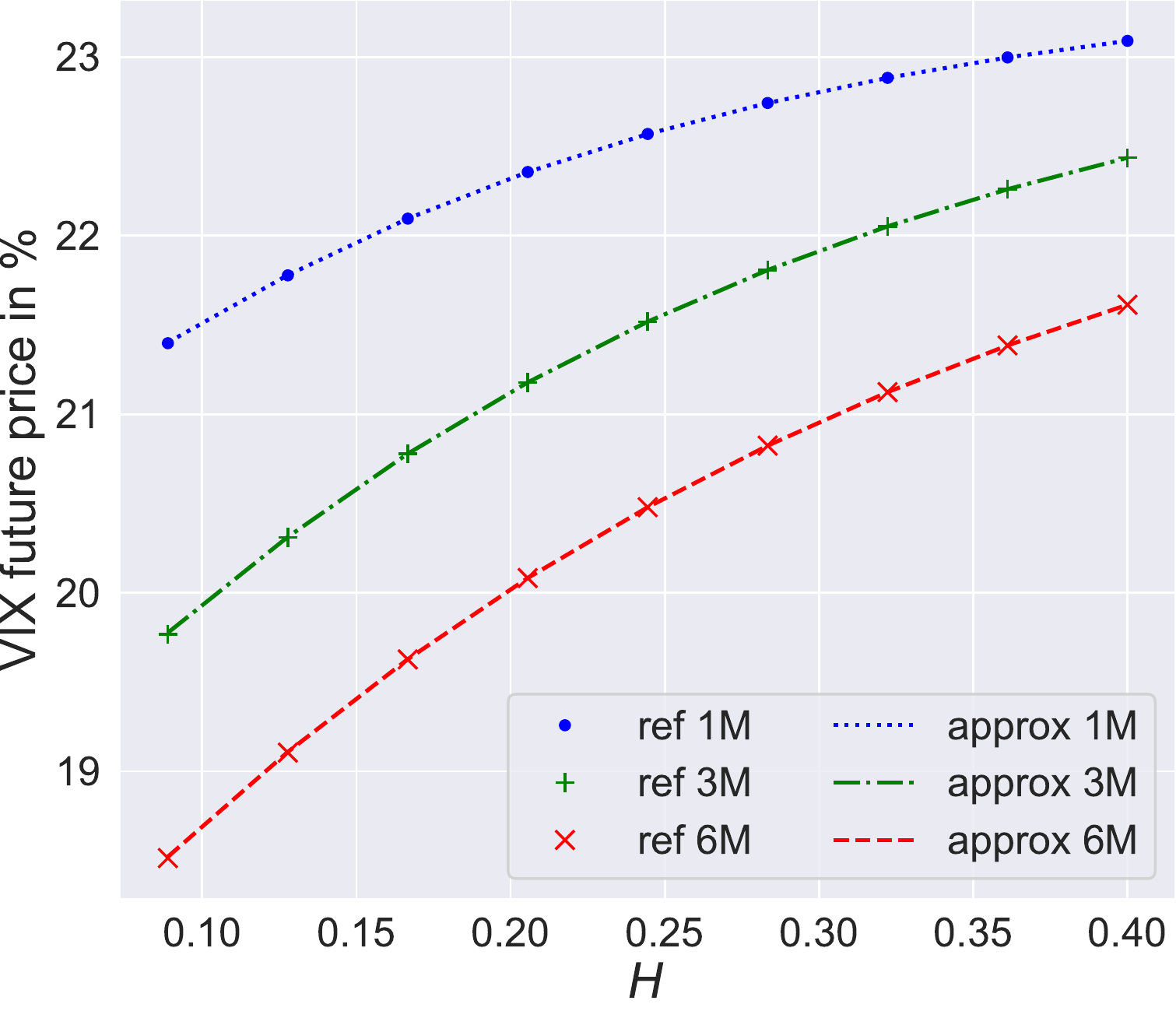}
\includegraphics[scale=0.5]{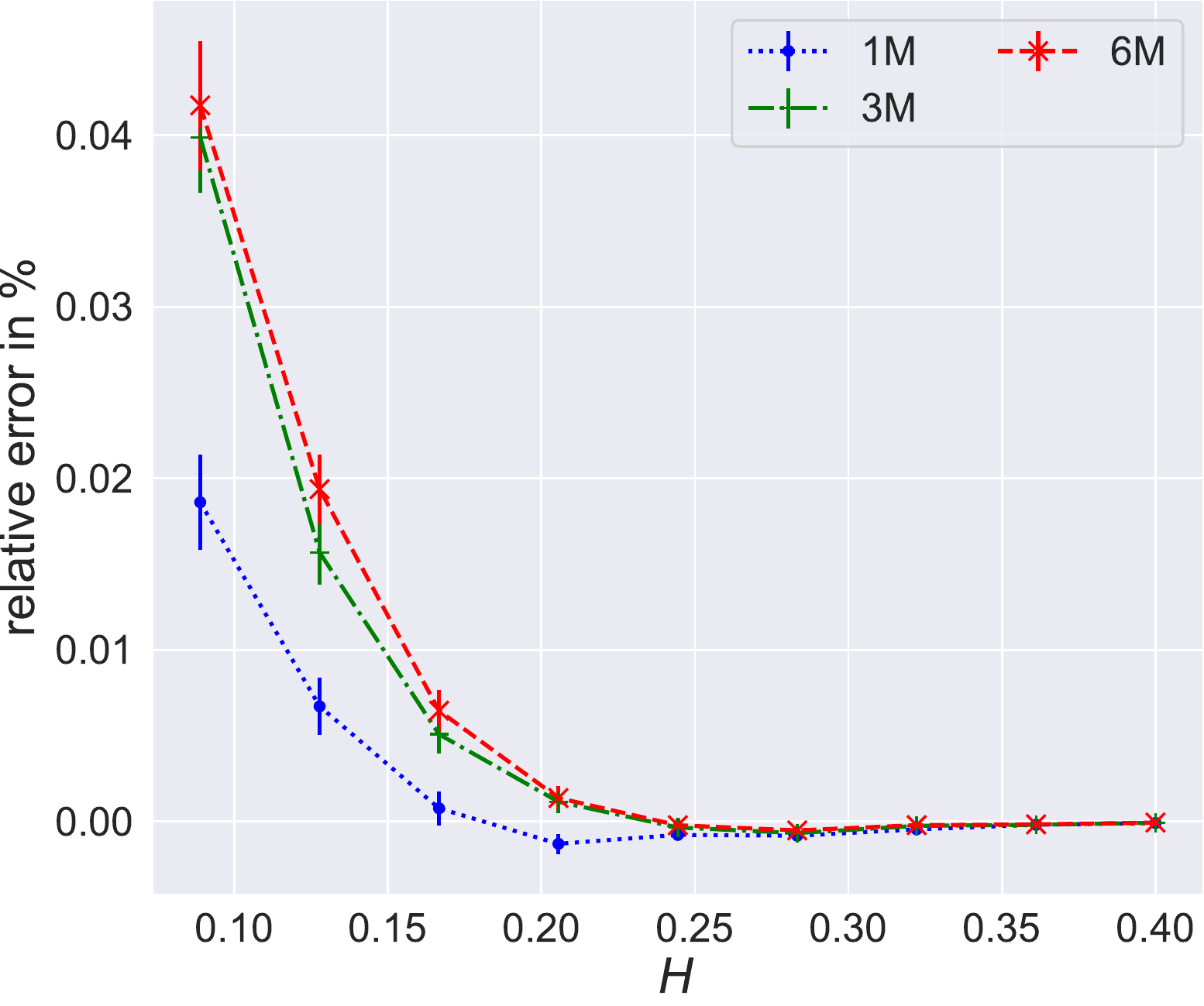}
\\
\includegraphics[scale=0.5]{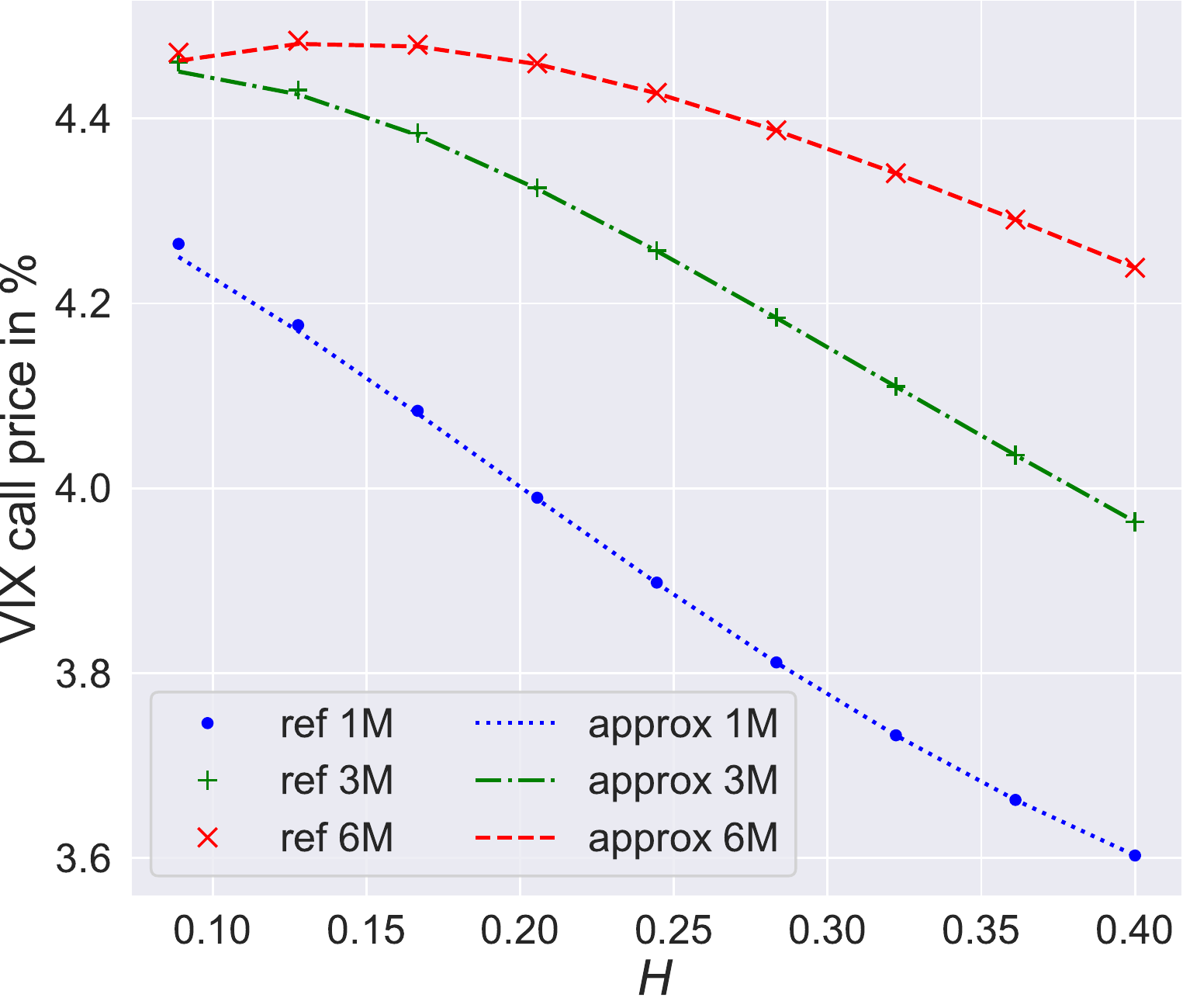}
\includegraphics[scale=0.5]{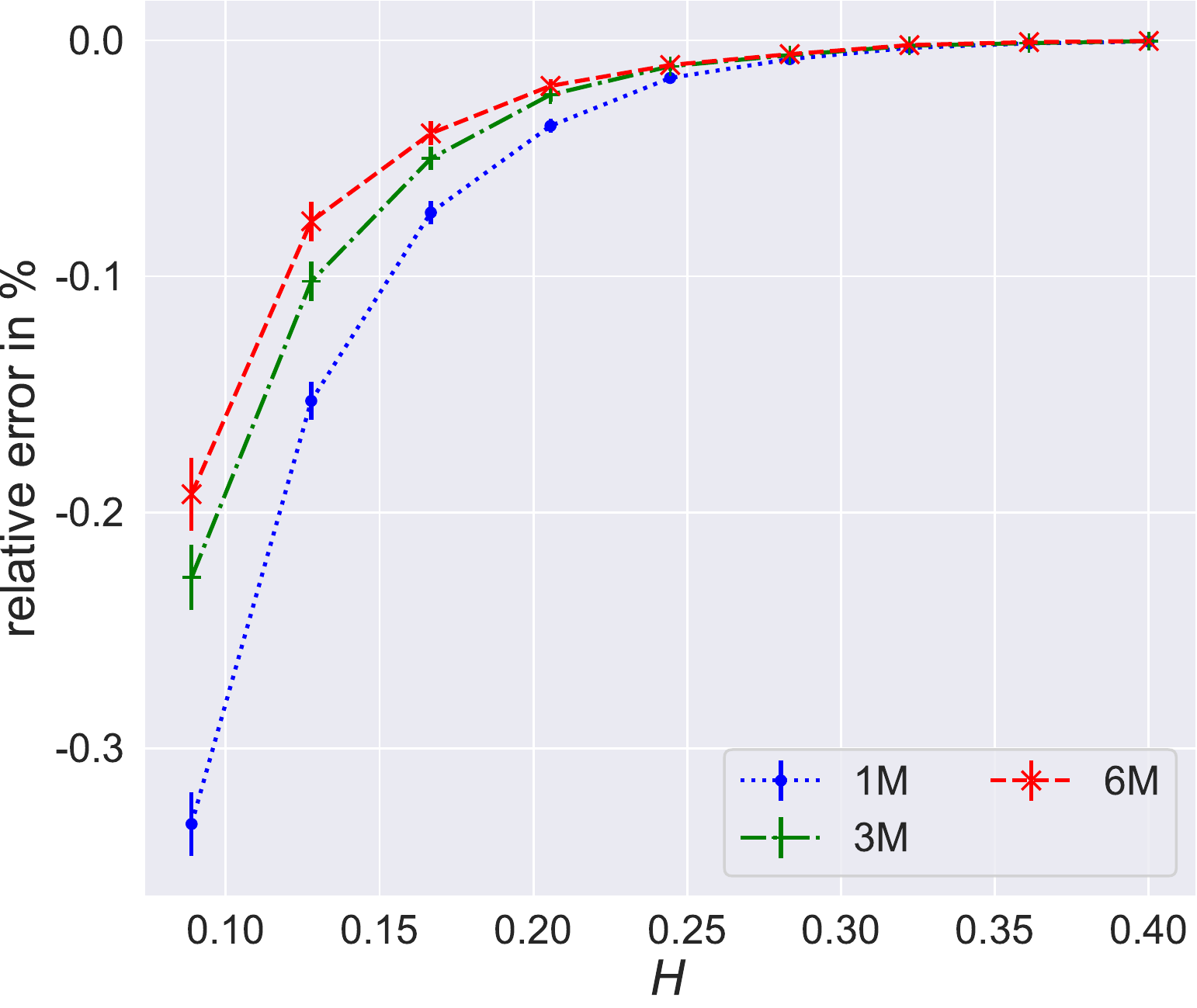}
\\
\includegraphics[scale=0.5]{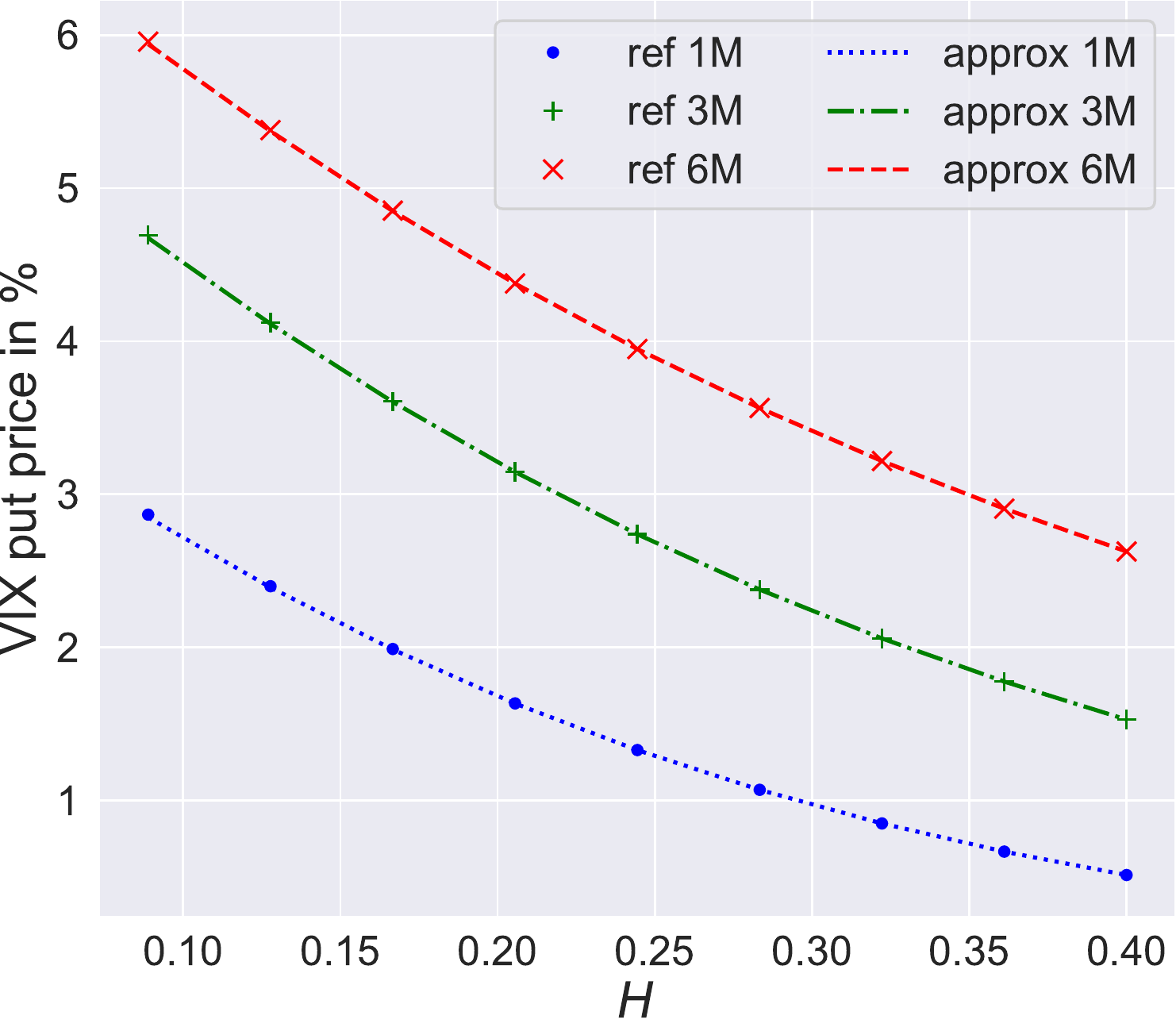}
\includegraphics[scale=0.5]{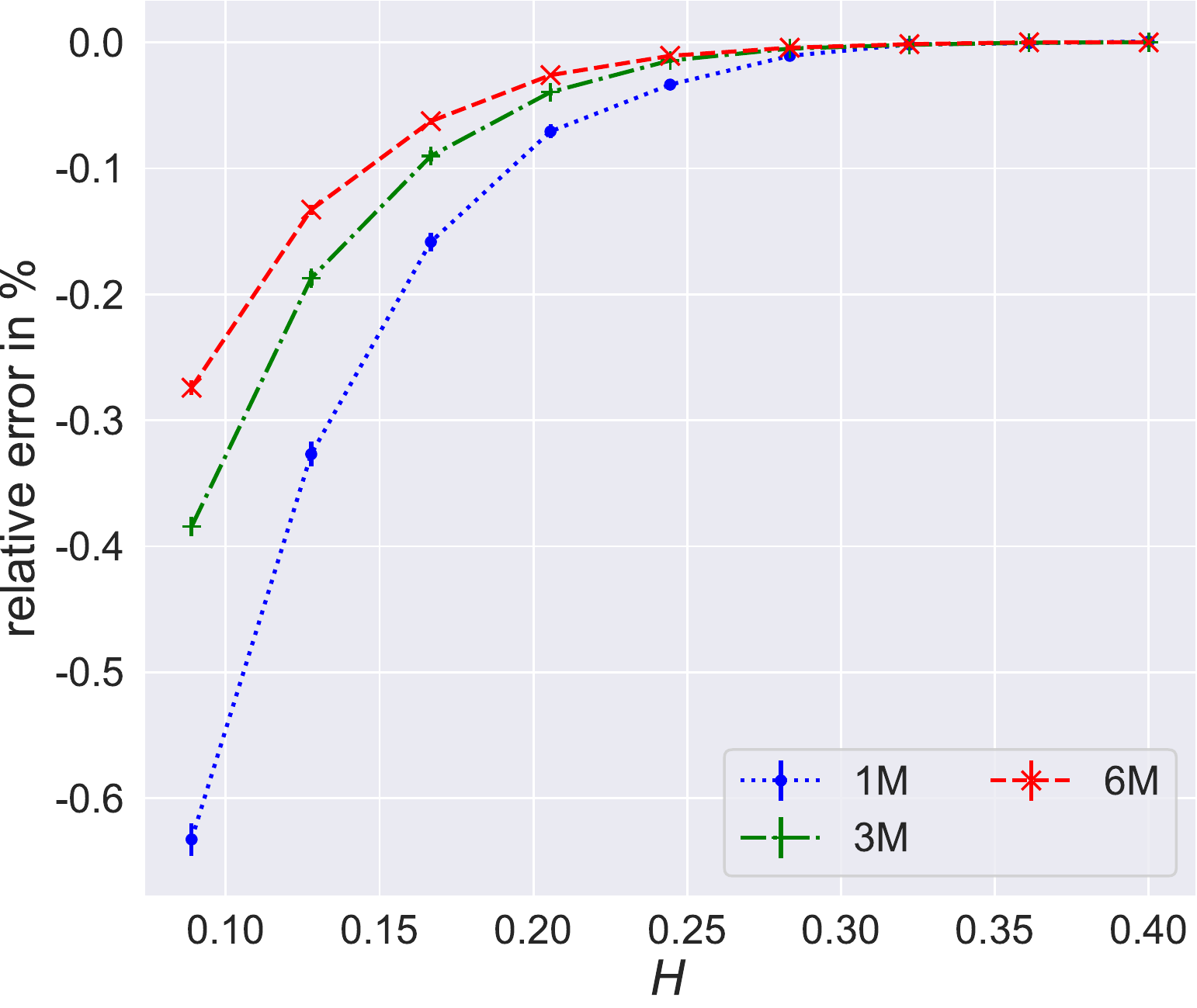}
 \caption{
$\VIX$ futures, call, and put option prices  for different values of the fractional parameter $H$
in the rough Bergomi model \eqref{eq:def:rough:bergomi:model}.
\emph{Left}: 
benchmark prices obtained with
the Monte Carlo procedure described in Remark \ref{rem:benchmark:prices} along 
with our explicit expansions from Theorem \ref{thm:expansion:plain:model}. 
\emph{Right}: 
relative error in $\%$ 
between the benchmark prices and the expansions.
}
\label{fig:vix:H} 
\end{figure}

\subsubsection{Numerical tests for the standard Bergomi model}

We now focus on the standard Bergomi model and
consider $\VIX$ futures and at-the-money $\VIX$ call and put options. 
The reference price is computed as described in Remark \ref{rem:benchmark:prices}, using $80$ nodes for the deterministic quadratures in the time and space dimensions.

\paragraph{Quality of the approximation for different values of the mean-reversion $k$.}
In Figure \ref{fig:vix:k}, we set $\omega = 2$, $\Delta = \frac{1}{12}$, and choose $10$ 
evenly-spaced values of $k$ ranging from $0.5$ to $15$. Note that the values for the call and put options are the same (for we considered at-the-money options); 
this is in line with the put-call parity 
\eqref{eq:put:call:parity:expansion} satisfied by our approximation formulas.
Once again, despite the wide range of values chosen for the mean-reversion parameter $k$, 
we note that the approximation formulas provided by Theorem \ref{thm:expansion:plain:model} are extremely accurate: relative errors are now
less than $10^{-3} \%$ for the $\VIX$ futures and less than $1 \%$ for at-the-money options.

\begin{figure}[H]
	\includegraphics[scale=0.49]{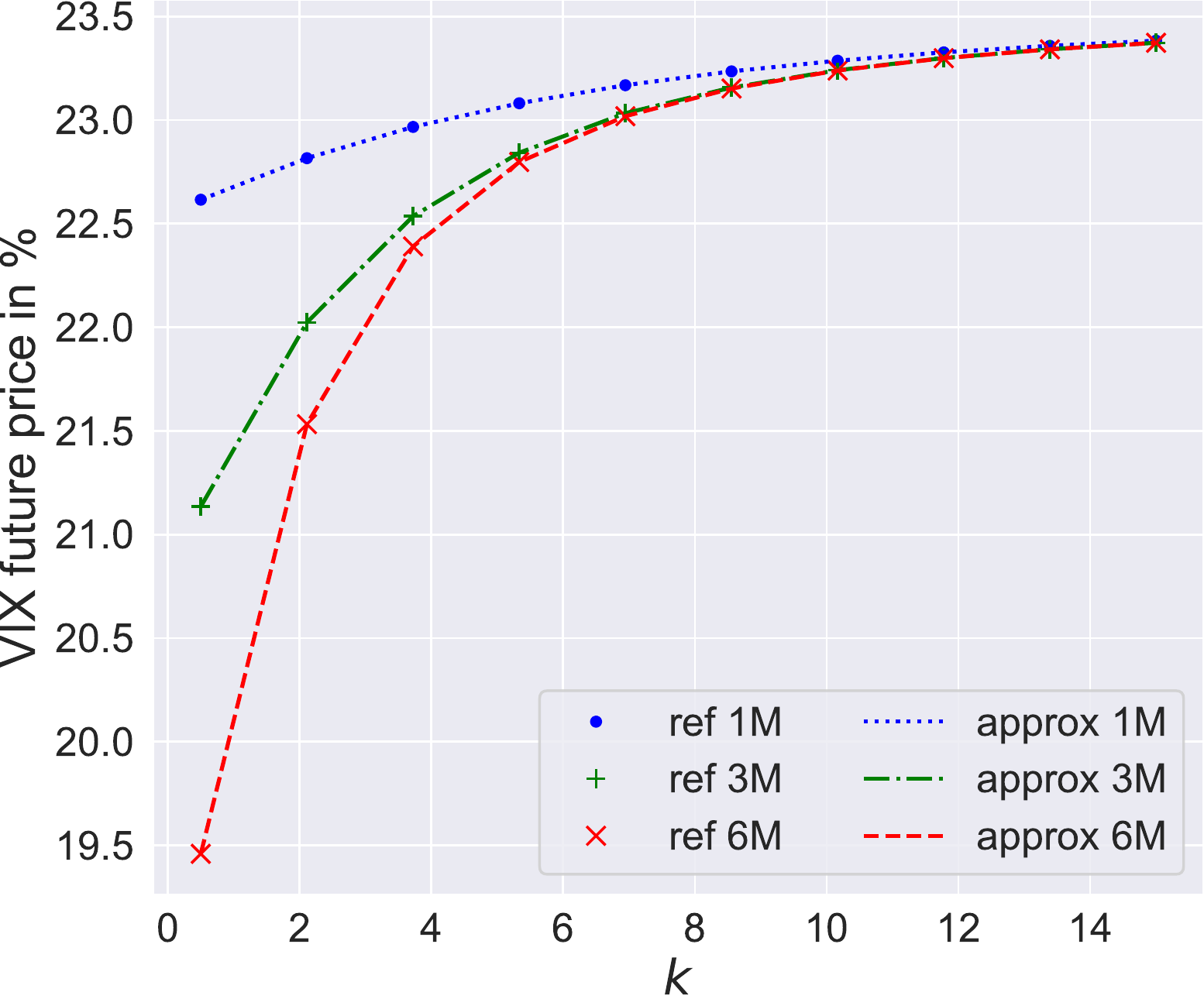}
	\includegraphics[scale=0.49]{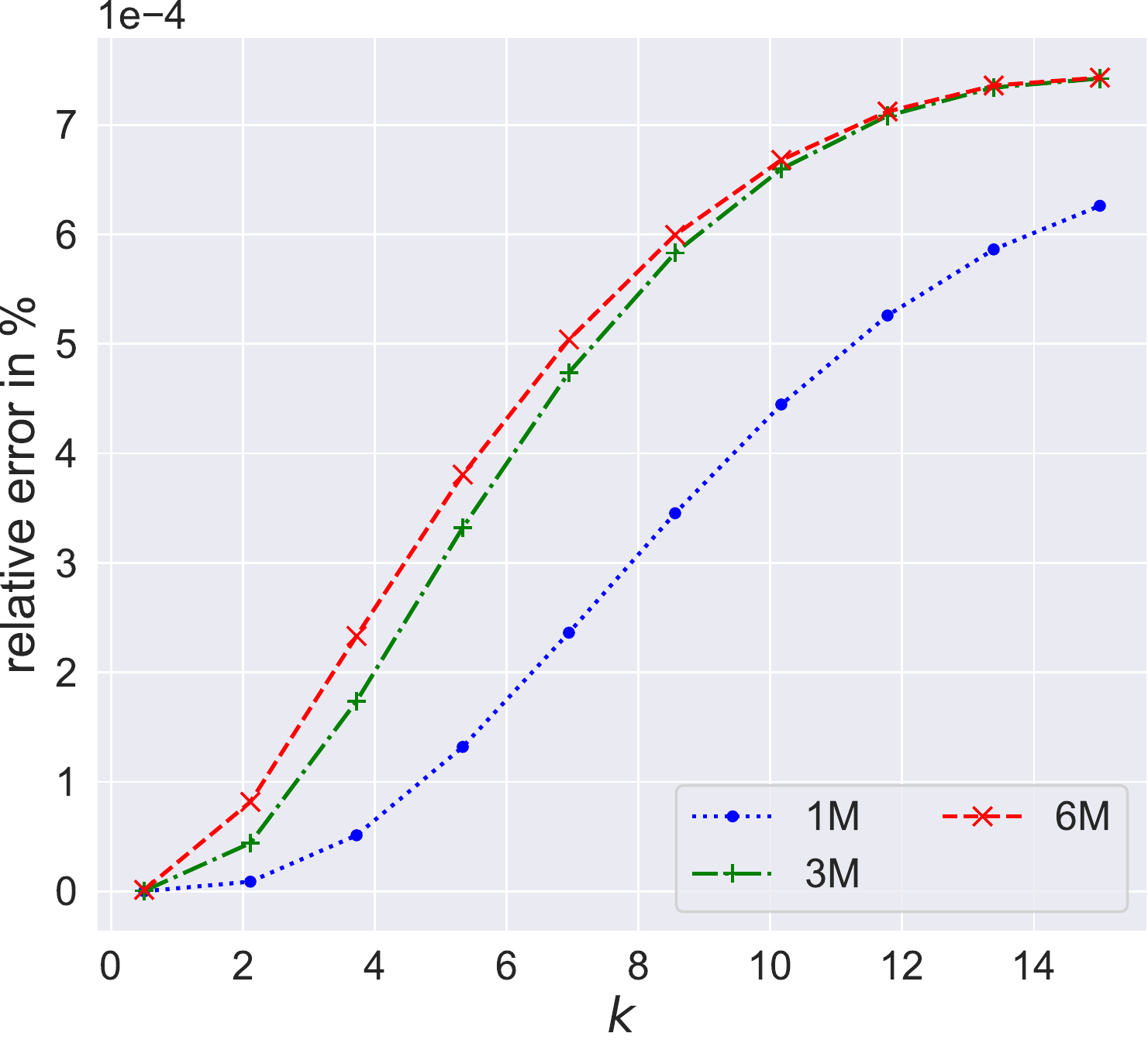}
	\\
	\includegraphics[scale=0.49]{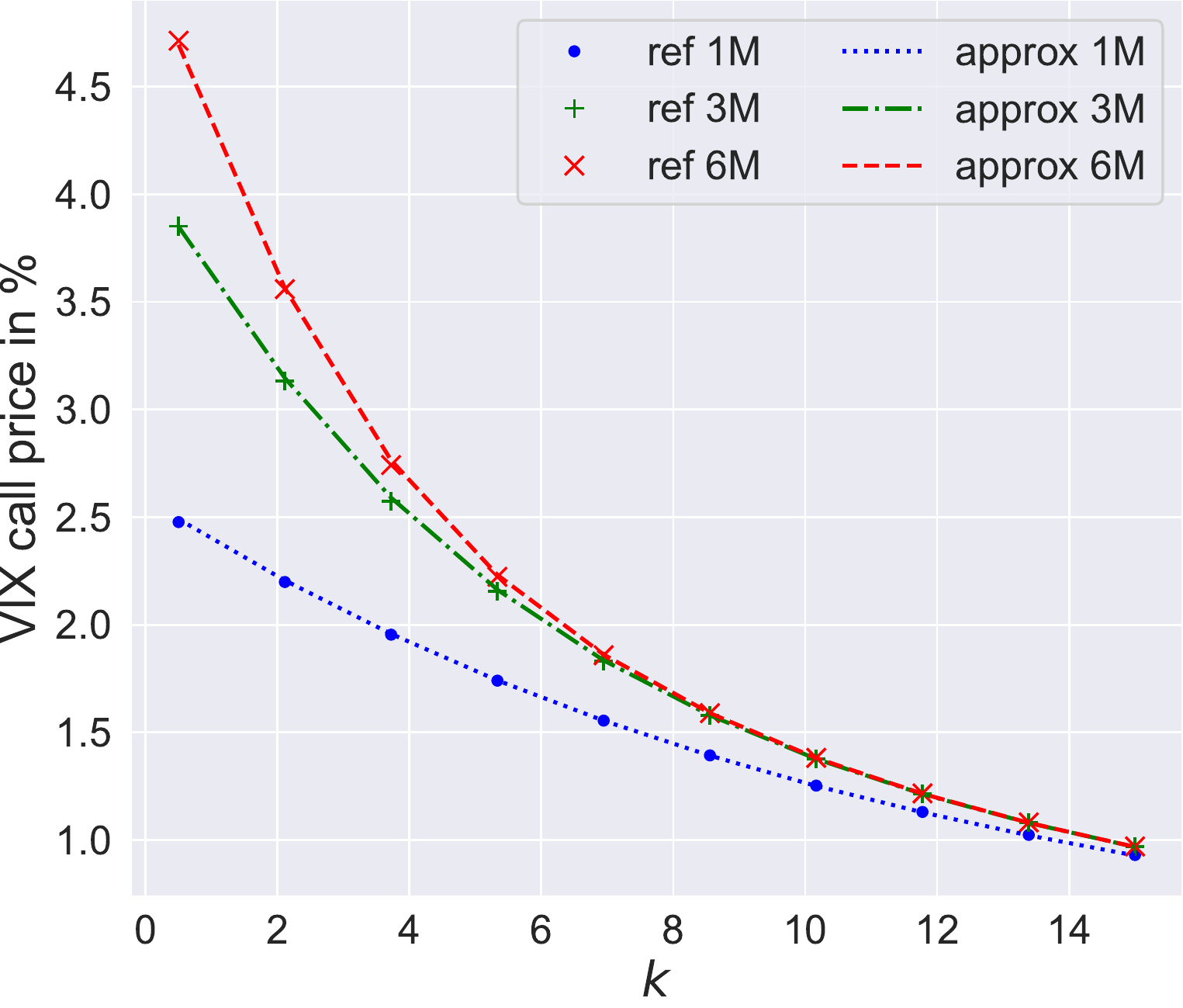}
	\includegraphics[scale=0.49]{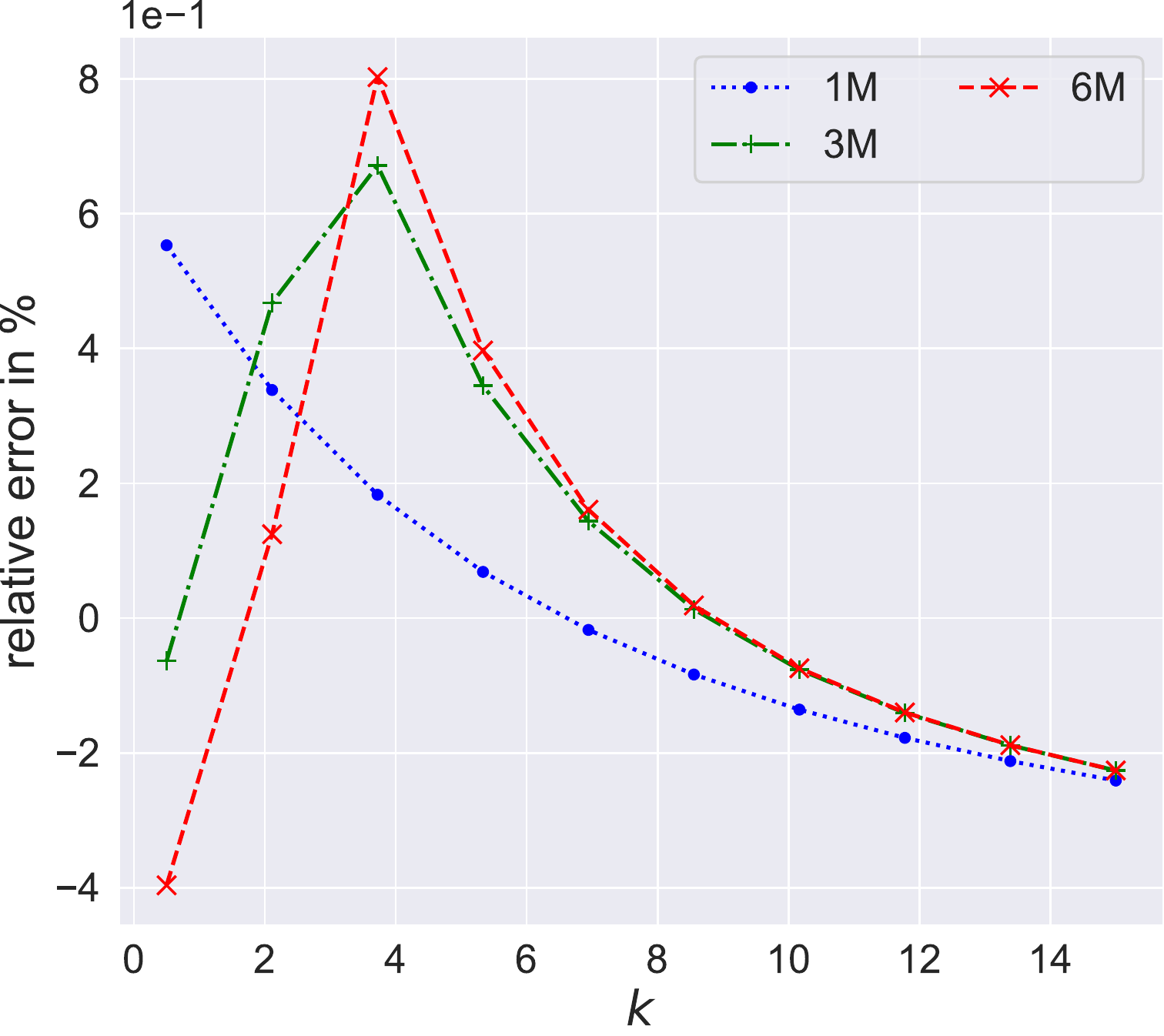}
	\\
	\includegraphics[scale=0.49]{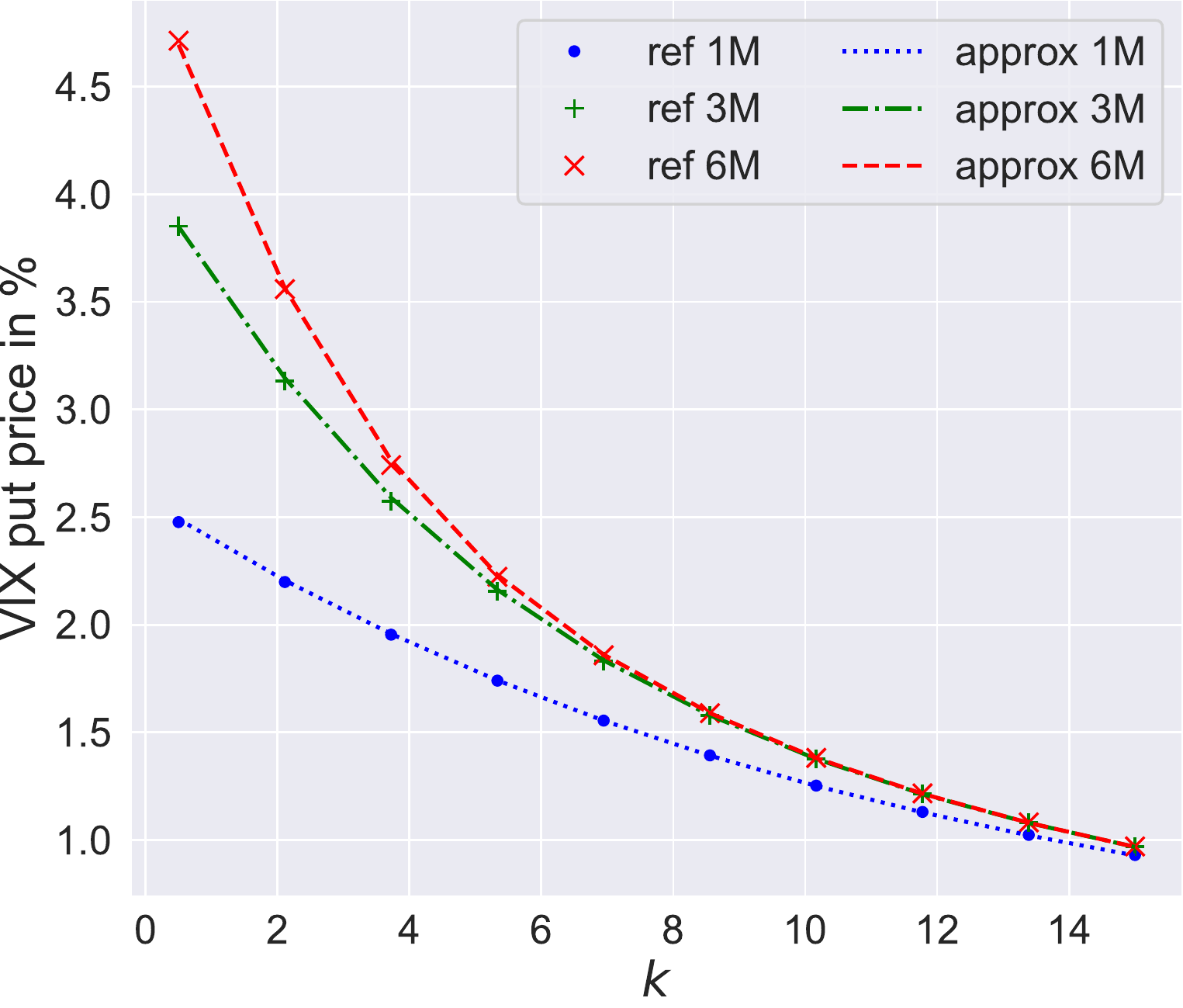}
	\includegraphics[scale=0.49]{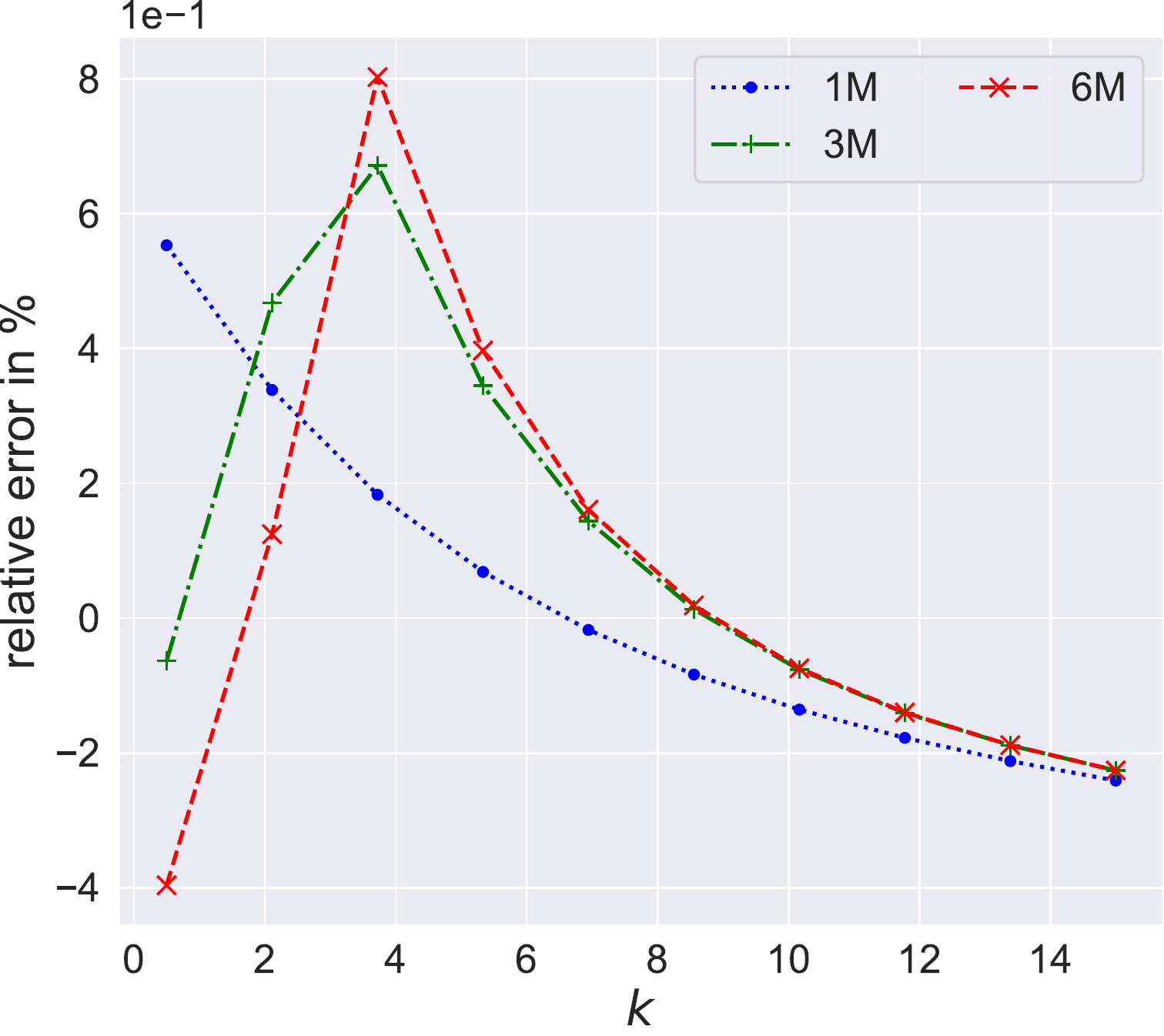}
 \caption{
	$\VIX$ futures, ATM call, and ATM put option prices  for different values of the mean-reversion speed parameter $k$ in the standard Bergomi model \eqref{eq:def:one:factor:bergomi}.
	\emph{Left}: 
	benchmark prices obtained according to the two-dimensional 
	quadrature described in Remark \ref{rem:benchmark:prices}, along 
	with our explicit expansion from Theorem \ref{thm:expansion:plain:model}. 
	\emph{Right}: 
	relative error in $\%$ 
	between the benchmark prices and the expansions.
}
\label{fig:vix:k} 
\end{figure}

\paragraph{Quality of the approximation for different values of the vol-of-variance $\omega$.}
In Figure \ref{fig:vix:w}, we set $k = 1$, $\Delta = \frac{1}{12}$, and choose $10$ evenly-spaced values of $\omega$ ranging from $0.5$ to $6$. 
	
\begin{figure}[H]
	\includegraphics[scale=0.49]{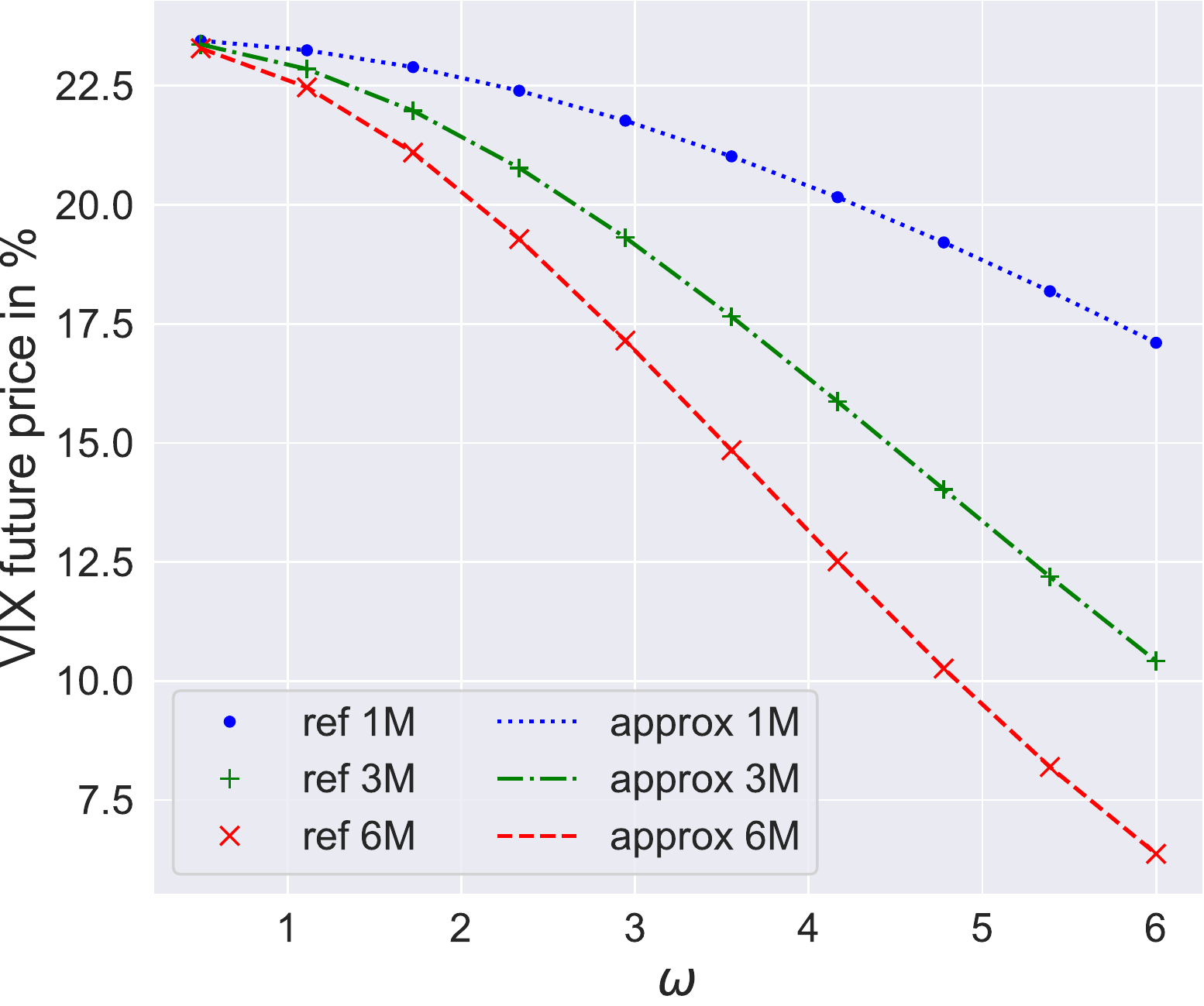}
	\includegraphics[scale=0.49]{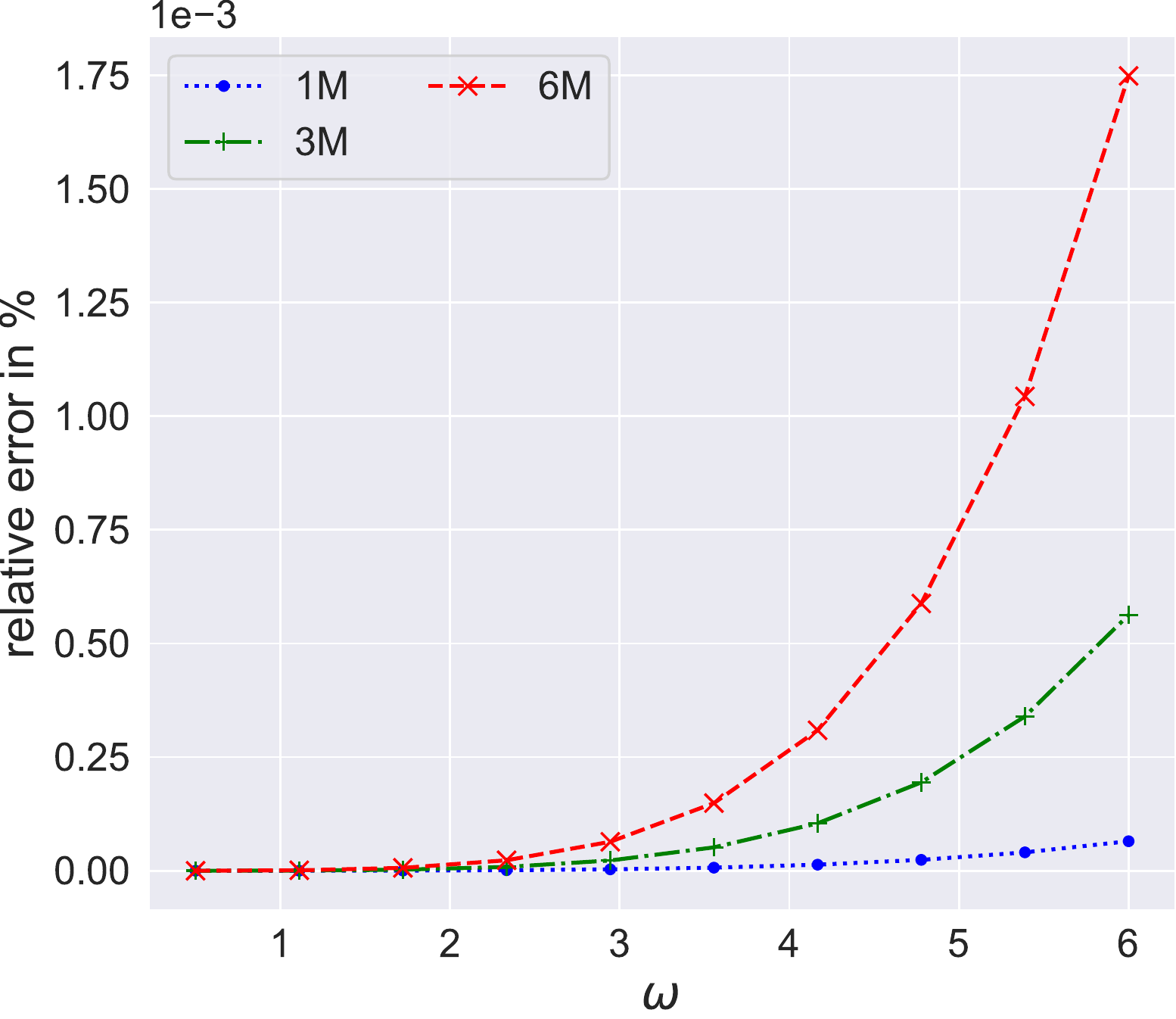}
	\\
	\includegraphics[scale=0.49]{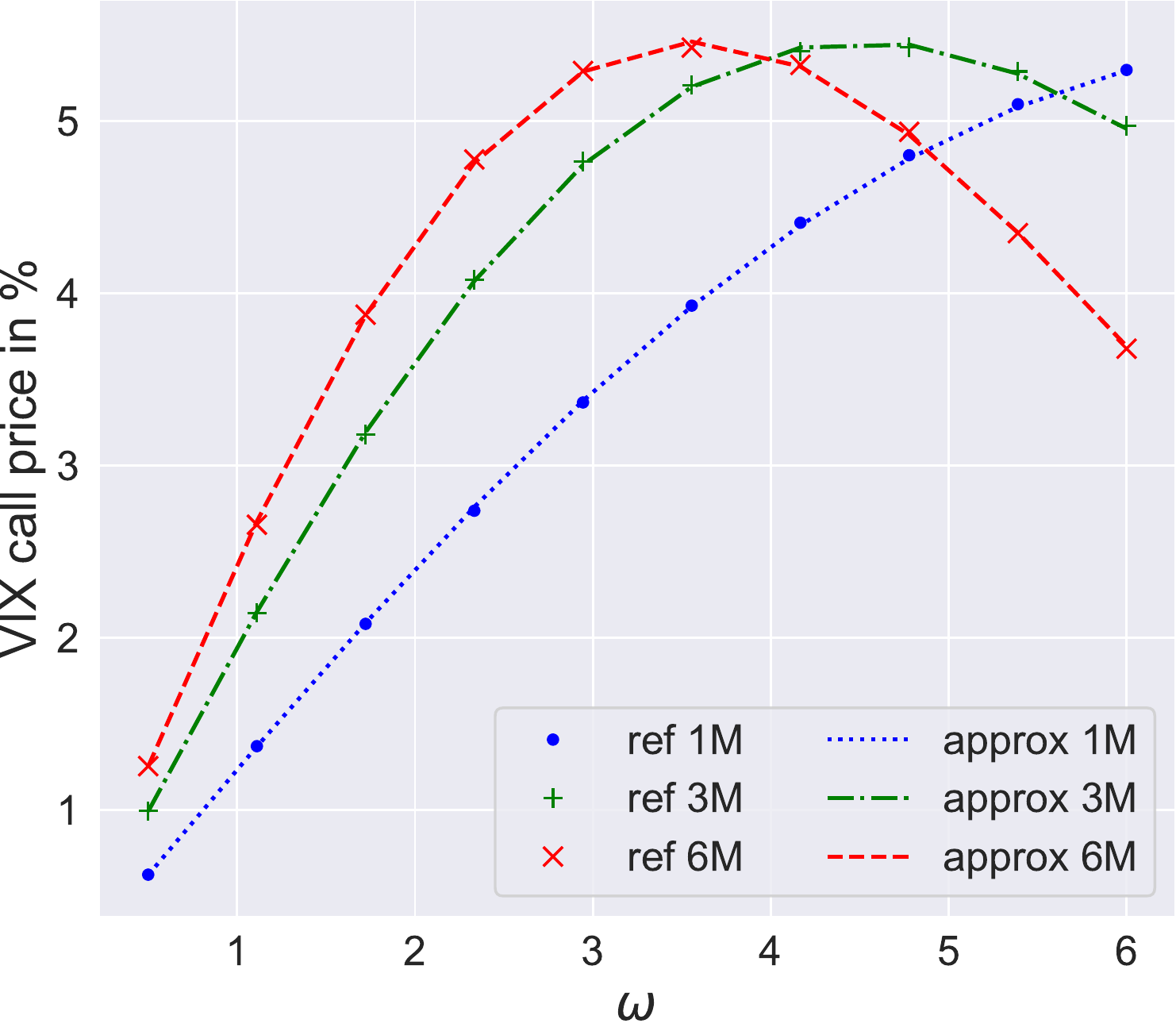}
	\includegraphics[scale=0.49]{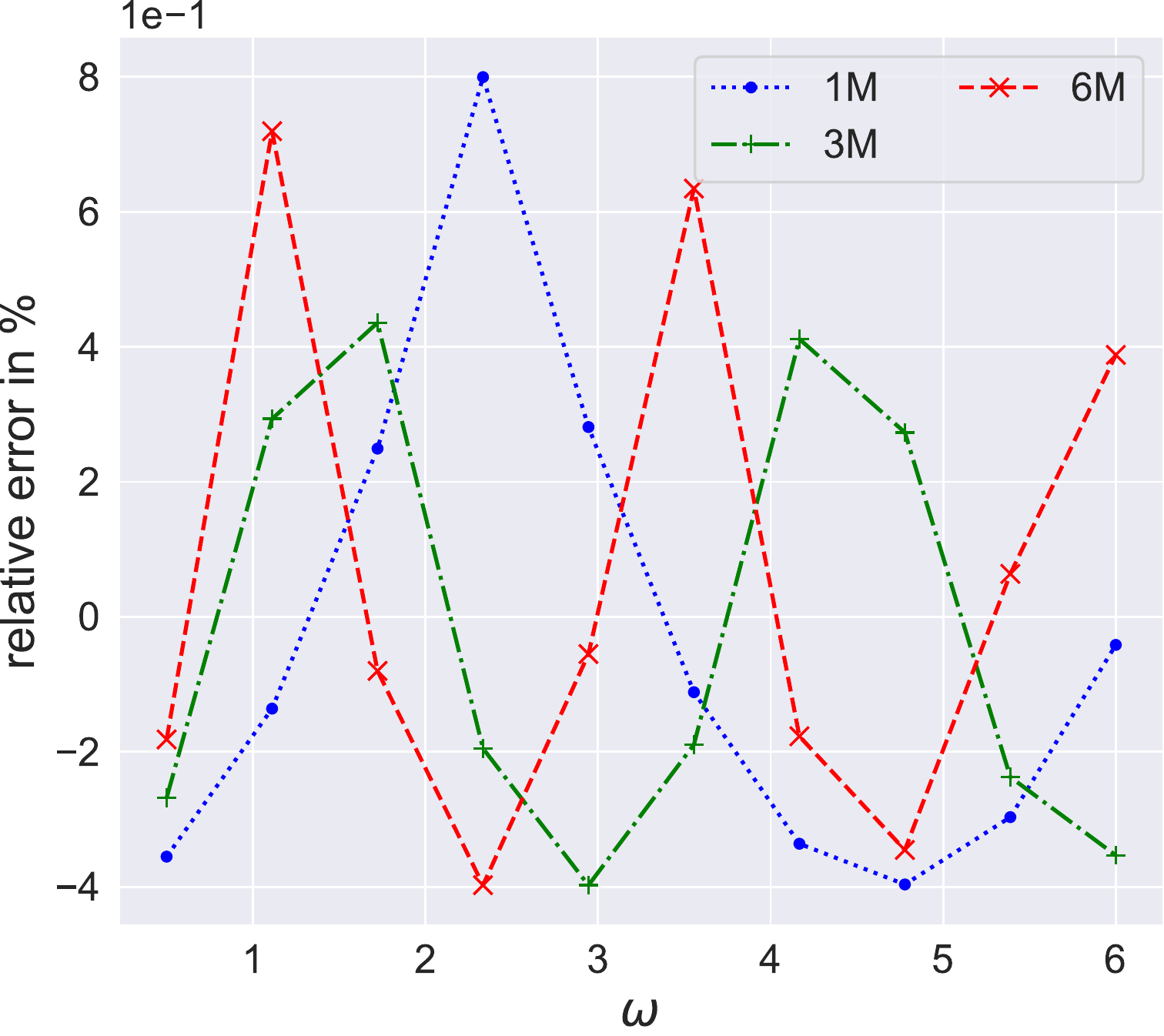}
	\\
	\includegraphics[scale=0.49]{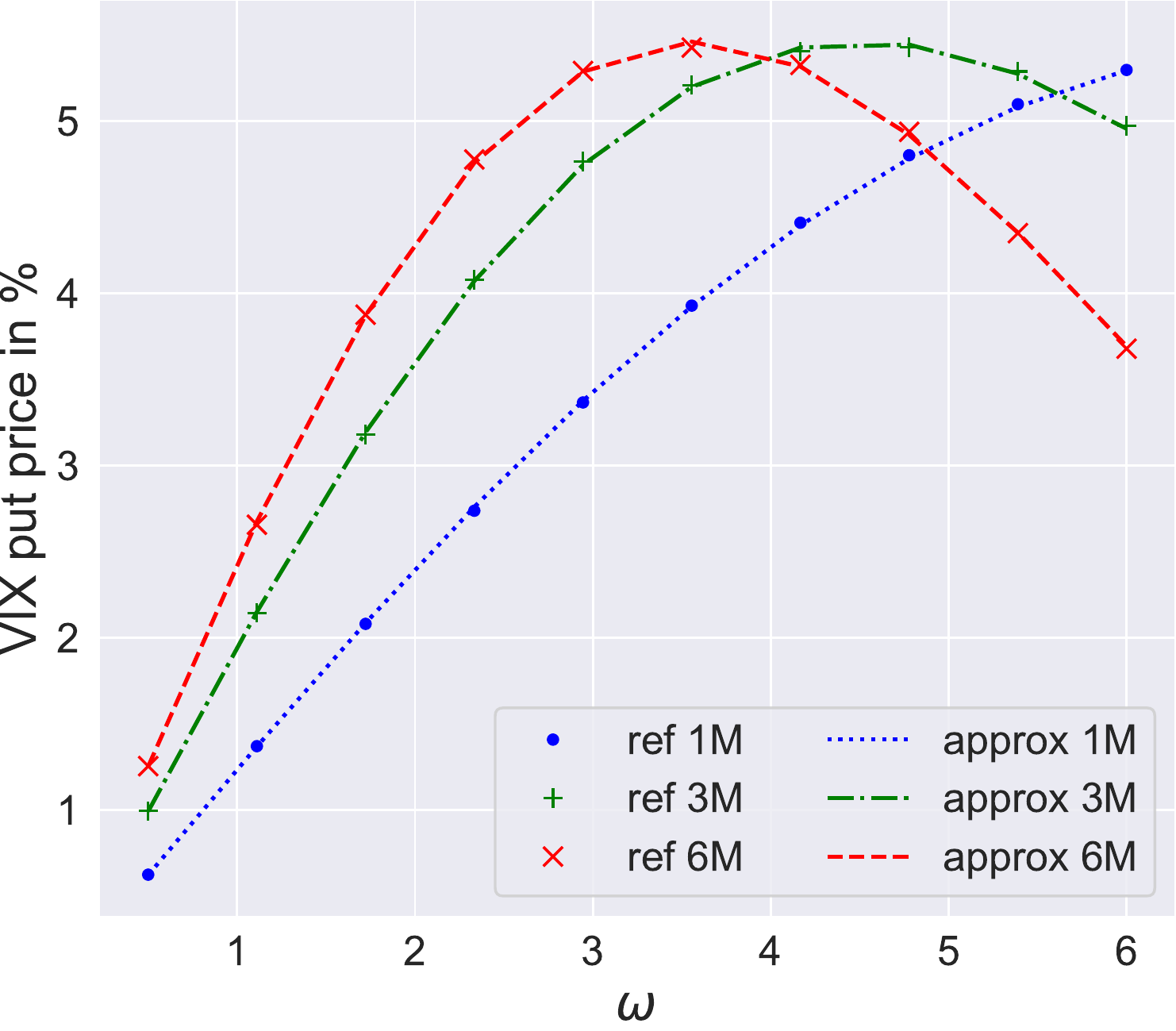}
	\includegraphics[scale=0.49]{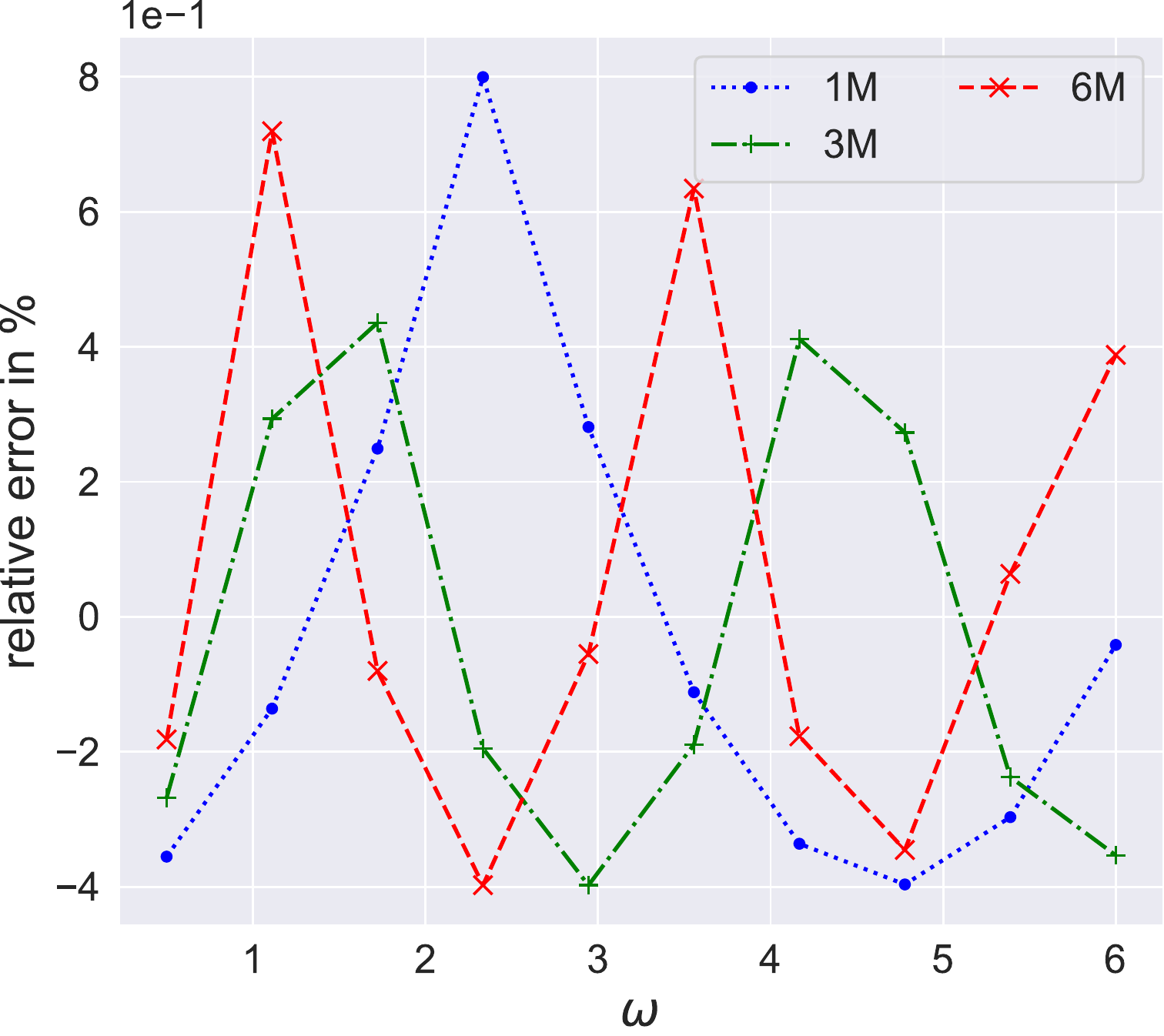}
	\caption{
		$\VIX$ futures, ATM call  and ATM put option prices for different values of the volatility-of-variance parameter $\omega$ in the standard Bergomi model  \eqref{eq:def:one:factor:bergomi}.
		\emph{Left}: 
		benchmark prices obtained according to the two-dimensional 
		quadrature described in Remark \ref{rem:benchmark:prices}, along 
		with our explicit expansion from Theorem \ref{thm:expansion:plain:model}. 
		\emph{Right}: 
		relative error in $\%$ 
		between the benchmark prices and the expansions.
	}
	\label{fig:vix:w} 
\end{figure}

\subsubsection{Behavior of the error terms for different values of $\Delta$}

We wish to compare the theoretical estimates $\cO(\Delta^{3})$ and $\cO(\Delta^{3H})$ given in Corollaries \ref{cor:check:onefactor:bergomi} and \ref{cor:check:rough:bergomi} for the error terms  in the standard Bergomi and rough Bergomi models with their empirical behavior.
In Figure \ref{fig:check:error},  we plot the absolute difference between the reference price and our approximations for futures and at-the-money call and put options, for several values of the time-window $\Delta$, in a log-log plot.
We consider a grid of $10$ evenly-spaced values of $\Delta$ in $[0.05,0.25]$, for
both models, and set $T= 1 \mbox{ month}$ and $\xi_0 = 0.04$. 
The parameters $H, \eta, k, \omega$ are displayed above each figure.

\begin{figure}[H]	
\includegraphics[scale=0.45]{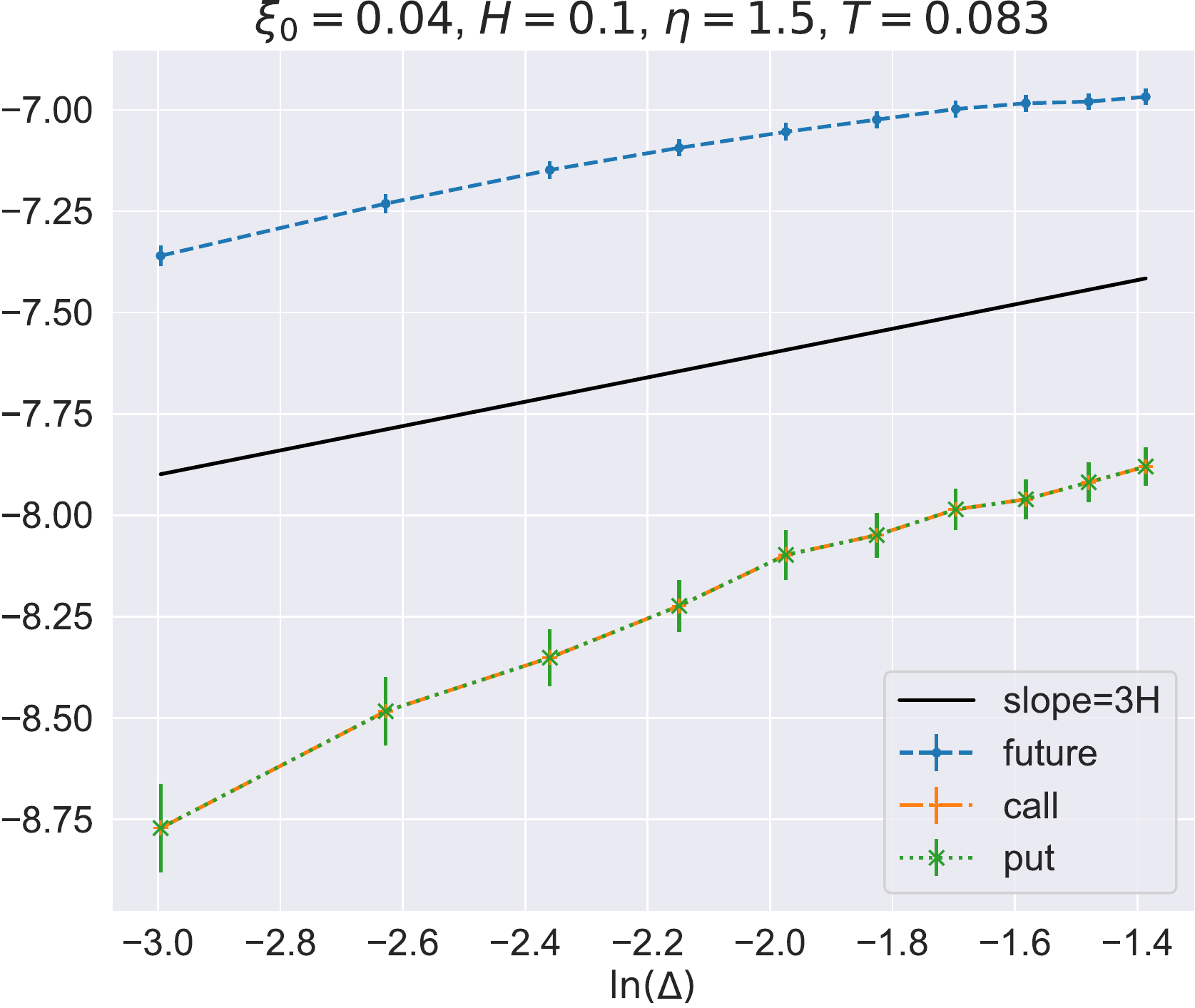}
\includegraphics[scale=0.45]{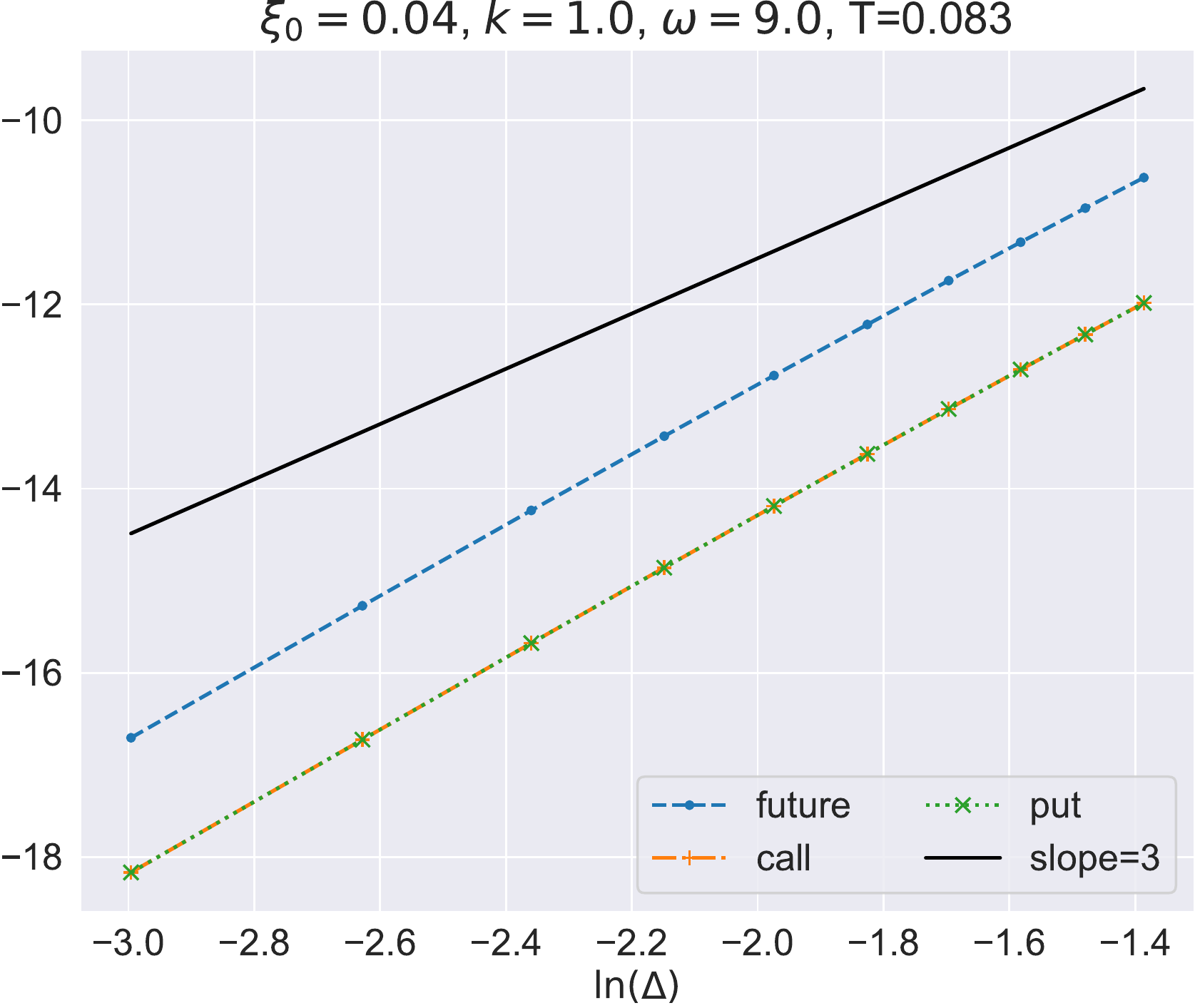}
\caption{ 
Log-log plot of the absolute difference between the reference price and our approximation formula from Theorem \ref{thm:expansion:plain:model} for futures, at-the-money call and at-the-money put options
as a function of $\Delta$, for the rough Bergomi model (left figure) and the standard Bergomi model (right figure).
}
\label{fig:check:error}
\end{figure}

We observe consistency with the error behavior announced in Corollaries \ref{cor:check:onefactor:bergomi} and \ref{cor:check:rough:bergomi} for the two models.

\subsubsection{VIX implied volatility}

\label{sec:implied:vol:plain}

The implied volatility of $\VIX$ options is computed from the Black-Scholes formula, using as forward parameter the model-generated $\VIX$ futures.
In Figure \ref{fig:implied:vol:plain}, we plot the $\VIX$ smile in the rough Bergomi model obtained from the reference option prices when $\xi_{0}= 0.235^{2},\eta=1$, and $H=0.1$,
along with its approximation from Theorem \ref{thm:expansion:plain:model} and
the associated signed relative error. Our approximation formula is again very accurate and yields relative errors for implied volatilities smaller (in absolute value) than $1.5\%$ for
a one-month maturity, $0.5\%$ for three months, and $0.35\%$ for six months.

It has already been observed and reported by several authors \cite{bayer2016pricing,horvath2018volatility} that the VIX smile generated by exponential forward
variance models \eqref{eq:sde:forward:variance} is almost flat, as also observed in Figure  \ref{fig:implied:vol:plain}.
This is precisely a consequence of the fact that the  true $\VIX_{T}^{2}$ random variable is well approximated by the log-normal proxy $\VIX_{T,{\rm P}}^{2}$ for realistic 
model parameters, so that the $\VIX$ itself is not far from a log-normal random variable with a flat smile structure (as we have already pointed out, our Theorems \ref{thm:estimate:nue:nueproxy} 
and \ref{thm:expansion:plain:model}  precisely quantify the difference existing between the distributions of these two random variables).
In practice, market $\VIX$ smiles exhibit a pronounced positive skew, which has motivated the search for more general model classes, which we now consider in section \ref{sec:mixed_model}.

\begin{figure}[H]
\includegraphics[scale=0.5]{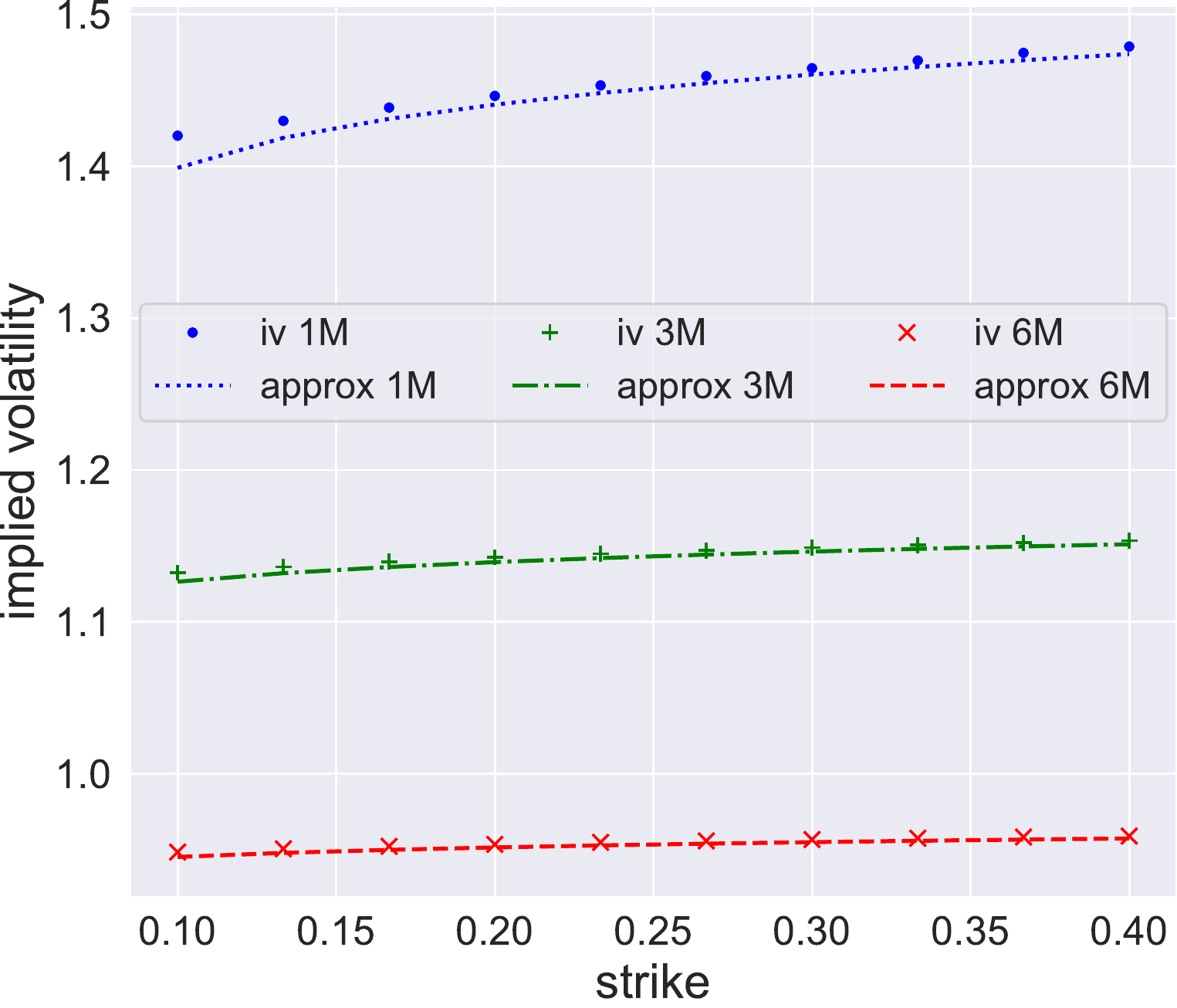} 
\includegraphics[scale=0.5]{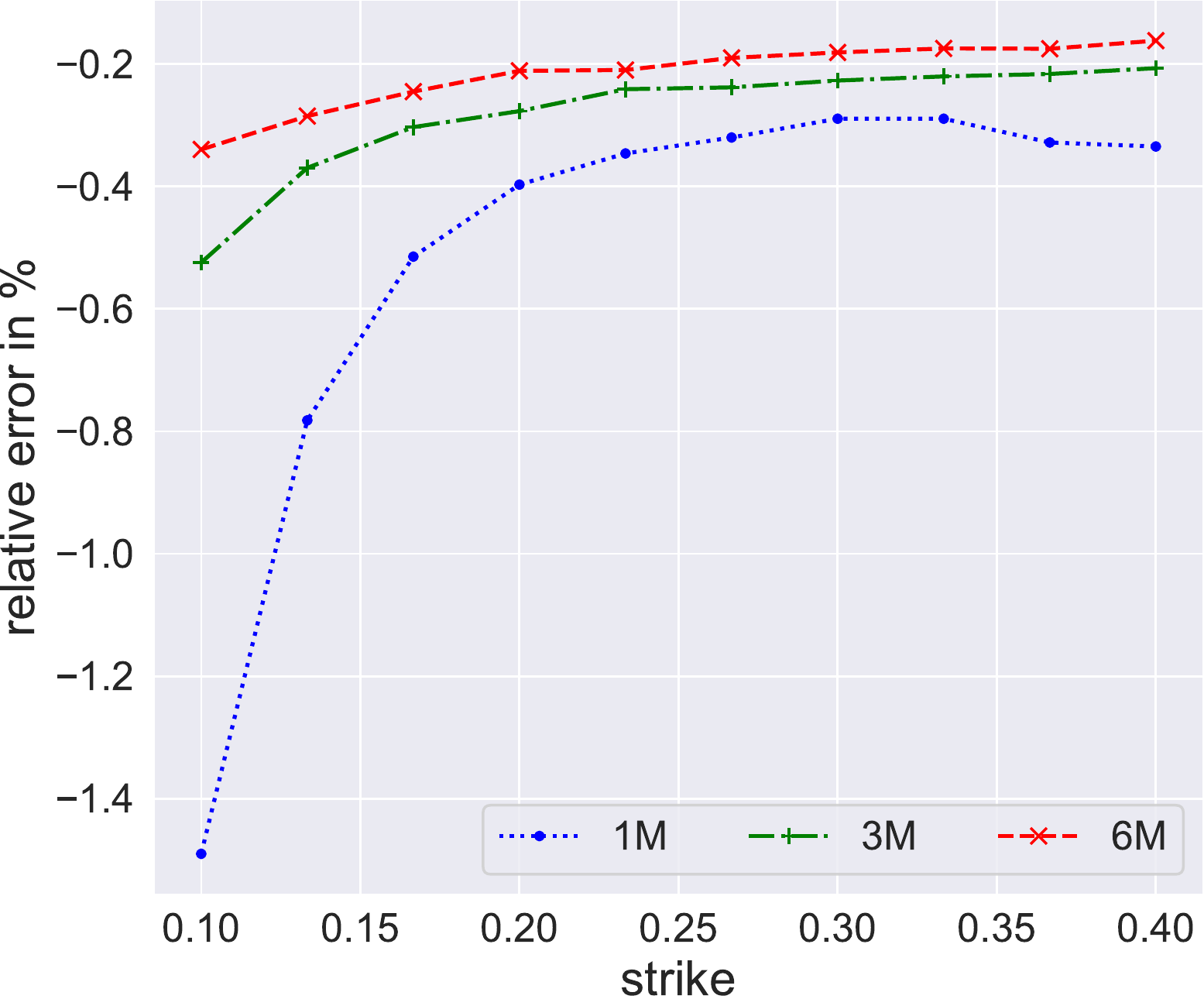}
\caption{ 
\emph{Left:} $\VIX$ smiles in the rough Bergomi model for $T=1,3,6$ months. 
\emph{Right:}
relative error in $\%$. 
The model parameters are $\xi_{0} = 0.235^{2},\eta=1,H=0.1$. 
}
\label{fig:implied:vol:plain} 
\end{figure}

The qualitative behavior of the VIX smiles that we observe for the standard Bergomi model is very similar, as already reported in \cite{bergomi2005smile}.
The $\VIX$ smiles we obtained from the reference prices and the price approximations in the Bergomi model are therefore not reported, being very similar 
to the  implied volatility smiles in Figure \ref{fig:implied:vol:plain} (the approximation formulas still being very accurate when compared to the reference implied volatilities).

\section{Mixed exponential models} 
\label{sec:mixed_model}

A class of models more general than \eqref{eq:sde:forward:variance}, and capable of appropriately capturing the behavior of market VIX smiles, was introduced in Bergomi \cite{bergomi2008smile}, who observed that a simple yet efficient way to twist the distribution of forward variances
is to replace the exponential process $\xi_t^u =  \xi_{0}^{u} \, e^{Y_t^u}$  stemming from the dynamics \eqref{eq:sde:forward:variance} with a convex combination of two exponential functions.
This observation yields the extended model family
\begin{equation}
\xi_{T}^{u}
=
\xi_{0}^{u}
\biggl[
\lambda \, \cE \Bigl( \int_{0}^{T}K_{1}^{u}(t) \dd W_{t} \Bigr)
+
(1 - \lambda) \cE \Bigl( \int_{0}^{T}K_{2}^{u}(t) \dd W_{t} \Bigr) \biggl],
\label{eq:def:model:mixed}
\end{equation}
which we refer to as ``mixed exponential models'' or ``mixed Bergomi models''.
In \eqref{eq:def:model:mixed}, $\cE\left(X\right)$
is a shorthand notation for $e^{X-\frac{1}{2}\Var\left(X\right)}$, and $\lambda\in\left[0,1\right]$ is a mixing parameter allowing to tune the relative importance of each exponential term.

\begin{remark}[Mixing log-normal models]
When $K_{i}^{u}(t)=\omega_{i} \, e^{-k(u-t)}$, $i\in\{1,2\}$, \eqref{eq:def:model:mixed} corresponds to the mixed one-factor Bergomi model introduced in \citet{bergomi2008smile}, and when $K_{i}^{u}=\eta_{i}(u-t)^{H-\frac{1}{2}}$, \eqref{eq:def:model:mixed} yields the mixed rough Bergomi model, introduced simultaneously by \citet{demarco2018pres} and \citet{guyon2018pres}.
It is of course possible to consider representations that encompass both the standard Bergomi and the rough Bergomi models, using kernels of the form $K^u(t) \propto e^{-k(u-t)}(u-t)^{H-\frac{1}{2}}$, as done in \cite[section 3.2]{alos2022smile},
\end{remark}

The squared VIX is of course still defined by integrated instantaneous forward variances as in \eqref{eq:def:VIX2};
under \eqref{eq:def:model:mixed}, we see that $\VIX_{T}^{2} $ is given by a convex combination of integral means of the form \eqref{eq:def:true:quantity}, that is
\be \label{e:true_VIX_mixed}
\VIX_{T}^{2} 
= \lambda \, \VIX_{T,1}^{2} + (1-\lambda) \VIX_{T,2}^{2}
:= \lambda \frac 1 \Delta \int_T^{T+\Delta} \xi_0^u \, e^{Y_{T,1}^{u}} \dd u + (1 - \lambda) \frac 1 \Delta \int_T^{T+\Delta} \xi_0^u \, e^{Y_{T,2}^{u}} \dd u \,,
\ee
where  
\begin{equation}\label{eq:def:Y_T_j}
Y_{T, j}^{u} 
:= -\frac{1}{2}\int_{0}^{T}K_{j}^{u}\left(t\right)^{2}\dd t+\int_{0}^{T}K_{j}^{u}\left(t\right)\dd W_{t}
\qquad j\in\{1,2\} \,.
\end{equation}
Following the approach of section \ref{sec:vix:expansion:}, we approximate each integral mean with a log-normal random random variable, so that  $\VIX_{T}^{2}$ is eventually approximated by
\begin{align*}
\VIX_{T,{\rm P}}^{2} 
& =
\lambda \, \nu\left(\xidot0\right)
e^{\nu_{0}\left(Y_{T,1}^{\cdot}\right)}+\left(1-\lambda\right)\nu\left(\xidot0\right)e^{\nu_{0}\left(Y_{T,2}^{\cdot}\right)}
\\
 & =:\lambda \,  \nu\left(\xidot0\right)
 \VIX_{T,{\rm P},1}^{2}
 +
 (1-\lambda)
 \nu\left(\xidot0\right)
 \VIX_{T,{\rm P},2}^{2} \,.
\end{align*}
The overall proxy $\VIX_{T,{\rm P}}^{2}$ is therefore a convex combination of correlated log-normal random variables. For $j\in\{1,2\},$
we have 
\begin{equation}
\ln\left(\VIX_{T,{\rm P},j}^{2}\right)
\egl \cN\left( \mu_{{\rm P},j},\,\sigma_{{\rm P},j}^{2}\right),\label{eq:mixed:model:mean:var}
\end{equation}
where 
$\mu_{{\rm P},j}:=-\frac{1}{2}\int_{0}^{T}\nu_{0}\bigl(K_{j}^{\cdot}(t)^{2}\bigr)\dd t$, 
$\sigma_{{\rm P},j}^{2}:=\int_{0}^{T}\nu_{0}\bigl(K_{j}^{\cdot}(t)\bigr)^{2} \dd t$.

\subsection{Price expansion}

Let us define the analogous of the coefficients  $\gamma$  in \eqref{eq:def:gamma:coeffs}: for $j\in\{1,2\}$,
\begin{align*}
\gamma_{1,j} & :=\frac{1}{8}\int_{\cA}\Bigl(\int_{0}^{T}\bigl[K_{j}^{u}(t)^{2}-\nu_{0}(K_{j}^{\cdot}(t)^{2})\bigr]\dd t\Bigr)^{2}\nu_{0}(\dd u)+\frac{1}{2}\int_{\cA}\Bigl(\int_{0}^{T}\bigl[K_{j}^{u}(t)-\nu_{0}(K_{j}^{\cdot}(t))\bigr]^{2}\dd t\Bigr)\nu_{0}(\dd u),\\
\gamma_{2,j} & :=-\frac{1}{2}\int_{\cA}\Bigl(\int_{0}^{T}\nu_{0}(K_{j}^{\cdot}(t))\bigl[K_{j}^{u}(t)-\nu_{0}(K_{j}^{\cdot}(t))\bigr]\dd t\Bigr)\Bigl(\int_{0}^{T}\bigl[K_{j}^{u}(t)^{2}-\nu_{0}(K_{j}^{\cdot}(t)^{2})\bigr]\dd t\Bigr)\nu_{0}(\dd u),\\
\gamma_{3,j} & :=\frac{1}{2}\int_{\cA}\Bigl(\int_{0}^{T}\nu_{0}(K_{j}^{\cdot}(t))\bigl[K_{j}^{u}(t)-\nu_{0}(K_{j}^{\cdot}(t))\bigr]\dd t\Bigr)^{2}\nu_{0}(\dd u).
\end{align*}

\begin{theorem} \label{thm:expansion:mixed:model}
Let $\varphi \in \mathcal{C}_{b}^{2}$. In the mixed rough Bergomi model obtained setting $K_{i}^{u}=\eta_{i}(u-t)^{H-\frac{1}{2}}$ in \eqref{eq:def:model:mixed}, the price of an option on $\VIX_{T}^{2}$ with payoff $\varphi$ is given by 
\begin{equation} \label{eq:expansion:mixed:model}
\bE\left[\varphi\left(\VIX_{T}^{2}\right)\right]
=
\bE\left[\varphi\left(\VIX_{T,{\rm P}}^{2}\right)\right]
+
\sum_{i=1}^{3}\sum_{j=1}^{2}\gamma_{i,j}P_{i,j}+\mathscr{E}_{\varphi},
\end{equation}
where $\mathscr{E}_{\varphi}$ is an error term satisfying $\left|\mathscr{E}_{\varphi}\right|\leq_{c}\Delta^{3(d_{1}\wedge\frac{d_{2}}{2})}$ 
with $d_1,d_2$  given in \eqref{eq:assu:drift}-\eqref{eq:assu:diffusion}, and
\begin{align}
\bE\left[\varphi\left(\VIX_{T,{\rm P}}^{2}\right)\right]
&
=\mathbb{E}\left[\varphi
\left(
\nu\left(\xidot0\right) \left[
\lambda \, e^{\mu_{{\rm P,}1}+\sigma_{{\rm P},1}Z}+(1-\lambda)e^{\mu_{{\rm P,}2}+\sigma_{{\rm P},2}Z}
\right]
\right)\right],
\label{def:price_mixed_proxy}
\\
P_{i,j}
&=
\partial_{\ve}^{i-1}\left.\mathbb{E}\left[\Psi_{j}\left(\mu_{{\rm P},j}+\sigma_{{\rm P,}j}Z+\ve\right)\right]\right|_{\ve=0},\quad  i\in\{1,2,3\},\ j\in\{1,2\},
\nonumber
\\ 
\Psi_{1}\left(x\right) 
&
=\partial_{y}\! \left.\varphi\left(\nu\left(\xidot0\right)\left[\lambda \, e^{x + y}
+
\left(1-\lambda\right)e^{
	\frac{\eta_{2}}{2}\left(\eta_{1}-\eta_{2}\right)  \int_{0}^{T}\nu_{0}\left(K_{0}^{\cdot}\left(t\right)^{2}\right)\dd t
	+
	\frac{\eta_{2}}{\eta_{1}} x
}\right]\right)\right|_{y = 0},
\label{def:psi1}
\\
\Psi_{2}\left(x\right) 
& 
=
\partial_{y} \! \left.
\varphi\left(\nu\left(\xidot0\right)\left[\lambda \, e^{
	\frac{\eta_{1}}{2}\left(\eta_{2}-\eta_{1}\right)\int_{0}^{T}\nu_{0}\left(K_{0}^{\cdot}\left(t\right)^{2}\right)\dd t
	+
	\frac{\eta_{1}}{\eta_{2}}x}+\left(1-\lambda\right)e^{x + y}\right]\right)
\right|_{y = 0}.
\label{def:psi2}
\\
K_{0}^{u}(t)
&
=(u-t)^{H-\frac{1}{2}}.
\nonumber
\end{align}
A similar expansion  holds for the mixed standard 
Bergomi model, taking $K_{0}^{u}(t) = e^{-k(u-t)}$ and replacing $\eta_i$ 
with $\omega_i$ for $i\in \{1,2\}$ in \eqref{def:psi1} and \eqref{def:psi2}.
\end{theorem}

\begin{remark}
We note that the form of  \eqref{def:price_mixed_proxy}-\eqref{def:psi1}-\eqref{def:psi2} is specific to the mixed one-factor Bergomi model where $K_{i}^{u}(t)=\omega_{i}e^{-k(u-t)}$ (same value of $k$ for the two kernels) and to the rough Bergomi model where $K_{i}^{u}=\eta_{i}(u-t)^{H-\frac{1}{2}}$ (same value of	$H$). In these cases, $Y_{T,1}^{u}$ and $Y_{T,2}^{u}$ are linearly dependent Gaussian variables (${\rm Corr}(Y_{T,1}^{u},Y_{T,2}^{u})=1$) and the VIX proxy is a function of a single Gaussian random variable,
	\begin{equation*}
		\VIX_{T,{\rm P}}^{2}
		\egl
		\nu(\xidot0)
		\left[
		\lambda \, 	e^{\mu_{{\rm P},1}+\sigma_{{\rm P},1}Z}+\left(1-\lambda\right)e^{\mu_{{\rm P,}2}+\sigma_{{\rm P,}2}Z}\right],
		\qquad Z\egl\cN\left(0,1\right),
		\label{eq:mixed:proxy:distribution}
	\end{equation*}
so that all the expressions in Theorem \ref{thm:expansion:mixed:model} can be evaluated with efficient one-dimensional Gaussian quadratures, as we explain in detail in the following section.
\end{remark}

Theorem \ref{thm:expansion:mixed:model} could be extended to non-smooth payoffs, using  similar arguments to the proof of Theorem \ref{thm:expansion:plain:model}. 
Leaving this rather long analysis for future work, we prove the current statement of Theorem \ref{thm:expansion:mixed:model} for smooth payoffs in section \ref{sec:proofs}, while still providing numerical tests for futures, call and put payoffs in the next section.

\subsection{Numerical tests for option price formulas and implied volatilities}
\label{sec:vix:expansion:numerical:tests:mixed:rough:bergomi}

Reference prices in the mixed models are still computed as described in Remark \ref{rem:benchmark:prices}: in the mixed rough Bergomi model, we discretize the	
$\VIX_{T,j}^2$ in \eqref{e:true_VIX_mixed} for $j \in \{1, 2\}$ with a rectangle scheme and simulate exactly the discretized variable, while in the mixed standard Bergomi model, we exploit the Markovian representation in  Remark \ref{rem:benchmark:prices} for each term $\VIX_{T,j}^2$, $j \in \{1, 2\}$, and apply a two-dimensional deterministic quadrature with respect to the parameter $u$ and to the space dimension.

All the numerical tests were performed on a MacBook Air laptop (M1, 2020) with 8GB of memory using the programming language Python 3.9.9.

\paragraph{Computation of the $P_{i,j}$ in Theorem \ref{thm:expansion:mixed:model} for $i \in \{1,2,3\}$ and $j \in \{1,2\}$.}
Let us drop the subscript $\mathrm P$ and denote $\mu_j = \mu_{{\rm P},j}$, $\sigma_j = \sigma_{{\rm P},j}$.
Recalling that 
\[
\mathbb{E} 
\left[
\varphi
\left(
{\rm VIX}_{T,{\rm P}}^{2}
\right)
\right]
=
\mathbb{E}\left[\varphi\left(\nu(\xi_{0}^{\cdot}) \bigl[ \lambda \,  e^{\mu_{1}+\sigma_{1}Z} 
+   
(1 - \lambda) e^{\mu_{2}+\sigma_{2}Z} 
\bigr]
\right)\right],
\]
where $Z\sim\mathcal{N}(0,1)$, we can use a one-dimensional Gauss--Hermite quadrature (with $80$ nodes in our tests) in order to evaluate the option price over the proxy $\mathbb{E}\left[\varphi\left({\rm VIX}_{T,{\rm P}}^{2}\right)\right]$.
Recall that VIX futures correspond to $\varphi(x) = \sqrt x$ and VIX call options to $\varphi(x) = (\sqrt x - \kappa)^+$.
The terms $P_{i,1}$  and $P_{i,2}$ for $i \in \{1,2,3\}$ are given by derivatives of the expectation above with respect to a parameter, and therefore they can be recast under the form of expectations using the likelihood method (derivation of the density function).
Consider the terms $P_{i,1}$ for $i \in \{1,2,3\}$. 
We have $P_{1,1}  =\mathbb{E}\left[\Psi_{1}\left(\mu_{1}+\sigma_{1}Z\right)\right]$, where the function $\Psi_1$ is explicity given in \eqref{def:psi1}. Since
\begin{equation*}
	f_{1}(\varepsilon):=\mathbb{E}\left[\Psi_{1}\left(\mu_{1}+\sigma_{1}Z+\varepsilon\right)\right] 
	=\frac{1}{\sigma_{1}}\int_{\mathbb{R}}\Psi_{1}\left(y\right)\frac{\exp(-\frac{1}{2\sigma_{1}^{2}}(y-\mu_{1}-\varepsilon)^{2})}{\sqrt{2\pi}}{\rm d}y \,,
\end{equation*}
we have
\begin{equation*}
	P_{2,1} = f_{1}^{\prime}(0) 
	=\frac{1}{\sigma_{1}}\int_{\mathbb{R}}\Psi_{1}\left(y\right)\frac{y-\mu_{1}}{\sigma_{1}}\frac{\exp(-\frac{1}{2\sigma_{1}^{2}}(y-\mu_{1})^{2})}{\sqrt{2\pi}}\frac{{\rm d}y}{\sigma_{1}}\\
	=\frac{1}{\sigma_{1}}\mathbb{E}\left[Z\Psi_{1}\left(\mu_{1}+\sigma_{1}Z\right)\right],
\end{equation*}
and 
\begin{equation*}
	P_{3,1} = f_{1}^{\prime\prime}(0) 
	=\frac{1}{\sigma_{1}^{2}}\mathbb{E}\left[\left(Z^{2}-1\right)\Psi_{1}\left(\mu_{1}+\sigma_{1}Z\right)\right].
\end{equation*}
Consequently, we can again use a one-dimensional Gauss--Hermite quadrature to evaluate the $P_{i,1}$ for $i \in \{1,2,3\}$.  The terms $P_{i, 2}$, $i \in \{1,2,3\}$, are treated analogously.

In terms of complexity for the pricing procedure, we have the replaced Monte Carlo simulation of the VIX discretization scheme in the mixed rough Bergomi model, resp.\ the two-dimensional quadrature in the mixed standard Bergomi model, with one-dimensional Gaussian quadratures.

\begin{remark}
Call and put option prices on the squared VIX proxy $\VIX_{T,{\rm P}}^{2}$, corresponding to $\varphi(x) = ( x - \kappa)^+$ and $\varphi(x) = (\kappa - x)^+$ in \eqref{def:price_mixed_proxy}, admit explicit expressions in terms  of Black--Scholes formulas, provided one evaluates the point $F^{-1}(\kappa)$, $F^{-1}$ being the inverse of the function $F:x\in\bR\to \nu(\xi_0^\cdot) \Bigl[\lambda \, e^{\mu_{1}+\sigma_{1}x}+\left(1-\lambda\right)e^{\mu_{2}+\sigma_{2}x}\Bigr]$. The evaluation of $F^{-1}(\kappa)$ can be performed with a simple root-finding procedure.
As a consequence, in the case of call and put options on $\VIX_{T}^{2}$, the whole expansion \eqref{eq:expansion:mixed:model} boils down to an explicit combination of Black--Scholes prices and greeks.
\end{remark}

\subsubsection{VIX implied volatility}

As done in the previous sections, we compare the reference $\VIX$ implied volatility with the approximate $\VIX$ implied volatility computed with our expansion in Theorem \ref{thm:expansion:mixed:model}. 
To test the approximation formulas on different $\VIX$ smiles, we consider two different parameter scenarios in the rough and standard Bergomi models.  We have considered options maturities equal to
$1,3,$ and $6$ months.

\paragraph{Implied volatility for the mixed rough Bergomi model.} 
Recall that the mixed rough Bergomi model is obtained by injecting $K_{i}^{u}=\eta_{i}(u-t)^{H-\frac{1}{2}}$, $i\in\{1,2\}$, in \eqref{eq:def:model:mixed}.
We set $\xi_0 = 0.235^{2}$ and $H=0.1$; the other model parameters can be found in Table \ref{tab:mixed_rBergomi_scenarios}.
We evaluate the reference option prices with $10^{6}$ Monte Carlo samples and $300$ discretization
points. 
\begin{table}[H] 
\begin{centering}
\begin{tabular}{|c|c|c|c|}
	\hline 
	Scenario  & 1-month $\VIX$ futures  & 3-month $\VIX$ futures  & 6-month $\VIX$ futures
	\tabularnewline
	\hline 
	\hline 
	$1$  & $0.218650\pm5\times10^{-6}$  & $0.206308\pm5\times10^{-6}$  & $0.196890\pm5\times10^{-6}$
	\tabularnewline
	\hline 
	$2$  & $0.229001\pm3\times10^{-6}$  & $0.224244\pm3\times10^{-6}$   & $0.220472\pm3\times10^{-6}$ 
	\tabularnewline
	\hline 
\end{tabular}
\begin{tabular}{|c|c|c|c|c|}
	\hline 
	Scenario  & $\eta_{1}$ & $\eta_{2}$ & $\lambda$ 
	\tabularnewline
	\hline 
	\hline 
	$1$  & $1.4$ & $0.7$ & $0.3$ 
	\tabularnewline
	\hline 
	$2$  & $0.9$ & $0$ & $0.6$ 
	\tabularnewline
	\hline 
\end{tabular}
\par\end{centering}
\caption{Term structure of $\VIX$ futures
and model parameters  for scenarios $1$ and $2$ in the mixed rough Bergomi model.}
	\label{tab:mixed_rBergomi_scenarios}
\end{table}

\begin{figure}[t]
\includegraphics[scale=0.5]{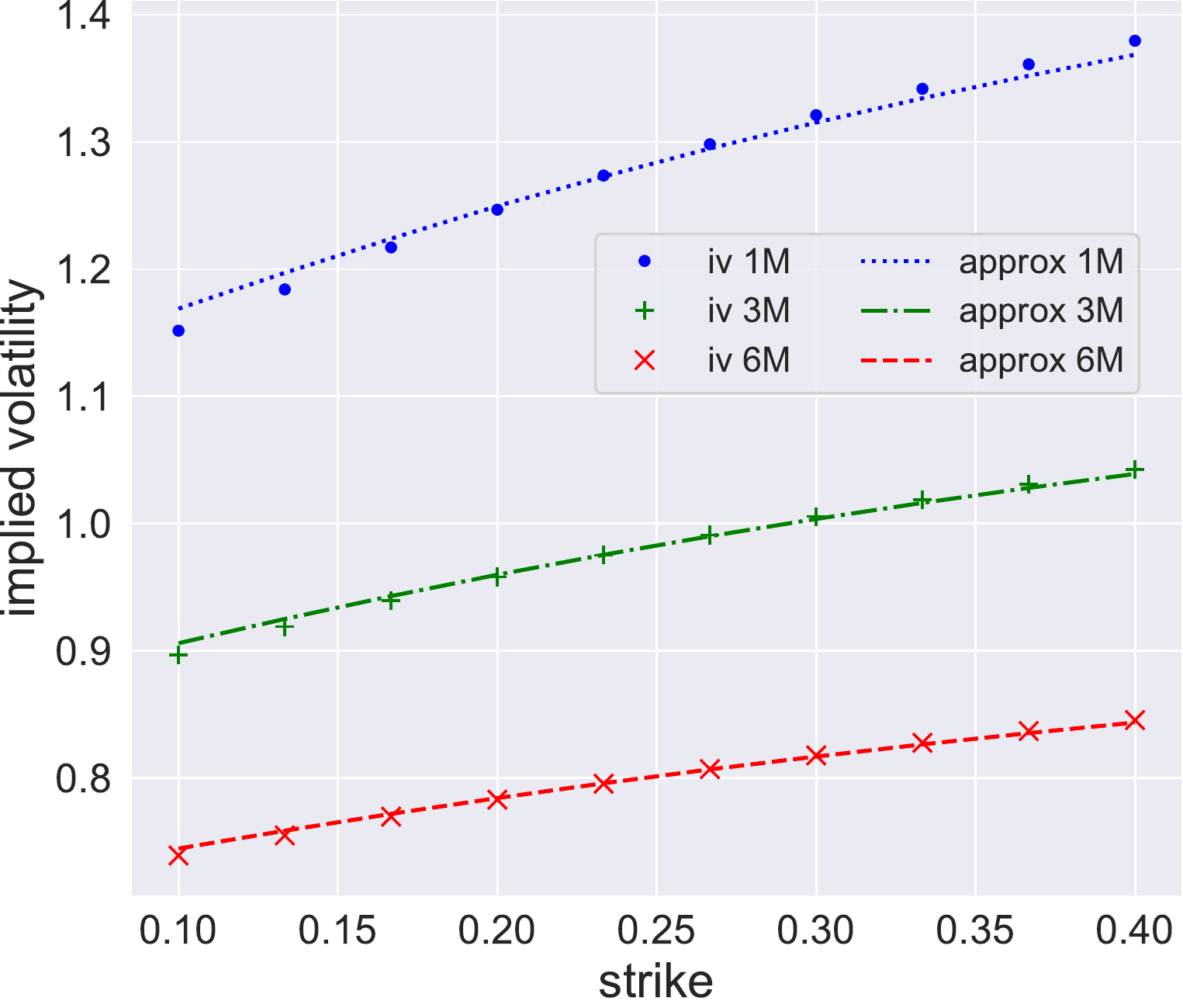}
\includegraphics[scale=0.5]{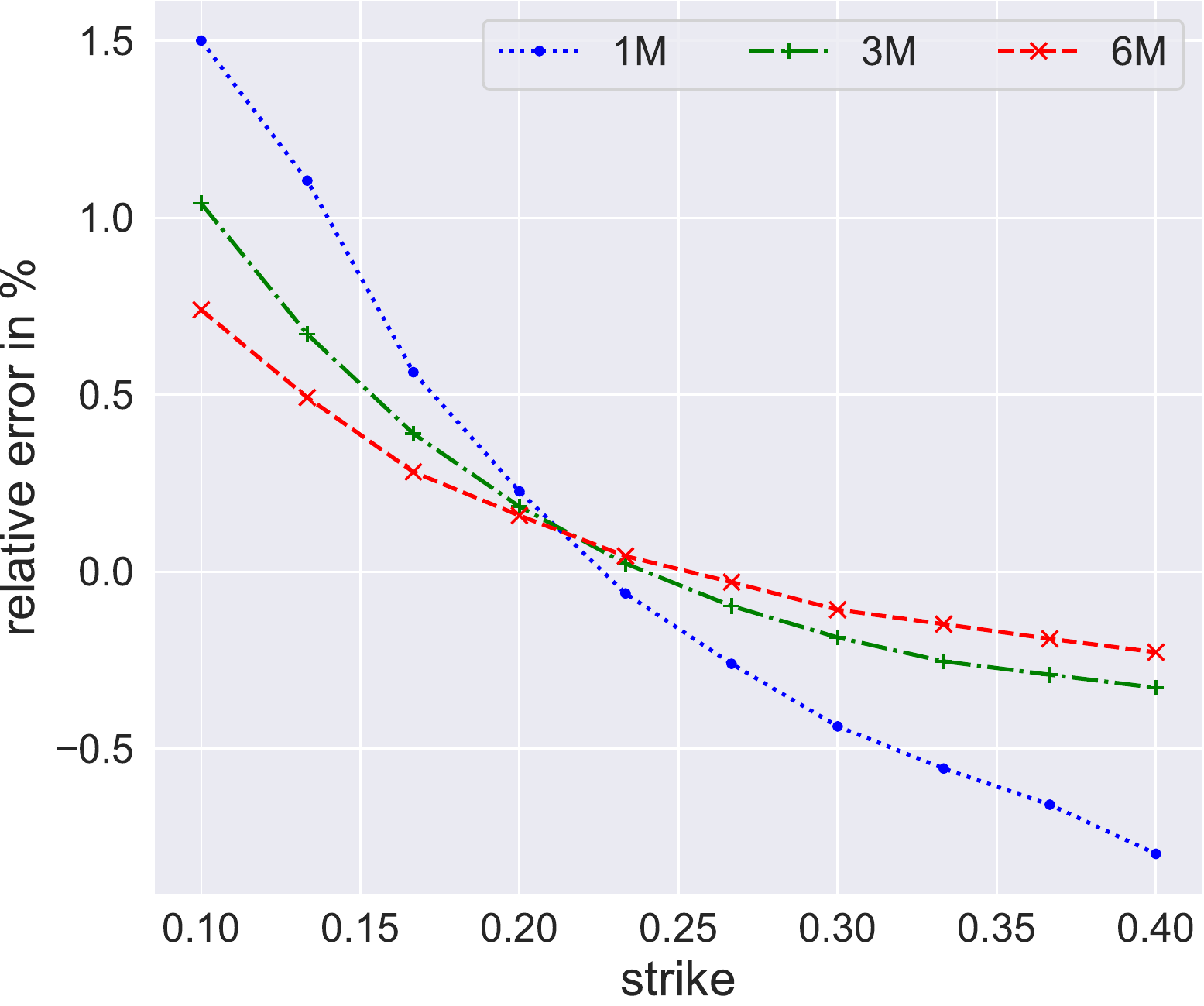}
\\
\includegraphics[scale=0.5]{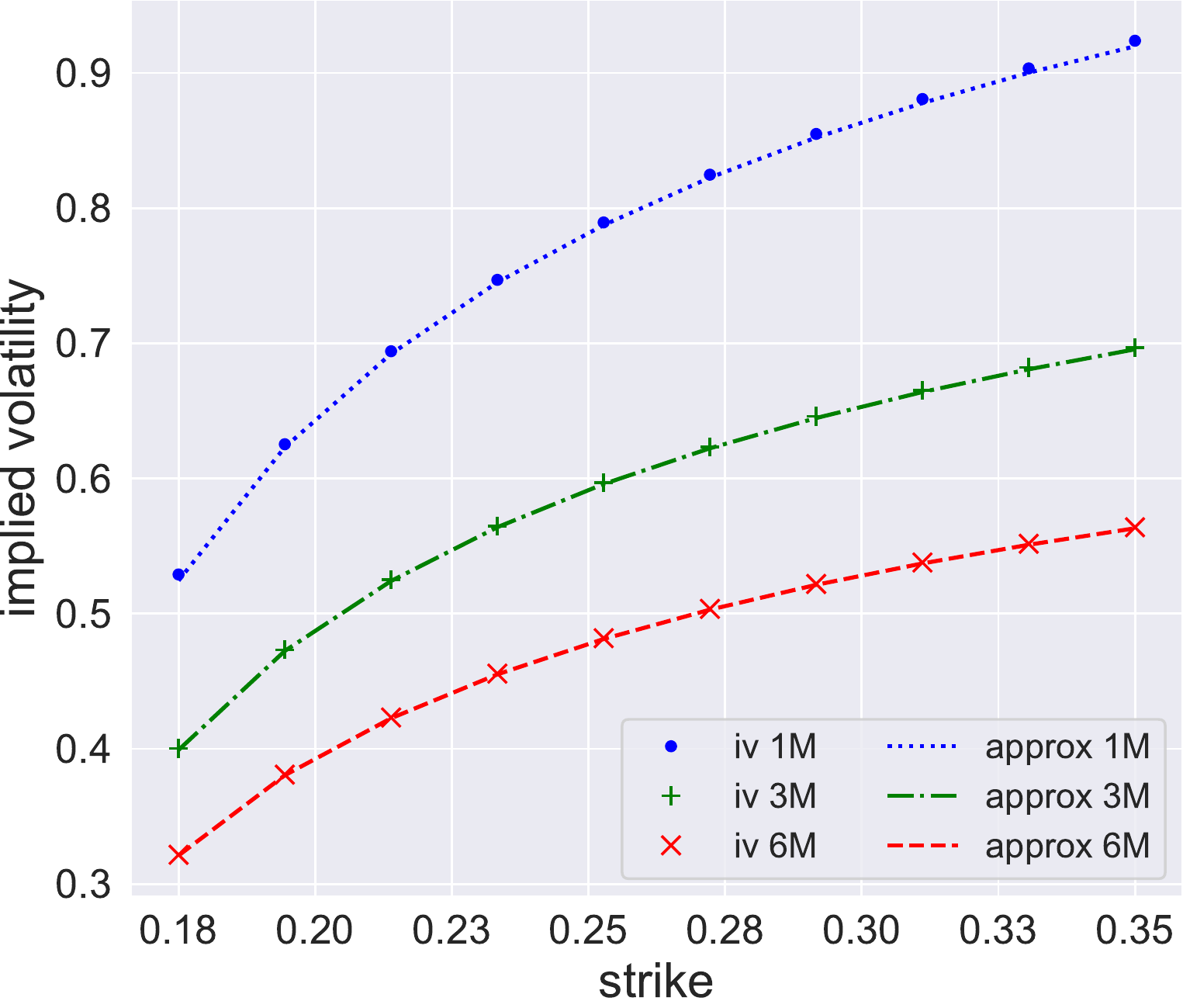}
\includegraphics[scale=0.5]{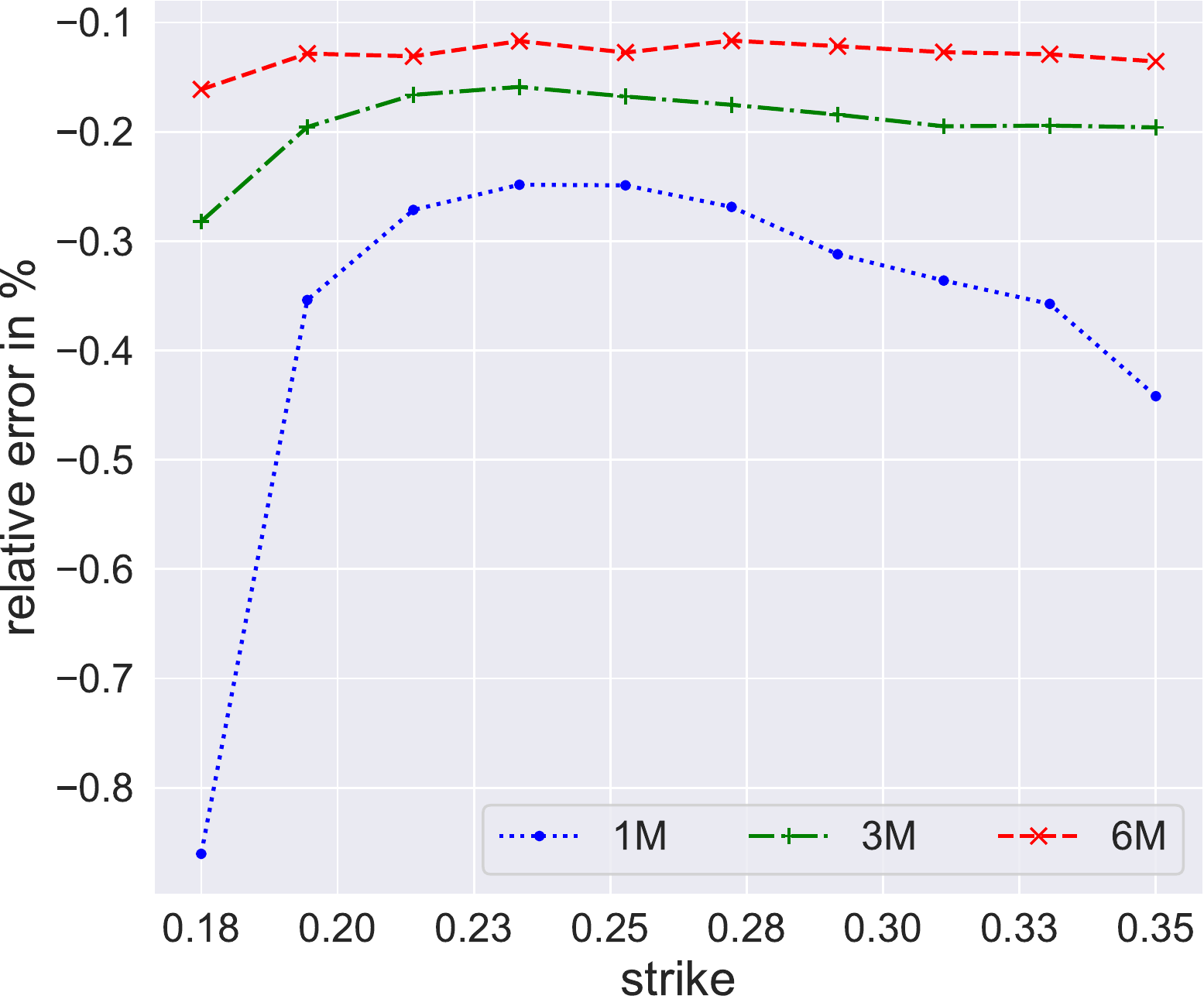}
\caption{$\VIX$ smiles in the mixed rough Bergomi model for $T=1,3,6$ months (left), and corresponding percentage relative errors between the reference  implied volatilities and their approximations  (right), for parameter scenarios $1$ (top figures) and $2$ (bottom figures).}
\end{figure}

In the mixed model as well, our approximation formula proves to be very accurate: the relative error is less than $1.6\%$ (in absolute value) for scenario $1$ and less than $0.9\%$ for scenario $2$.

\paragraph{Implied volatility for the mixed standard Bergomi model.}
We perform a similar numerical analysis for the mixed one-factor standard Bergomi model, obtained setting $K_{i}^{u}(t)=\omega_{i} \, e^{-k(u-t)}$, $i\in\{1,2\}$, in \eqref{eq:def:model:mixed}.
We set $\xi_0 = 0.2^2$ and $k = 1$. The other model parameters are given in Table \ref{tab:scenarios_mixed_Bergomi}.
\begin{table}[H]
	\begin{centering}
		\begin{tabular}{|c|c|c|c|}
			\hline 
			Scenario  & 1-month $\VIX$ futures  & 3-month $\VIX$ futures  & 6-month $\VIX$ futures
			\tabularnewline
			\hline 
			\hline 
			$3$  & $0.172764$  & $0.145976$  & $0.130503$ 
			\tabularnewline
			\hline 
			$4$  & $0.181527$  & $0.165480$  & $0.155141$ 
			\tabularnewline
			\hline 
\end{tabular}
\begin{tabular}{|c|c|c|c|c|}
		\hline 
		Scenario  & $\omega_{1}$ & $\omega_{2}$ & $\lambda$ 
		\tabularnewline
		\hline 
		\hline 
		$3$  & $0.5$ & $6$ & $0.3$ 
		\tabularnewline
		\hline 
		$4$  & $10$ & $2$ & $0.2$ 
		\tabularnewline
		\hline 
\end{tabular}
\par\end{centering}
\caption{Term structure of $\VIX$ futures
and model parameters for scenarios
$3$ and $4$ in the mixed standard Bergomi model.}
\label{tab:scenarios_mixed_Bergomi}
\end{table}

\begin{figure}[t]
	\includegraphics[scale=0.45]{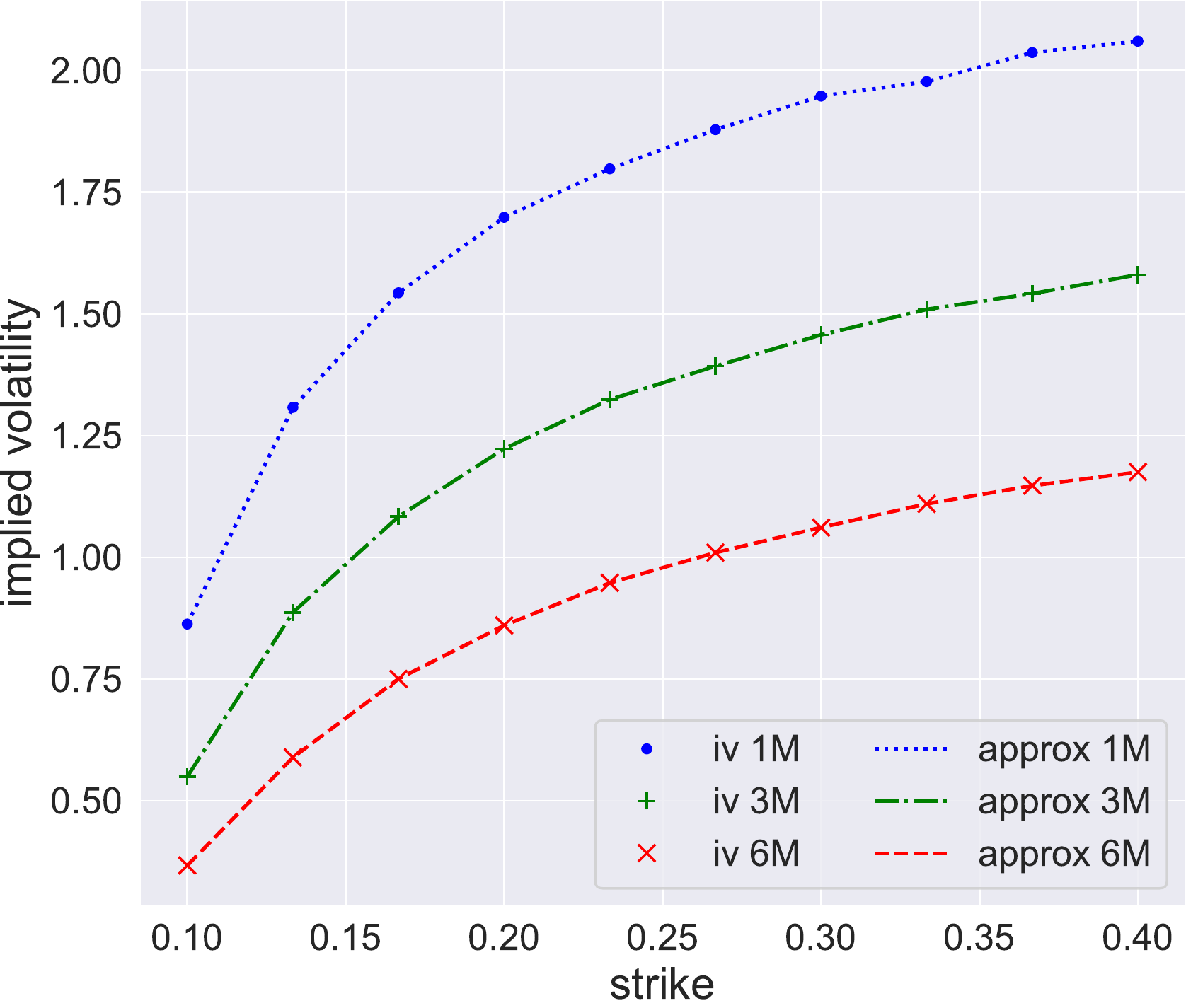}
	\includegraphics[scale=0.45]{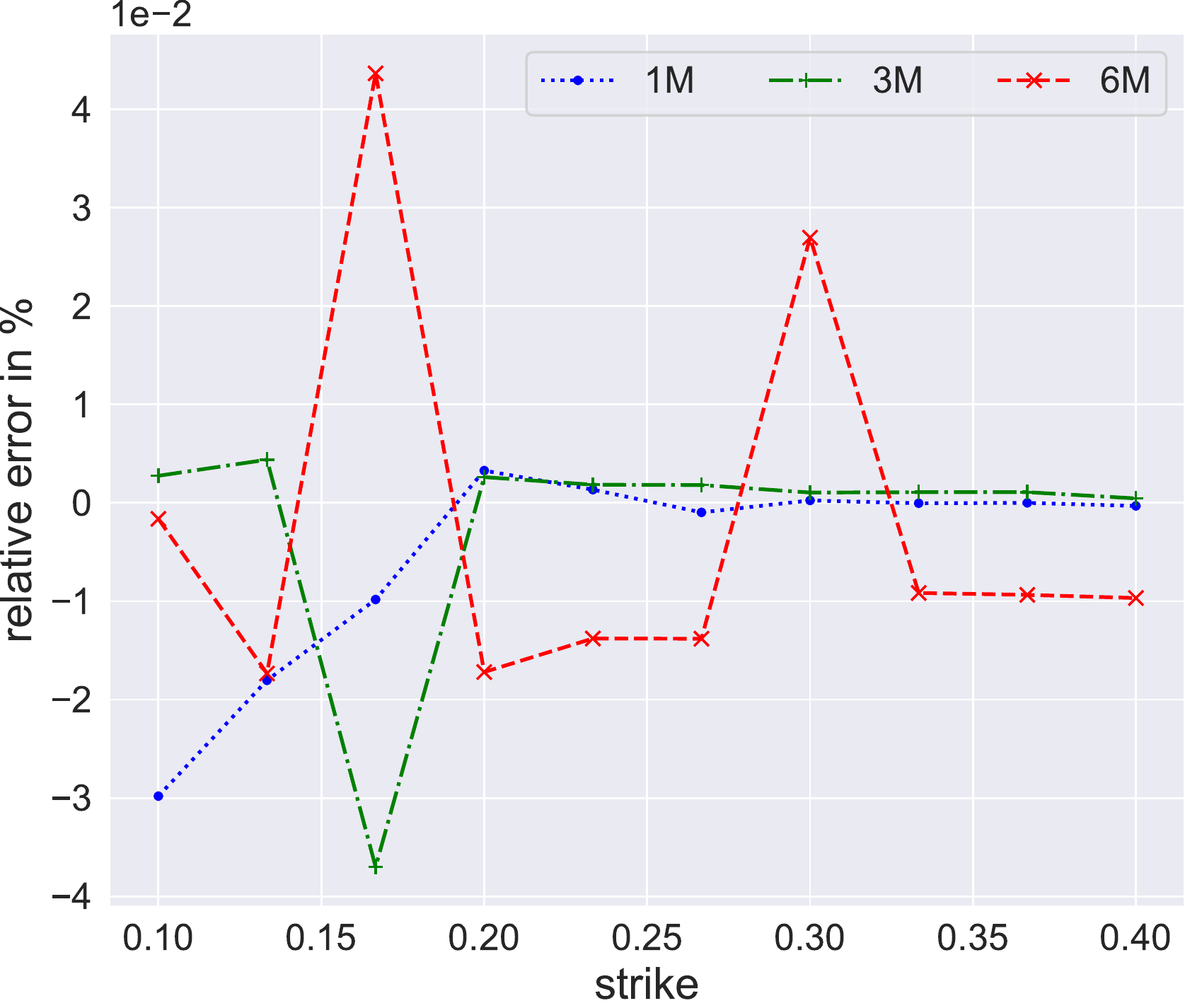}
	\\
	\includegraphics[scale=0.45]{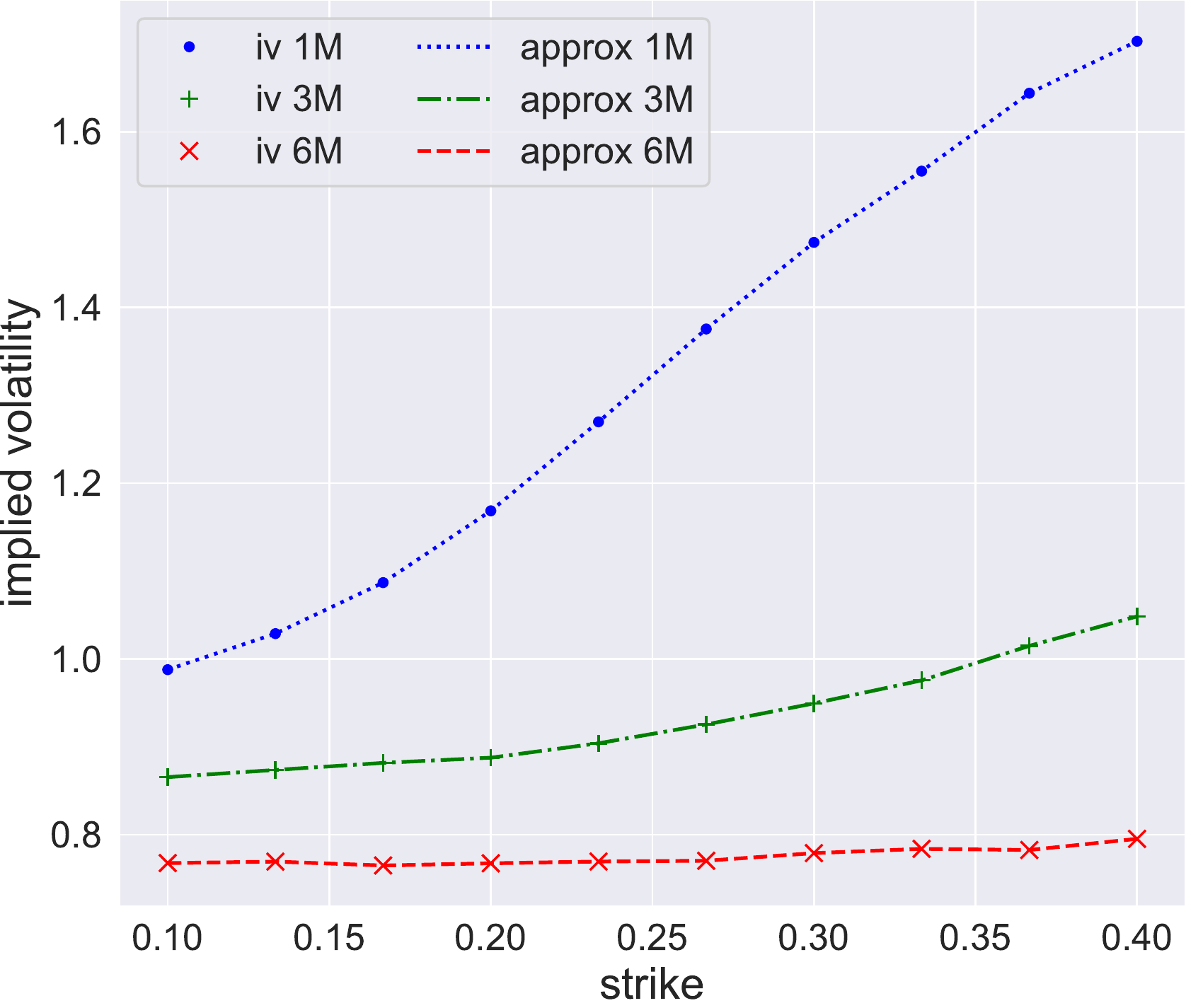}
	\includegraphics[scale=0.45]{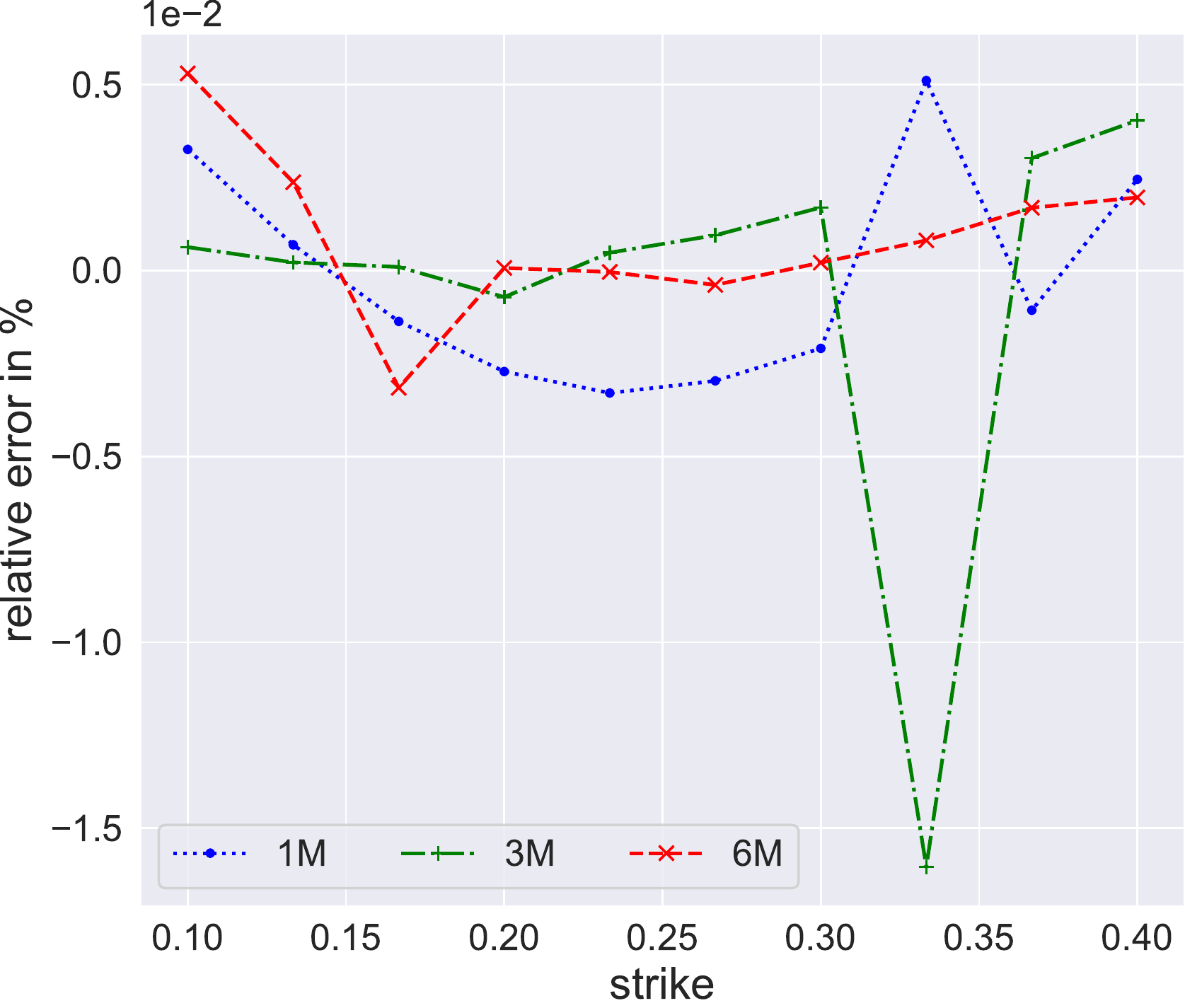}
	\caption{$\VIX$ smiles in the mixed standard Bergomi model for $T=1,3,6$ months (left figures), and relative errors between the reference implied volatilities and their approximations (right figures) for parameter scenarios $3$ (top figures) and $4$ (bottom figures).}
\end{figure}

Also, in this case, the approximation formula from Theorem \ref{thm:expansion:mixed:model}  turns out to be extremely accurate (relative errors are less than
$5 \times 10^{-2} \, \%$ for parameter scenario $3$ and less than $2 \times 10^{-2} \, \%$ for scenario $4$). 

\subsection{Calibration to VIX market data}
In this section, we perform a calibration test of the mixed rough Bergomi model and of the standard Bergomi model to market $\VIX$ smiles as of November 22, 2017, using our pricing formula \eqref{eq:expansion:mixed:model}.
We decide to set $H=0.1$ (resp.\ $k = 1$) for the mixed rough (resp.\ for the standard) Bergomi model and calibrate the other free parameters; of course it is also possible to calibrate the parameter $H$ (or $k$).

Let us describe the calibration procedure for the rough model.
The model contains the initial forward variance curve $\xi_{0}^{T}$ and the additional parameters $(\eta_1, \eta_2,\lambda)$. 
We can decide to use the variance curve to match the market  term structure of VIX futures exactly while using the other parameter to fit the smile of VIX options.
We consider the $n = 4$ shortest $\VIX$ futures quoted on the observation date, each associated to a maturity
$\left(T_{i}\right)_{i=1,\dots,n}$, with market values $(F_i)_{i=1,\dots,n}$.
Note that we can introduce a term structure also in the parameters $(\eta_1, \eta_2,\lambda)$ making them maturity-dependent and piece-wise constant between $T_i$ and $T_{i+1}$; the $i$-th VIX futures will be attached to its own parameter set $(\xi_{0}^{T_i},\eta_{1}^{T_i},\eta_{2}^{T_i},\lambda^{T_i})$.
We calibrate the model sequentially from the shortest to the largest futures maturity ; for each maturity, the procedure
follows two steps:
\begin{enumerate}
\item For given $(\eta_{1}^{T_i},\eta_{2}^{T_i},\lambda^{T_i})$, we set $\xi_{0}^{T_i}$ as the unique solution to $F_{i}=\PF \left(T_{i},\xi_{0}^{T_i},
\eta_{1}^{T_i},\eta_{2}^{T_i},\lambda^{T_i}\right)$
where $\PF(\cdot)$ corresponds to the approximate price
for VIX futures given by Theorem \ref{thm:expansion:mixed:model}. 
\item We then compute the $\VIX$ implied volatility smile using the value found for $\xi_{0}^{T_i}$ in the previous step and the values of the other parameters $(\eta_{1}^{T_i},\eta_{2}^{T_i},\lambda^{T_i})$. 
We evaluate the $L^{2}$ distance between the model implied volatility and the market implied volatility. Until we find a minimum, we go back to step $1.$ 
\end{enumerate}

\noindent In our tests, we used the function $\texttt{scipy.optimize.least\_squares}$ from the $\texttt{scipy}$ library \cite{virtanen2020scipy} in step 2.
The procedure is the same for the mixed standard Bergomi model, replacing $H$ with $k$ and $\eta_1, \eta_2$ with $\omega_1,\omega_2$.

\begin{figure}[t]
\includegraphics[scale=0.46]{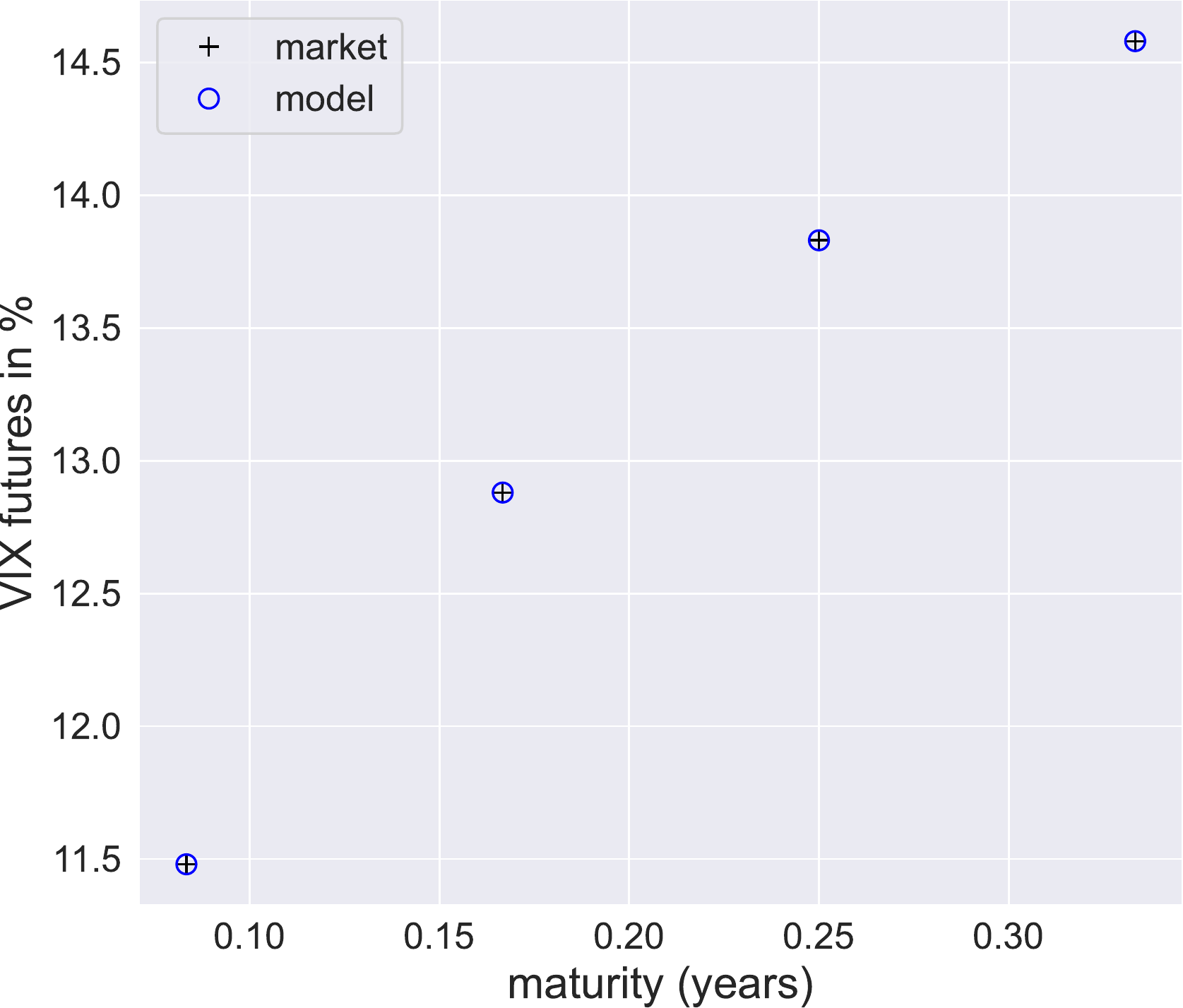}
\includegraphics[scale=0.46]{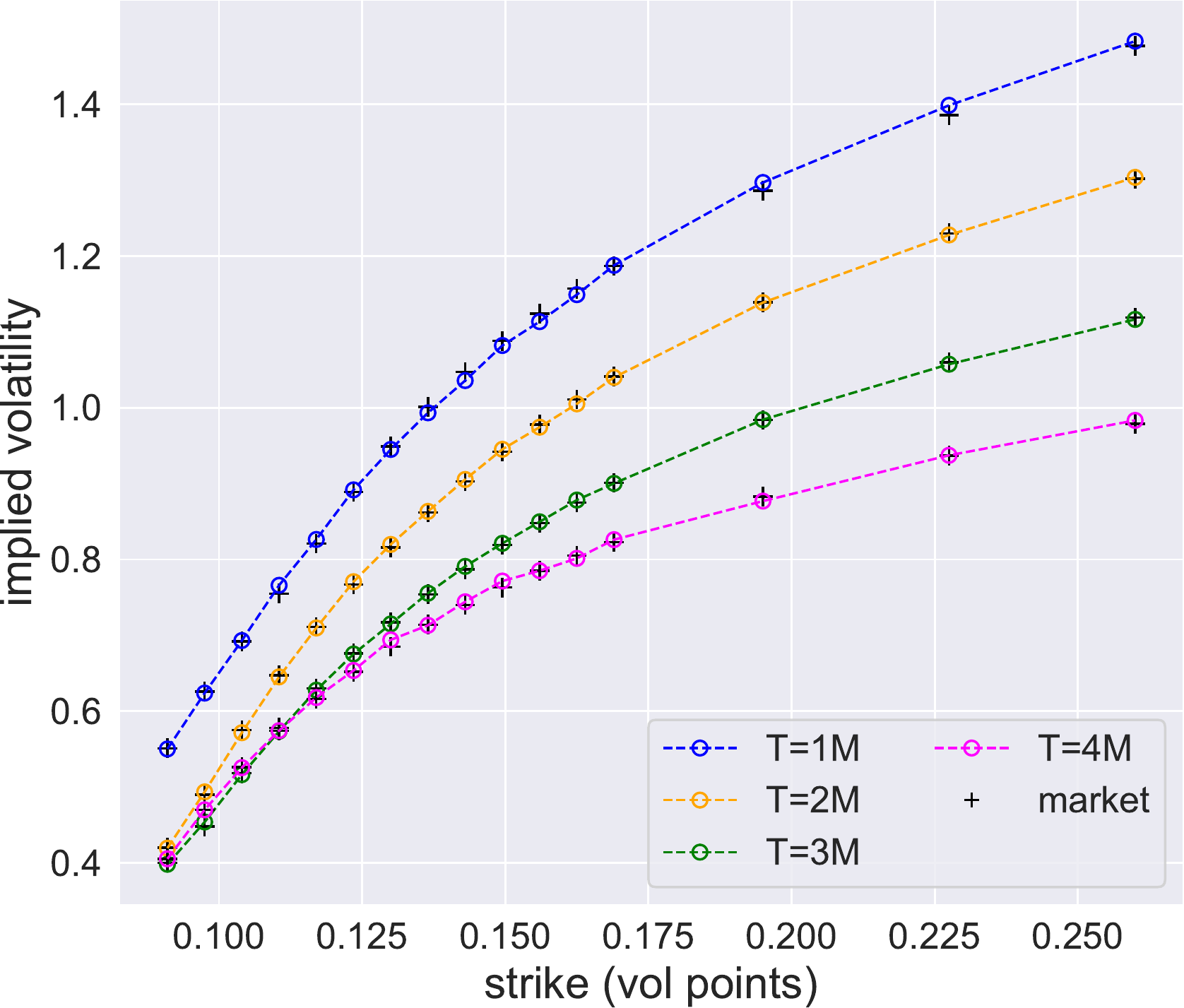}
\caption{\emph{Left}: term structure of $\VIX$ futures. \emph{Right}: market $\VIX$ smiles calibrated using our approximate price formulas in the mixed rough Bergomi model as of November 22, 2017.}
\label{fig:calib:vix:mixed}
\end{figure}

\begin{table}[h!]
\begin{centering}
\begin{tabular}{|c|c|c|c|c|}
\hline 
 & $T\in[\frac{1}{12},\frac{2}{12}[$  & $T\in[\frac{2}{12},\frac{3}{12}[$  & $T\in[\frac{3}{12},\frac{4}{12}[$  & $T\in[\frac{4}{12},\frac{5}{12}[$\tabularnewline
\hline 
\hline 
$\xi_{0}^{T}$  & $1.449 \times10^{-2}$  & $2.074 \times 10^{-2}$  & $2.543 \times 10^{-2}$  & $2.871 \times 10^{-2}$
\tabularnewline
\hline 
$\eta_{1}^{T}$  & $1.899$  & $1.887$  & $1.684$  & $1.410$
\tabularnewline
\hline 
$\eta_{2}^{T}$  & $0.1937$  & $0.1481$  & $0.1482$  & $0.1166$
\tabularnewline
\hline 
$\lambda^{T}$  & $0.3208$  & $0.4849$  & $0.5614$  & $0.6511$
\tabularnewline
\hline 
\end{tabular}
\par\end{centering}
\caption{Term structure for the calibrated parameters $\left(\xi_{0}^{T},\eta_{1}^{T},\eta_{2}^{T},\lambda^{T}\right)$ in the mixed rough Bergomi model as of November 22, 2017.}
\label{tab:calib:params:mixed}
\end{table}

\begin{figure}[t]
	\includegraphics[scale=0.46]{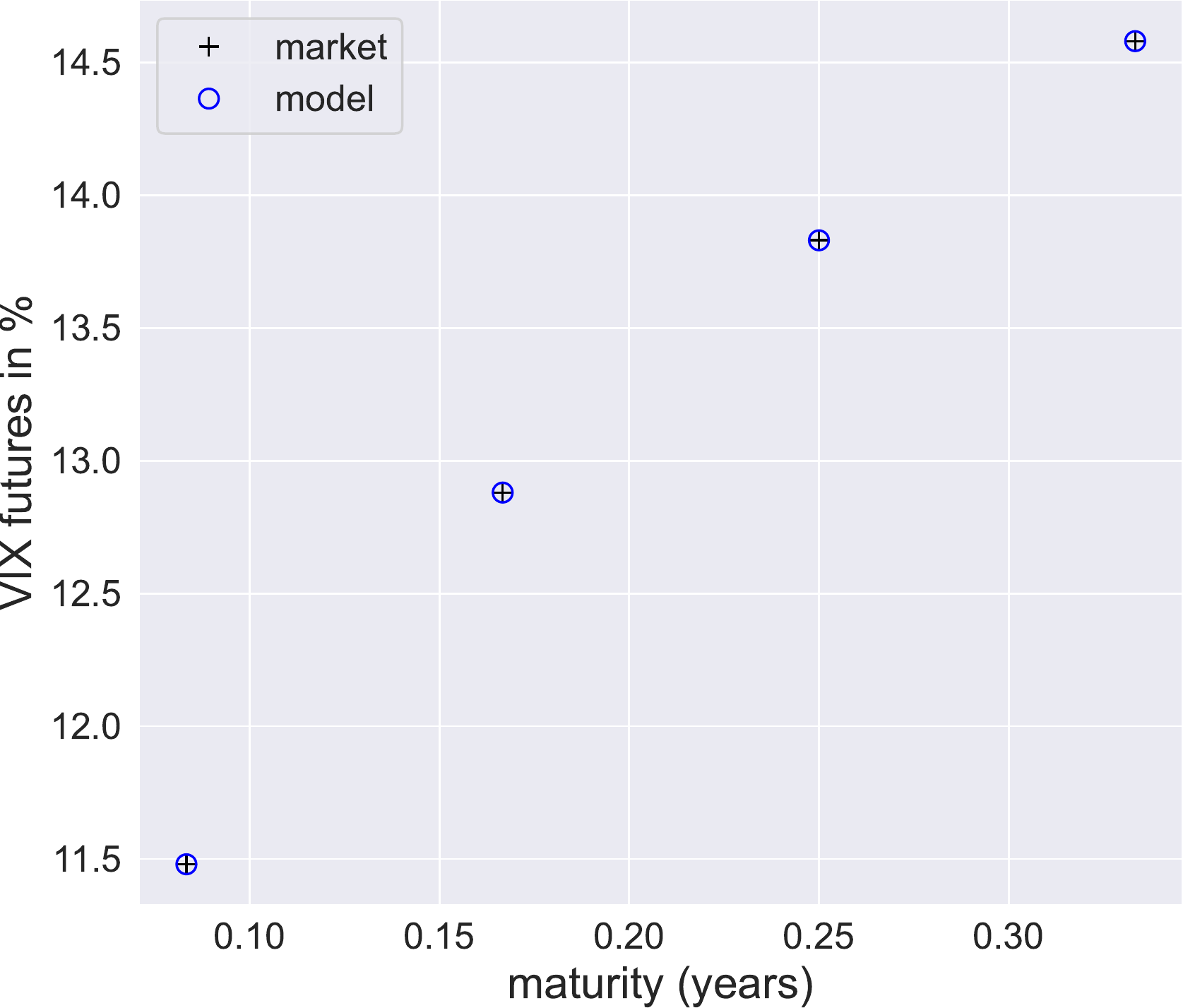}
	\includegraphics[scale=0.46]{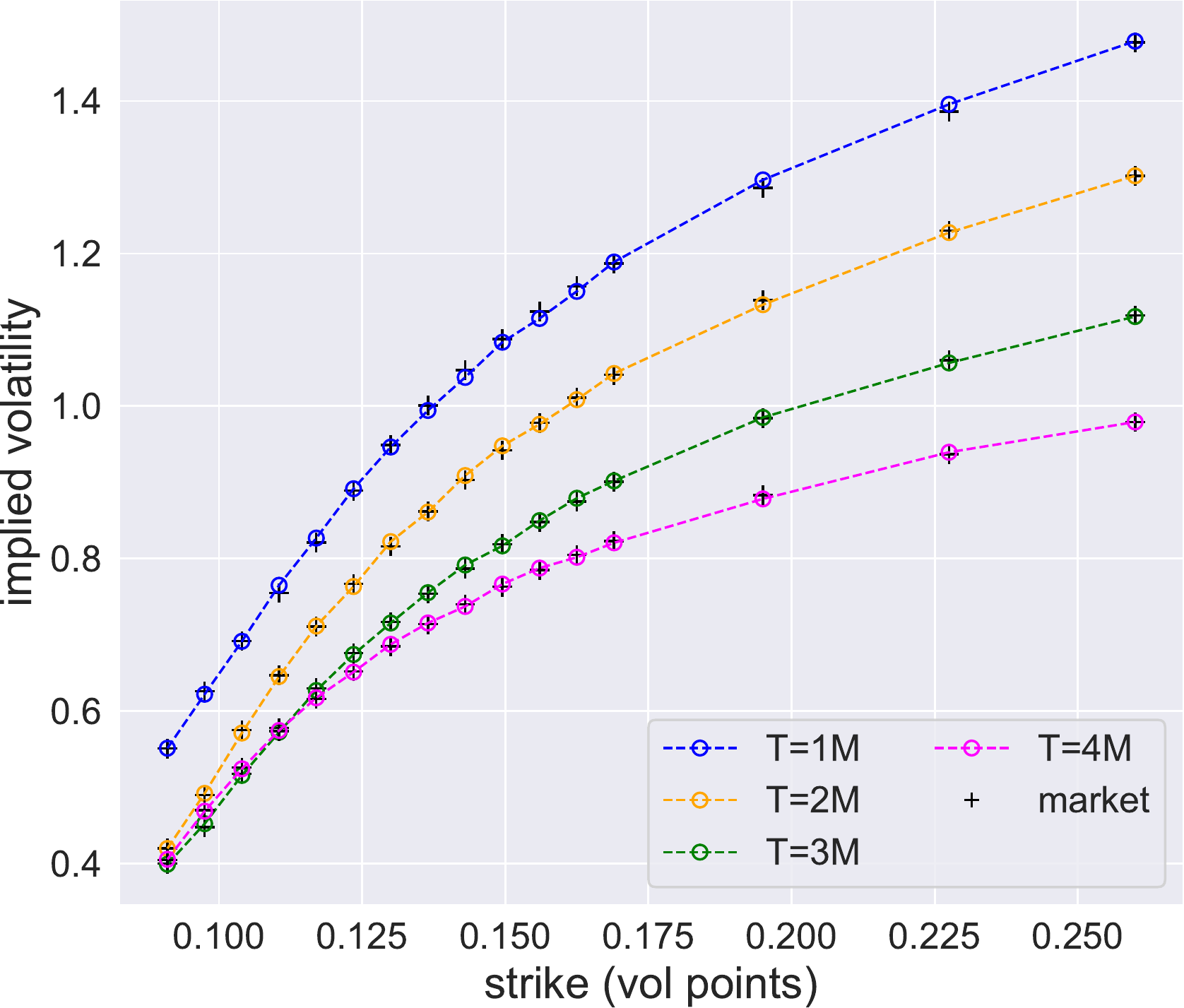}
	\caption{\emph{Left}: term structure of $\VIX$ futures. \emph{Right}: market $\VIX$ smiles calibrated using our approximate price formulas in the mixed standard Bergomi model as of November 22, 2017.}
	\label{fig:calib:vix:rough_mixed}
\end{figure}

\begin{table}[h!]
	\begin{centering}
		\begin{tabular}{|c|c|c|c|c|}
			\hline 
			& $T\in[\frac{1}{12},\frac{2}{12}[$  & $T\in[\frac{2}{12},\frac{3}{12}[$  & $T\in[\frac{3}{12},\frac{4}{12}[$  & $T\in[\frac{4}{12},\frac{5}{12}[$\tabularnewline
			\hline 
			\hline 
			$\xi_{0}^{T}$  & $1.445 \times 10^{-2}$  & $2.065 \times10^{-2}$  & $2.533 \times10^{-2}$  & $2.862 \times10^{-2}$\tabularnewline
			\hline 
			$\omega_{1}^{T}$  & $6.1970$  & $5.3118$  & $4.5273$  & $3.6860$ \tabularnewline
			\hline 
			$\omega_{2}^{T}$  & $0.6586$  & $0.4301$  & $0.4238$  & $0.3226$ \tabularnewline
			\hline 
			$\lambda^{T}$  & $0.3021$  & $0.4790$  & $0.5497$  & $0.6426$ \tabularnewline
			\hline 
		\end{tabular}
		\par\end{centering}
	\caption{Term structure for the calibrated parameters $\left(\xi_{0}^{T}, \omega_{1}^{T}, \omega_{2}^{T}, \lambda^{T}\right)$  in the mixed 
		standard Bergomi model as of November 22, 2017.}
	\label{tab:calib:params:mixed_rough}
\end{table}

The calibration of the two models based on our approximate price formula
proves to be very efficient and fast.
For the standard Bergomi model, in our tests the overall procedure is $3.5$ times faster than the calibration based on option pricing with a two-dimensional quadrature.
More precisely, starting from the initial guess $\xi_{0}^T = 2 \times 10^{-2}$, 
$\omega_1^T = 1.5$, $\omega_{2}^T = 0.5$, $\lambda^T = 0.5$, and using quadratures with $120$ nodes for the space and time integration,
the calibration of the whole implied volatility surface takes about $10$ seconds using our approximation formulas
as opposed to $36$ seconds when using the two-dimensional quadrature. 
In the mixed rough Bergomi model, the ratio of the two calibration times becomes much higher.
Starting from the same initial guess (replacing $\omega_{1}, \omega_{2}$ with $\eta_{1}, \eta_{2}$)
and using again  Gaussian quadrature with $120$ nodes in our approximation formulas, the calibration shown in Figure \ref{fig:calib:vix:mixed} also took about $10$ seconds.
On the contrary, we chose not to perform any calibration test with the Monte Carlo pricing procedure  described in Remark \ref{rem:benchmark:prices}, as the calibration time is likely to be over an hour in this case.
For reference, a single $\VIX$ futures pricing takes about $10$ seconds with $10^6$ Monte Carlo samples and $300$ discretization points.

\section{Conclusion\label{sec:conclusion}}

We have analyzed the accuracy of weak approximations for the $\VIX$ based on log-normal random variables in different classes of forward variance curve models, including basic and mixed Bergomi models with different convolution kernels. We have provided expansion results for $\VIX$ derivatives that are accurate, showing that the resulting approximation formulas can be used for fast and efficient calibration of the $\VIX$ implied volatility surface.

\section{Proofs\label{sec:proofs}}

The following estimate is standard.
\begin{lemma} \label{lem:inequality:power:sum}For
every $n\in\bN^{*}$ and $\bm{x}=\left(x_{1},\dots,x_{n}\right)\in\bR^{n}$,
the following inequality holds 
\[
\left|\sum_{i=1}^{n}x_{i}\right|^{p}\leq\begin{cases}
\sum_{i=1}^{n}\left|x_{i}\right|^{p} & \text{if }p\in(0,1],\\
n^{p-1}\sum_{i=1}^{n}\left|x_{i}\right|^{p} & \text{if }p>1.
\end{cases}
\]
\end{lemma}

\subsection{Proof of Proposition \ref{prop:estimate:YT:meanYT}}

From \eqref{eq:def:YT}, we have  
\begin{equation}
Y_{T}^{u}-\nu_{0}\left(\Ydot T\right)=-\frac{1}{2}\int_{0}^{T}\left(K^{u}\left(t\right)^{2}-\nu_{0}\left(\Kdot\left(t\right)^{2}\right)\right)\dd t+\int_{0}^{T}\left(K^{u}\left(t\right)-\nu_{0}\left(\Kdot\left(t\right)\right)\right)\dd W_{t}\,,
\label{eq:Y-nu0Y}
\end{equation}
for every $u\in[T,T+\Delta]$.
Estimate \eqref{eq:Lpq:Y} can be obtained with an application of Minkowski and Burkholder--Davis--Gundy inequalities,
Lemma \ref{lem:inequality:power:sum}, and conditions \eqref{eq:assu:drift} and \eqref{eq:assu:diffusion}.
\qed

\subsection{Proof of Theorem \ref{thm:estimate:nue:nueproxy}}

Let $p\geq1$ and $\ve\in[0,1]$.
Starting from \eqref{eq:I:nth:derivative} and repeatedly applying the generalized Minkowski and H\"{o}lder inequalities, we have 
\begin{align}
 & \frac{\bigl\Vert I^{(n)}(\ve)\bigr\Vert_{p}}{\nu(\xidot0)}\leq\int_{\cA}\left\Vert \left(Y_{T}^{u}-\nu_{0}\left(\Ydot T\right)\right)^{n}e^{\ve\left(Y_{T}^{u}-\nu_{0}\left(\Ydot T\right)\right)+\nu_{0}\left(\Ydot T\right)}\right\Vert _{p}\nu_{0}\left(\dd u\right)\nonumber \\
 & \qquad\leq\int_{\cA}\left\Vert Y_{T}^{u}-\nu_{0}\left(\Ydot T\right)\right\Vert _{3np}^{n}\times\left\Vert e^{\ve\left(Y_{T}^{u}-\nu_{0}\left(\Ydot T\right)\right)}\right\Vert _{3p}\times\left\Vert e^{\nu_{0}\left(\Ydot T\right)}\right\Vert _{3p}\nu_{0}\left(\dd u\right)\nonumber \\
 & \qquad\leq\bigl\Vert e^{\nu_{0}\left(\Ydot T\right)}\bigr\Vert_{3p}\Bigl(\int_{\cA}\bigl\Vert Y_{T}^{u}-\nu_{0}\left(\Ydot T\right)\bigr\Vert_{3np}^{2n}\nu_{0}\left(\dd u\right)\Bigr)^{\frac{1}{2}}\Bigl(\int_{\cA}\bigl\Vert e^{\ve\left(Y_{T}^{u}-\nu_{0}\left(\Ydot T\right)\right)}\bigr\Vert_{3p}^{2}\nu_{0}\left(\dd u\right)\Bigr)^{\frac{1}{2}}.\label{eq:proof:norm:In}
\end{align}
From \eqref{eq:def:proxy} and \eqref{eq:mean:variance:proxy}, we
have $\bigl\Vert e^{\nu_{0}(\Ydot T)}\bigr\Vert_{p}\leq e^{\frac{p\sigmap^{2}}{2}}$
which is uniformly bounded in $\Delta$ thanks to \eqref{eq:assu:limit:variance:proxy}.
Using\eqref{eq:assu:drift} and the boundedness of $u\mapsto\xi_{0}^{u}$, we have 
\[
\bigl\Vert e^{\ve(Y_{T}^{u}-\nu_{0}(\Ydot T))}\bigr\Vert_{p}=e^{-\frac{\ve}{2}\int_{0}^{T}(K^{u}(t)^{2}-\nu_{0}(\Kdot(t)^{2}))\dd t}e^{\frac{p\ve^{2}}{2}\int_{0}^{T}(K^{u}(t)-\nu_{0}(\Kdot(t)))^{2}\dd t}\leq C_{1}e^{C_{2}(\int_{0}^{T}K^{u}(t)^{2}\dd t+\sigmap^{2})},
\]
for some positive constants $C_{1},C_{2}$. Integrating with respect to $u$ and applying \eqref{eq:assu:kernel:exp:integral:u} and \eqref{eq:assu:diffusion},
we infer that the third factor on the right-hand side of \eqref{eq:proof:norm:In}
is uniformly bounded in $\Delta$. Finally, applying Proposition \ref{prop:estimate:YT:meanYT}
with $q=2n$, we see that the second factor on the right-hand side of \eqref{eq:proof:norm:In} is bounded
by $\Delta^{(d_{1}\wedge\frac{d_{2}}{2})n}$ up to a multiplying constant,
and we obtain the estimate \eqref{eq:Lp:I:nth:derivative}.

We now consider estimate \eqref{eq:Lp:exp:exp:proxy}. Applying the generalized Minkowski inequality to \eqref{eq:taylor:I1:I0}, we obtain
\[
\left\Vert \VIX_{T}^{2}-\VIX_{T,{\rm P}}^{2}\right\Vert _{p}\leq\int_{0}^{1}\left(1-\ve\right)\left\Vert I^{\left(2\right)}\left(\ve\right)\right\Vert _{p}\dd\ve\le\sup_{\ve\in[0,1]}\left\Vert I^{\left(2\right)}\left(\ve\right)\right\Vert _{p},
\]
and we can conclude applying \eqref{eq:Lp:I:nth:derivative} with $n=2.$
\qed

\subsection{Proof of Theorem \ref{thm:expansion:plain:model}} \label{sec:proof_main_thm}

Recall that $X_{T}^{u}=\ln\left(\xi_{T}^{u}\right),$ $Y_{T}^{u}=X_{T}^{u}-X_{0}^{u}$,
$\VIX_{T}^{2}=\nu(\xidot0)\nu_{0}(e^{Y_{T}^{\cdot}})$, $\VIX_{T,{\rm P}}^{2}=\nu(\xidot0)e^{\nu_{0}\left(\Ydot T\right)}$,
and let $\left[\varphi\right]_{\Hol}$ be the H\"{o}lder coefficient of
$\varphi$.

We are inspired by the techniques of \cite{bompis:gobet:2018}, where  the authors provide analytical approximations of option prices with Lipschitz payoffs in a  local-Heston volatility model.
In our case, to alleviate the possible lack
of smoothness of $\varphi$ and to overcome some degeneracy problems in the Malliavin sense (see later for more details), we introduce the Gaussian regularization $x\mapsto\varphi_{\delta}(x)=\bE\left[\varphi(x+\delta B_{T})\right]$ with a positive parameter $\delta$ defined by 
\begin{equation}
\delta:=\Delta^{\frac{3}{\theta}(d_{1}\wedge\frac{d_{2}}{2})},\label{eq:condition:on:delta}
\end{equation}
and where $B$ is a standard Brownian motion independent of $W$.
Note that 
\begin{align}
\varphi_{\sqrt{2}\delta}(x)=\bE\left[\varphi(x+\delta B_{T}+\delta\tilde{B}_{T})\right]=\bE\left[\varphi_{\delta}\left(x+\delta B_{T}\right)\right],\label{eq:convolution}
\end{align}
where $\tilde{B}$ is another standard Brownian motion such that $B,\tilde{B},W$ are all independent.  The functions
$x\mapsto\varphi_{\delta}(x)$ and $x\mapsto\varphi_{\sqrt{2}\delta}(x)$
are smooth ($\cC^{\infty}\left(\bR\right)$). Applying Taylor's
theorem with integral remainder to $\varphi_{\sqrt{2}\delta}$ at
the points $\VIX_{T}^{2}$ and $\VIX_{T,{\rm P}}^{2}$, one gets
\begin{align}
\bE\bigl[\varphi_{\sqrt{2}\delta}\left(\VIX_{T}^{2}\right)\bigr]
& =\bE\bigl[\varphi_{\sqrt{2}\delta}\left(\VIX_{T,{\rm P}}^{2}\right)\bigr]
+ \bE\Bigl[\varphi_{\sqrt{2}\delta}^{\prime}\left(\VIX_{T,{\rm P}}^{2}\right)\left(\VIX_{T}^{2}-\VIX_{T,{\rm P}}^{2}\right)\Bigr] + E_{0}\bigl(\varphi_{\sqrt{2}\delta}^{\prime \prime}\bigr) \,, \label{eq:taylor:phi}
\\
E_{0}\bigl(\varphi_{\sqrt{2}\delta}^{\prime \prime}\bigr) & :=\int_{0}^{1}\left(1-\lambda\right)\bE\bigl[\varphi_{\sqrt{2}\delta}^{\prime \prime}\left(\lambda\VIX_{T}^{2}+\left(1-\lambda\right)\VIX_{T,{\rm P}}^{2}\right)\left(\VIX_{T}^{2}-\VIX_{T,{\rm P}}^{2}\right)^{2}\bigr]\dd\lambda.\nonumber 
\end{align}
Recall that, from \eqref{eq:taylor:I1:I0} and \eqref{eq:I:nth:derivative}, we have the representation formula
\begin{align*}
\VIX_{T}^{2}-\VIX_{T,{\rm P}}^{2}
=
\frac{1}{2} \VIX_{T,{\rm P}}^{2} \int_{\cA}\left(Y_{T}^{u}-\nu_{0}\left(\Ydot T\right)\right)^{2}\nu_{0}\left(\dd u\right)
+ \int_{0}^{1}\frac{\left(1-\ve\right)^{2}}{2}I^{(3)}\left(\ve\right)\dd\ve.
\end{align*}
Now, introducing the new function 
\[
\Psi\left(x\right):=\varphi_{\sqrt{2}\delta}\left(\nu\left(\xidot0\right)e^{x}\right),
\]
we have 
\begin{align*}
\VIX_{T,{\rm P}}^{2} \, \varphi_{\sqrt{2}\delta}^{\prime}\left(\VIX_{T,{\rm P}}^{2}\right) 
& = \nu\left(\xidot0\right)e^{\nu_{0}\left(\Ydot T\right)}\varphi_{\sqrt{2}\delta}^{\prime}\left(\nu\left(\xidot0\right)e^{\nu_{0}\left(\Ydot T\right)}\right)
=\Psi^{\prime}\left(\nu_{0}\left(\Ydot T\right)\right),
\end{align*}
and the second expectation on the right-hand side of \eqref{eq:taylor:phi} can be rewritten as
\begin{multline}
\bE\Bigl[\varphi_{\sqrt{2}\delta}^{\prime}\left(\VIX_{T,{\rm P}}^{2}\right)\left(\VIX_{T}^{2}-\VIX_{T,{\rm P}}^{2}\right)\Bigr]
\\
= \int_{\cA}
\bE\Bigl[\Psi^{\prime} \left(\nu_{0}\left(\Ydot T\right)\right) \frac{1}{2}\left(Y_{T}^{u}-\nu_{0}\left(\Ydot T\right)\right)^{2}\Bigr] \nu_{0}\left(\dd u\right)
+ 
E_{1}\bigl(\varphi_{\sqrt{2}\delta}'\bigr)\label{eq:dphi:nu:moins:nuexp}\,,
\end{multline}
where $E_{1}\bigl(\varphi_{\sqrt{2}\delta}'\bigr):=\bE\left[\varphi_{\sqrt{2}\delta}'\left(\VIX_{T,{\rm P}}^{2}\right)\int_{0}^{1}\frac{\left(1-\ve\right)^{2}}{2}I^{(3)}\left(\ve\right)\dd\ve\right]$.

We are going to further manipulate the expectation term appearing inside the integral $\int_{\cA} \nu_{0}\left(\dd u\right)$ on the right-hand side of \eqref{eq:dphi:nu:moins:nuexp}.
Recalling \eqref{eq:Y-nu0Y}, an application of It\^{o}'s formula to the process $\left(Y_{t}^{u}-\nu_{0}\left(\Ydot t\right)\right)^{2}$, $0 \le t \le T$, for fixed $u$ yields
\be \label{e:Ito_Y}
\begin{aligned}
\frac{1}{2}\left(Y_{T}^{u}-\nu_{0}\left(\Ydot T\right)\right)^{2} 
&= 
\int_0^T \bigl[ Y_{t}^{u}-\nu_{0}\left(\Ydot t\right) \bigr] \dd \bigl( Y_{t}^{u}-\nu_{0}\left(\Ydot t\right)  \bigr)
+ \frac 12 \int_0^T \dd \langle Y_{t}^{u}-\nu_{0}\left(\Ydot t\right) \rangle
\\
&= \int_{0}^{T}\left[\int_{0}^{t}\left(-\frac{1}{2}K^{u}\left(s\right)^{2}+\frac{1}{2}\nu_{0}\left(\Kdot\left(s\right)^{2}\right)\right)\dd s+\left(K^{u}\left(s\right)-\nu_{0}\left(\Kdot\left(s\right)\right)\right)\dd W_{s}\right]
\\
 & \qquad\times\left(\left[-\frac{1}{2}K^{u}\left(t\right)^{2}+\frac{1}{2}\nu_{0}\left(\Kdot\left(t\right)^{2}\right)\right]\dd t+\left[K^{u}\left(t\right)-\nu_{0}\left(\Kdot\left(t\right)\right)\right]\dd W_{t}\right)
\\
 & \quad +\int_{0}^{T}\frac{1}{2}\left[K^{u}\left(t\right)-\nu_{0}\left(\Kdot\left(t\right)\right)\right]^{2}\dd t.
\end{aligned}
\ee
In light of \eqref{e:Ito_Y}, the expectation $\bE\Bigl[\Psi'\left(\nu_{0}\left(\Ydot T\right)\right)\frac{1}{2}\left(Y_{T}^{u}-\nu_{0}\left(\Ydot T\right)\right)^{2}\Bigr]$ we want to evaluate is equal to
\begin{equation}\label{eq:proof:psi:prime}
	\bE\Bigl[\Psi'\left(\nu_{0}\left(\Ydot T\right)\right)\int_{0}^{T}\bigl[Y_{t}^{u}-\nu_{0}\left(\Ydot t\right)\bigr]\dd\bigl(Y_{t}^{u}-\nu_{0}\left(\Ydot t\right)\bigr)\Bigr]
+\frac{1}{2}\int_{0}^{T}  \bE\left[\Psi^{\prime}(\nu_{0}\left(\Ydot T\right))\right]
\dd\langle Y_{t}^{u}-\nu_{0}\left(\Ydot t\right)\rangle \,,
\end{equation}
since the quadratic variation $\langle Y_{t}^{u}-\nu_{0}\left(\Ydot t\right)\rangle = \int_{0}^{T}\frac{1}{2}\left[K^{u}\left(t\right)-\nu_{0}\left(\Kdot\left(t\right)\right)\right]^{2}\dd t$ is deterministic.
The expectation $ \bE\left[\Psi^{\prime}(\nu_{0}\left(\Ydot T\right))\right]$ gives rise to the first derivative of a Black-Scholes price, the variable $\nu_{0}\left(\Ydot T\right)$ being Gaussian.
Our goal is to give a similar formulation for the first expectation in \eqref{eq:proof:psi:prime}, too.
To do so, we recall a formula for the integration by parts of functions of iterated Wiener integrals, under a form that is suitable for our purposes.
\begin{lemma}[{\cite[Lemma A.2.]{gobet2014weak}}]
\label{lem:integration:by:parts}Let $a_{t},e_{t},f_{t},g_{t},h_{t}:\left[0,T\right]\to\bR$
be square-integrable and deterministic processes, and let $l:\bR \to \bR$ be a bounded
smooth function with bounded derivatives. Then, 
\begin{multline*}
\bE\left[l\left(\int_{0}^{T}a_{t}\dd W_{t}\right)\left(\int_{0}^{T}\left[\int_{0}^{t}g_{s}\dd s+h_{s}\dd W_{s}\right]\left(e_{t}\dd t+f_{t}\dd W_{t}\right)\right)\right]\\
=\sum_{i=0}^{2}\lambda_{i} \, \partial_{\ve}^{i}\left.\bE\left[l\left(\int_{0}^{T}a_{t}\dd W_{t}+\ve\right)\right]\right|_{\ve=0}
\end{multline*}
where 
\[
\lambda_{0}=\int_{0}^{T}\int_{0}^{t}e_{t}g_{s}\dd s\dd t,\quad\lambda_{1}=\int_{0}^{T}\int_{0}^{t}\left(g_{s}a_{t}f_{t}+e_{t}a_{s}h_{s}\right)\dd s\dd t,\quad\lambda_{2}=\int_{0}^{T}\int_{0}^{t}a_{t}f_{t}a_{s}h_{s}\dd s\dd t.
\]
\end{lemma}


\noindent We set 
\begin{align*}
	l(x) & =\Psi'\biggl(x-\frac{1}{2}\int_{0}^{T}\nu_{0}(\Kdot(t)^{2})\dd t\biggr),\quad a_{t}=\nu_{0}\left(\Kdot\left(t\right)\right),
\\
f_{t} (u)&=h_{t}(u)=K^{u}\left(t\right)-\nu_{0}\left(\Kdot\left(t\right)\right),\quad e_{t}(u) = g_{t}(u) = -\frac{1}{2}K^{u}\left(t\right)^{2}+\frac{1}{2}\nu_{0}\Bigl(\Kdot\left(t\right)^{2}\Bigr), \\
	\lambda_{0}(u)&=\int_{0}^{T}\int_{0}^{t}e_{t}(u)g_{s}(u)\dd s\dd t,\quad\lambda_{1}(u)=\int_{0}^{T}\int_{0}^{t}\left(g_{s}(u)a_{t}f_{t}(u)+e_{t}(u)a_{s}h_{s}(u)\right)\dd s\dd t,
	\\
	\lambda_{2}(u)&=\int_{0}^{T}\int_{0}^{t}a_{t}f_{t}(u)a_{s}h_{s}(u)\dd s\dd t \,,
\end{align*}
and apply Lemma \ref{lem:integration:by:parts} to the first expectation in \eqref{eq:proof:psi:prime}.
Recalling that $\Psi^{\prime}(x)=\nu\left(\xidot0\right)e^{x}\varphi_{\sqrt{2}\delta}^{\prime}\left(\nu\left(\xidot0\right)e^{x}\right)$ for every $x\in\mathbb{R}$, we obtain
\[
\int_{\cA}
\bE\left[\Psi'\left(\nu_{0}\left(\Ydot T\right)\right)\frac{1}{2}\left(Y_{T}^{u}-\nu_{0}\left(\Ydot T\right)\right)^{2}\right]
\nu_{0}\left(\dd u\right)
 =\sum_{i=1}^{3}\gamma_{i} \left.\partial_{\ve}^{i} \, \bE\left[\varphi_{\sqrt{2}\delta}\left(\nu\left(\xidot0\right)e^{\nu_{0}\left(\Ydot T\right)+\ve}\right)\right]\right|_{\ve=0}.
\]
where
\begin{align*}
\gamma_{1} & =\int_{\mathcal{A}}\lambda_{0}(u)\nu_{0}({\rm d}u)
+ \frac{1}{2}\int_{\mathcal{A}}\int_{0}^{T}\dd\langle Y_{t}^{u}-\nu_{0}\left(\Ydot t\right)\rangle \, \nu_{0}({\rm d}u),
\\
\gamma_{2} & =\int_{\mathcal{A}}\lambda_{1}(u)\nu_{0}({\rm d}u),\qquad\gamma_{3}=\int_{\mathcal{A}}\lambda_{2}(u)\nu_{0}({\rm d}u).
\end{align*}
Putting things together, recalling that $ \nu\left(\xidot0\right)e^{\nu_{0}\left(\Ydot T\right)} = \VIX_{T,{\rm P}}^{2}$, we have shown that 
\begin{multline}
\bE\left[\varphi_{\sqrt{2}\delta}\left(\VIX_{T}^{2}\right)\right]=\bE\left[\varphi_{\sqrt{2}\delta}\left(\VIX_{T,{\rm P}}^{2}\right)\right]+\sum_{i=1}^{3}\gamma_{i}\left.\partial_{\ve}^{i}\bE\left[\varphi_{\sqrt{2}\delta}\left(\VIX_{T,{\rm P}}^{2} \, e^{\ve}\right)\right]\right|_{\ve=0}\\
+ E_{1}\Bigl(\varphi_{\sqrt{2}\delta}'\Bigr)
+ E_{0}\Bigr(\varphi_{\sqrt{2}\delta}^{\prime \prime}\Bigr) \,.
\label{eq:proof:expansion:formula:varphi:n}
\end{multline}

Before estimating the error terms $E_{1}$ and $E_{0}$ in \eqref{eq:proof:expansion:formula:varphi:n}, we wish to get back to the true payoff function, and provide an identity analogous to \eqref{eq:proof:expansion:formula:varphi:n}
for the true payoff function $\varphi$ instead of its regularized version $\varphi_{\sqrt{2}\delta}$.
We start with some estimates of the coefficients $(\gamma_{i})_{i=1,2,3}$. 
\begin{lem}
\label{lem:estimates:mup:sigmap:gamma:i} For every $i\in\left\{ 1,2,3\right\} $,
$\gamma_{i}=\cO(\Delta^{2d_{1}\wedge d_{2}})$ as $\Delta\to0$. 
\end{lem}

\begin{proof}
From \eqref{eq:assu:drift}
and \eqref{eq:assu:diffusion}, we have $\gamma_{1}=\cO(\Delta^{2d_{1}})+\cO(\Delta^{d_{2}})=\cO(\Delta^{2d_{1}\wedge d_{2}}).$
Using the trivial inequality $xy\leq\frac{1}{2}\left(x^{2}+y^{2}\right)$
for $x,y\in\bR$ along with the Cauchy--Schwarz inequality, we see
that $\left|\gamma_{2}\right|$ is upper-bounded by 
\[
\frac{\sigmap^{2}}{2}\int_{\cA}\left(\int_{0}^{T}\left[K^{u}\left(t\right)-\nu_{0}\left(\Kdot\left(t\right)\right)\right]^{2}\dd t\right)\nu_{0}(\dd u)+\frac{1}{8}\int_{\cA}\left|\int_{0}^{T}\left[K^{u}\left(t\right)^{2}-\nu_{0}\left(\Kdot\left(t\right)^{2}\right)\right]\dd t\right|^{2}\nu_{0}(\dd u) \,.
\]
Applying \eqref{eq:assu:drift}, \eqref{eq:assu:diffusion} and \eqref{eq:assu:limit:variance:proxy},
we get $\gamma_{2}=\cO(\Delta^{d_{2}})+\cO(\Delta^{2d_{1}})=\cO(\Delta^{2d_{1}\wedge d_{2}}).$
Finally, applying once again the Cauchy--Schwarz inequality, \eqref{eq:assu:diffusion} and \eqref{eq:assu:limit:variance:proxy},
we have 
\[
\gamma_{3}\leq\frac{\sigmap^{2}}{2}\int_{\cA}\left(\int_{0}^{T}\left[K^{u}\left(t\right)-\nu_{0}\left(\Kdot\left(t\right)\right)\right]^{2}\dd t\right)\nu_{0}\left(\dd u\right)=\cO(\Delta^{d_{2}}).
\]
\end{proof}

Note that, since $\varphi$ is $\theta$-H\"{o}lder continuous, the following
useful estimate holds
\begin{equation}
\sup_{x\in\bR}\left|\varphi_{\delta}(x)-\varphi(x)\right|\leq\left[\varphi\right]_{\Hol}\bE\left[\left|\delta B_{T}\right|^{\theta}\right]=\cO(\delta^{\theta}) \,. \label{eq:estimate:phidelta:phi}
\end{equation}
Consequently, in view of \eqref{eq:condition:on:delta} we have
\begin{equation}
\bE\left[\varphi\left(\VIX_{T}^{2}\right)-\varphi\left(\VIX_{T,{\rm P}}^{2}\right)\right]=\bE\left[\varphi_{\sqrt{2}\delta}\left(\VIX_{T}^{2}\right)-\varphi_{\sqrt{2}\delta}\left(\VIX_{T,{\rm P}}^{2}\right)\right]+\cO(\Delta^{3(d_{1}\wedge\frac{d_{2}}{2})}).
\label{eq:estimate:Error:0}
\end{equation}
For $i\in\{1,2,3\}$, we also have
\begin{align}
\left.\partial_{\ve}^{i}\bE\left[\varphi_{\sqrt{2}\delta}\left(\VIX_{T,{\rm P}}^{2}e^{\ve}\right)\right]\right|_{\ve=0} & =\left.\partial_{\ve}^{i}\bE\left[\varphi\left(\VIX_{T,{\rm P}}^{2}e^{\ve}\right)\right]\right|_{\ve=0}+\Error_{1,i},\label{eq:equality:Error:i}
\\
\left|\Error_{1,i}\right|
& \leq_{c}\delta^{\theta}\int_{0}^{\infty}
\biggl|\partial_{\ve}^{i}\Big(e^{-\frac{\left(\ln\left(y\right)-\mup-\ve\right)^{2}}{2\sigmap^{2}}}\Big)\Big|_{\ve=0}\biggr|\frac{\dd y}{\sqrt{2\pi} \, \sigmap \, y}
=\cO(\Delta^{3(d_{1}\wedge\frac{d_{2}}{2})}) \,. \label{eq:error:greeks}
\end{align}
The first inequality  in \eqref{eq:error:greeks} is  obtained applying \eqref{eq:estimate:phidelta:phi} and Proposition \ref{prop:proxies:characteristics}.
In order to prove the last identity in \eqref{eq:error:greeks},
observe that the partial derivatives are related to Hermite polynomials
and it is easy to prove that the integral is uniformly bounded with respect to
$\Delta$ thanks to \eqref{eq:assu:limit:mean:proxy} and \eqref{eq:assu:limit:variance:proxy}.
Putting Lemma \ref{lem:estimates:mup:sigmap:gamma:i},
\eqref{eq:proof:expansion:formula:varphi:n}, \eqref{eq:estimate:Error:0}, and
\eqref{eq:error:greeks} together, the VIX option price $\bE\left[\varphi\left(\VIX_{T}^{2}\right)\right]$
is equal to 
\begin{equation}
\bE\left[\varphi\left(\VIX_{T,{\rm P}}^{2}\right)\right]+\sum_{i=1}^{3}\gamma_{i}\left.\partial_{\ve}^{i}\bE\left[\varphi\left(\VIX_{T,{\rm P}}^{2} \, e^{\ve}\right)\right]\right|_{\ve=0}
+\cO \Bigl(\Delta^{3(d_{1}\wedge\frac{d_{2}}{2})}\Bigr)
+E_{1}\Bigl(\varphi_{\sqrt{2}\delta}^{\prime} \Bigr)
+E_{0}\Bigl(\varphi_{\sqrt{2}\delta}^{\prime \prime}\Bigr) \,.
\label{eq:proof:expansion:error:en:delta}
\end{equation}

\textbf{Error terms}. We now wish to establish estimates for the two error terms $E_{0}(\varphi_{\sqrt{2}\delta}^{\prime \prime})$
and $E_{1}(\varphi_{\sqrt{2}\delta}^{\prime})$ defined in \eqref{eq:taylor:phi}-\eqref{eq:dphi:nu:moins:nuexp}.
Though the function $\varphi_{\sqrt{2}\delta}$ is smooth, establishing
$\Delta$-uniform estimates for its first and second derivatives seem
difficult. Therefore, we exploit an integration-by-parts formula to get back
to $\varphi_{\sqrt{2}\delta}$. To do so, we use the calculus of variations with respect to the two-dimensional Brownian motion $\left(W_{t},B_{t}\right)_{t\in\left[0,T\right]}$.
We refer to the notations of \cite{nual:06} for the Malliavin Sobolev spaces $\bD^{k,p}$
associated to the norm $\Vert\cdot\Vert_{k,p}$. The Malliavin derivative
operator with respect to  $W$ and $B$ is denoted $D(\cdot):=(D_{t}^{1}(\cdot),D_{t}^{2}(\cdot))_{t\in[0,T]}$.
For the second-order Malliavin derivatives, we use the notation $(D^{i,j}(\cdot))_{s,t\in[0,T]}$ for $i,j\in\left\{ 1,2\right\} ,$ and so on for higher derivatives.
For any $\lambda\in[0,1]$, we set 
\begin{align}
G_{\delta}^{\lambda} & :=\lambda\VIX_{T}^{2}+(1-\lambda)\VIX_{T,{\rm P}}^{2}+\delta B_{T}\nonumber \\
 & =\lambda I(1)+(1-\lambda)I(0)+\delta B_{T}=:G^{\lambda}+\delta B_{T},\label{eq:def:G:delta}
\end{align}
recalling \eqref{eq:def:I:eps}. Even though the proxy $\VIX_{T,{\rm P}}^{2}$
is a non-degenerate random variable (owing to \eqref{eq:assu:limit:variance:proxy}
in Assumption \ref{assu:kernel}), with Malliavin covariance matrix
$\sigmap^{2}\VIX_{T,{\rm P}}^{4}>0$, the convex combination $G^{\lambda}$
may be degenerated in the Malliavin sense, and the integration-by-parts
formula may not hold. Having introduced the Gaussian perturbation
$\sqrt{2}\delta B_{T}$ precisely answers this issue and corresponds
to the second reason for having regularized $\varphi$. For these
reasons, it is crucial to observe that the convolution identity \eqref{eq:convolution}
yields the new expressions 
\begin{align}
E_{1}(\varphi_{\sqrt{2}\delta}^{\prime}) & =\bE\Big[\varphi_{\delta}^{\prime}\left(G_{\delta}^{0}\right)\int_{0}^{1}\frac{\left(1-\ve\right)^{2}}{2}I^{(3)}\left(\ve\right)\dd\ve\Big],\label{eq:error:E1:G}
\\
E_{0}(\varphi_{\sqrt{2}\delta}^{\prime \prime})
& =\int_{0}^{1}\left(1-\lambda\right)
\bE\Big[\varphi_{\delta}^{\prime \prime}\left(G_{\delta}^{\lambda}\right)\Big(\int_{0}^{1}\left(1-\ve\right)I^{(2)}\left(\ve\right)\dd\ve\Big)^{2}\Big]\dd\lambda
\label{eq:error:E0:G}
\end{align}
where we have started from \eqref{eq:taylor:phi}-\eqref{eq:dphi:nu:moins:nuexp},
and used \eqref{eq:taylor:I1:I0}.

We start with some estimates for the Sobolev
norms of $I^{(n)}(\ve)$ and $I^{(n)}(\ve)^{2}.$ 
\begin{lem}
\label{lem:estimates:I(n)}For every $p\geq1$, $n\in\bN,k\in\left\{ 1,2,3\right\} ,$
the following estimates hold 
\begin{align}
\sup_{\ve\in\left[0,1\right]}\left\Vert I^{(n)}(\ve)\right\Vert _{k,p} & \leq_{c}\Delta^{n(d_{1}\wedge\frac{d_{2}}{2})}.
\label{eq:estimate:I(n):norm:k:p}
\end{align}
\end{lem}

\begin{proof}
It is enough to assume $p\geq2$ thanks to the non-expansivity of
the $L^{p}$ norm. For every $\ve\in\left[0,1\right],u\in\cA$, introduce
\[
\cE_{n}^{u}(\ve):=\left(Y_{T}^{u}-\nu_{0}\left(\Ydot T\right)\right)^{n}e^{\nu_{0}\left(\Ydot T\right)+\ve\left(Y_{T}^{u}-\nu_{0}\left(\Ydot T\right)\right)},\qquad I^{\left(n\right)}\left(\ve\right)=\nu\left(\xidot0\right)\int_{\cA}\cE_{n}^{u}(\ve)\nu_{0}\left(\dd u\right).
\]
Proceeding as in the proof of Theorem \ref{thm:estimate:nue:nueproxy}
and evoking Proposition \ref{prop:estimate:YT:meanYT} with $p=2n$,
we have that 
\begin{align}
\big\Vert\bigl(\int_{\cA}\left|\cE_{n}^{u}(\ve)\right|^{2}\nu_{0}\left(\dd u\right)\bigr)^{\frac{1}{2}}\big\Vert_{p} & \leq\bigl(\int_{\cA}\left\Vert \left(\left(Y_{T}^{u}-\nu_{0}\left(\Ydot T\right)\right)^{n}e^{\nu_{0}\left(\Ydot T\right)+\ve\left(Y_{T}^{u}-\nu_{0}\left(\Ydot T\right)\right)}\right)^{2}\right\Vert _{\frac{p}{2}}\nu_{0}\left(\dd u\right)\bigr)^{\frac{1}{2}}\nonumber \\
 & \leq\bigl(\int_{\cA}\left\Vert Y_{T}^{u}-\nu_{0}\left(\Ydot T\right)\right\Vert _{2pn}^{2n}\left\Vert e^{\nu_{0}\left(\Ydot T\right)+\ve\left(Y_{T}^{u}-\nu_{0}\left(\Ydot T\right)\right)}\right\Vert _{2p}^{2}\nu_{0}\left(\dd u\right)\bigr)^{\frac{1}{2}}\nonumber \\
 & \leq_{c}\Delta^{n(d_{1}\wedge\frac{d_{2}}{2})}.\label{eq:proof:Eronde:eps}
\end{align}
Observing that $D_{t}^{1}(Y_{T}^{u}-\nu_{0}(\Ydot{T}))=K^{u}(t)-\nu_{0}(\Kdot(t)),$
and applying the chain rule, we have 
\begin{align}
D_{t}^{1}\Bigl(\frac{I^{(n)}(\ve)}{\nu(\xidot0)}\Bigr) & =n\int_{\cA}\left(K^{u}(t)-\nu_{0}(\Kdot(t))\right)\cE_{n-1}^{u}(\ve)\nu_{0}\left(\dd u\right)\nonumber \\
 & \qquad+\nu_{0}\left(\Kdot(t)\right)\int_{\cA}\cE_{n}^{u}(\ve)\nu_{0}\left(\dd u\right)+\int_{\cA}\left[\ve\left(K^{u}\left(t\right)-\nu_{0}\left(\Kdot(t)\right)\right)\right]\cE_{n}^{u}(\ve)\nu_{0}\left(\dd u\right)\nonumber \\
 & =:nR_{1,n-1}(t)+R_{2,n}(t)+\ve R_{1,n}(t).\label{eq:malliavin:derivee:1}
\end{align}
From Cauchy--Schwarz inequality, it holds 
\[
\left|R_{1,n}(t)\right|^{2}\leq\bigl(\int_{\cA}\left|K^{u}(t)-\nu_{0}(\Kdot(t))\right|^{2}\nu_{0}\left(\dd u\right)\bigr)\bigl(\int_{\cA}\left|\cE_{n}^{u}(\ve)\right|^{2}\nu_{0}\left(\dd u\right)\bigr),
\]
and reintegrating in $t$ and evoking \eqref{eq:assu:diffusion},
\eqref{eq:proof:Eronde:eps}, and taking the $L^{p}$ norm, we have
\[
\big\Vert\bigl(\int_{0}^{T}\left|R_{1,n}(t)\right|^{2}\dd t\bigr)^{\frac{1}{2}}\big\Vert_{p}\leq_{c}\Delta^{\frac{d_{2}}{2}+n(d_{1}\wedge\frac{d_{2}}{2})},\quad n\geq0.
\]
Similarly, combining Cauchy--Schwarz inequality with \eqref{eq:mean:variance:proxy},
\eqref{eq:assu:limit:variance:proxy} and \eqref{eq:proof:Eronde:eps},
\[
\big\Vert\bigl(\int_{0}^{T}\left|R_{2,n}(t)\right|^{2}\dd t\bigr)^{\frac{1}{2}}\big\Vert_{p}\leq\sigmap\big\Vert\bigl(\int_{\cA}\left|\cE_{n}^{u}(\ve)\right|^{2}\nu_{0}\left(\dd u\right)\bigr)^{\frac{1}{2}}\big\Vert_{p}\leq_{c}\Delta^{n(d_{1}\wedge\frac{d_{2}}{2})}.
\]
Noting that $D_{t}^{2}(I^{(n)}(\ve))=0$ and plugging estimates into
\eqref{eq:malliavin:derivee:1}, we get that $\left\Vert \bigl(\int_{0}^{T}\left|D_{t}I^{(n)}(\ve)\right|^{2}\dd t\bigr)^{\frac{1}{2}}\right\Vert _{p}\leq_{c}\Delta^{n(d_{1}\wedge\frac{d_{2}}{2})}$.
The last estimate being independent of $\ve$, \eqref{eq:estimate:I(n):norm:k:p}
is proved for $k=1$.

For every $s,t\in[0,T],$ the second-order Malliavin derivative writes
as 
\begin{align*}
D_{t,s}^{1,1}\Bigl(\frac{I^{(n)}(\ve)}{\nu(\xidot0)}\Bigr) & \underset{\eqref{eq:malliavin:derivee:1}}{}=nD_{s}^{1}R_{1,n-1}(t)+D_{s}^{1}R_{2,n}(t)+\ve D_{s}^{1}R_{1,n}(t),\\
D_{s}^{1}R_{1,n}(t) & =nS_{1,n-1}(t,s)+S_{2,n}(t,s)+\ve S_{1,n}(t,s),\\
D_{s}^{1}R_{2,n}(t) & =\nu_{0}\left(\Kdot(t)\right)\Big(nR_{1,n-1}(s)+R_{2,n}(s)+\ve R_{1,n}(s)\Big),\\
S_{1,n}(t,s) & :=\int_{\cA}\left[K^{u}(t)-\nu_{0}\left(\Kdot\left(t\right)\right)\right]\left[K^{u}(s)-\nu_{0}\left(\Kdot\left(s\right)\right)\right]\cE_{n}^{u}(\ve)\nu_{0}(\dd u),\\
S_{2,n}(t,s) & :=\int_{\cA}\left[K^{u}(t)-\nu_{0}\left(\Kdot\left(t\right)\right)\right]\nu_{0}\left(\Kdot\left(s\right)\right)\cE_{n}^{u}(\ve)\nu_{0}(\dd u).
\end{align*}
Using \eqref{eq:assu:diffusion}, \eqref{eq:proof:Eronde:eps}, and
Cauchy--Schwarz inequality, we get 
\begin{align*}
 & \left|S_{1,n}(t,s)\right|^{2}\leq\int_{\cA}\left[K^{u}(t)-\nu_{0}\left(\Kdot\left(t\right)\right)\right]^{2}\left[K^{u}(s)-\nu_{0}\left(\Kdot\left(s\right)\right)\right]^{2}\nu_{0}(\dd u)\int_{\cA}\left|\cE_{n}^{u}(\ve)\right|^{2}\nu_{0}(\dd u),\\
 & \Vert\big(\int_{[0,T]^{2}}\left|S_{1,n}(t,s)\right|^{2}\dd t\dd s\big)^{\frac{1}{2}}\Vert_{p}\leq_{c}\Delta^{d_{2}+n(d_{1}\wedge\frac{d_{2}}{2})}.
\end{align*}
A similar study gives that $\Vert\bigl(\int_{[0,T]^{2}}\left|S_{2,n}(t,s)\right|^{2}\dd t\dd s\bigr)^{\frac{1}{2}}\Vert_{p}\leq_{c}\Delta^{\frac{d_{2}}{2}+n(d_{1}\wedge\frac{d_{2}}{2})}$
and using that $D_{t}^{1,2}(I^{(n)}(\ve))=D_{t}^{2,2}(I^{(n)}(\ve))=0$,
we obtain \eqref{eq:estimate:I(n):norm:k:p} for $k=2$. The same
estimate holds for $k=3$, the proof being similar we skip it. 
\end{proof}

The following lemma contains an estimate for the Sobolev norm of the Malliavin
derivative of $G_{\delta}^{\lambda}$. 
The proof is a direct consequence of \eqref{eq:def:G:delta}, Lemma \ref{lem:estimates:I(n)},
and \eqref{eq:Lp:I:nth:derivative}. 
\begin{lem}
\label{lemma:estimate:malliavin:derivative:G} For every $p\geq1,k\in\left\{ 1,2\right\} ,$
we have 
\begin{align}
\sup_{\lambda\in[0,1]}\left\Vert DG_{\delta}^{\lambda}\right\Vert _{k,p} & =\cO\left(1\right).\label{eq:estimate:malliavin:derivative:G}
\end{align}
\end{lem}

\emph{Estimate of the Malliavin covariance matrix.} The Malliavin
covariance matrix of $G_{\delta}^{\lambda}$ is defined as 
\begin{align*}
\gamma_{G_{\delta}^{\lambda}} & :=\int_{0}^{T}\big(|D_{t}^{1}(G_{\delta}^{\lambda})|^{2}+|D_{t}^{2}(G_{\delta}^{\lambda})|^{2}\bigr)\dd t\\
 & =\int_{0}^{T}|\lambda D_{t}^{1}(I(1))+(1-\lambda)\nu_{0}(\Kdot(t))I(0)|^{2}\dd t+\delta^{2}T=:\gamma_{G^{\lambda}}+\delta^{2}T,
\end{align*}
and is invertible as $\gamma_{G_{\delta}^{\lambda}}\geq\delta^{2}T>0$.
Observing that $\sigmap^{2}I(0)^{2}=\int_{0}^{T}\left|D_{t}^{1}(I(0))\right|^{2}\dd t$
and from H\"{o}lder inequality and Lemma \ref{lem:estimates:I(n)}, we
have for any $p\geq2$ 
\begin{align}
\sup_{\lambda\in[0,1]}\left\Vert \gamma_{G^{\lambda}}-\sigmap^{2}I(0)^{2}\right\Vert _{p} & =\sup_{\lambda\in[0,1]}\Bigl\Vert\int_{0}^{T}\lambda[D_{t}^{1}\left(I(1)-I(0)\right)][\lambda D_{t}^{1}(I(1))+(2-\lambda)D_{t}^{1}(I(0))]\dd t\Bigr\Vert_{p}\nonumber \\
 & \leq\Bigl\Vert\big(\int_{0}^{T}[D_{t}^{1}\left(I(1)-I(0)\right)]^{2}\dd t\bigr)^{\frac{1}{2}}\Bigr\Vert_{2p}\nonumber \\
 & \qquad\qquad\qquad\qquad\times\sup_{\lambda\in[0,1]}\Bigl\Vert\big(\int_{0}^{T}[\lambda D_{t}^{1}(I(1))+(2-\lambda)D_{t}^{1}(I(0))]^{2}\dd t\bigr)^{\frac{1}{2}}\Bigr\Vert_{2p}\nonumber \\
 & \leq_{c}\sup_{\ve\in[0,1]}\Bigl\Vert\big(\int_{0}^{T}\big|D_{t}^{1}(I^{(2)}(\ve))\bigr|^{2}\dd t\bigl)^{\frac{1}{2}}\Bigr\Vert_{2p}\leq_{c}\Delta^{2(d_{1}\wedge\frac{d_{2}}{2})},\label{eq:estimate:gammaG:sigmap2}
\end{align}
using the representation \eqref{eq:taylor:I1:I0}. The above estimate
allows us to prove the next lemma, which provides uniform estimates for
the Sobolev norms of the inverse of the Malliavin covariance matrix. 
\begin{lem}
\label{lemma:estimate:gammaG:-1:norm:i:p} For any $p\geq1$, $i\in\{0,1,2\},$
\begin{equation}
\sup_{\lambda\in[0,1]}\left\Vert (\gamma_{G_{\delta}^{\lambda}})^{-1}\right\Vert _{i,p}=\cO(1).\label{eq:estimate:gammaG:-1:norm:i:p}
\end{equation}
\end{lem}

\begin{proof}
For any $p\geq1,\lambda\in[0,1]$, using that $\gamma_{G_{\delta}^{\lambda}}=\gamma_{G^{\lambda}}+\delta^{2}T\geq\delta^{2}T>0$,
we have 
\begin{align*}
\bE\left[(\gamma_{G_{\delta}^{\lambda}})^{-p}\right] & =\bE\left[(\gamma_{G_{\delta}^{\lambda}})^{-p}\1_{\gamma_{G^{\lambda}}\leq\frac{\sigmap^{2}I(0)^{2}}{2}}\right]+\bE\left[(\gamma_{G_{\delta}^{\lambda}})^{-p}\1_{\gamma_{G^{\lambda}}>\frac{\sigmap^{2}I(0)^{2}}{2}}\right]\\
 & \leq(\delta^{2}T)^{-p}\bP\Big(\sigmap^{2}I(0)^{2}-\gamma_{G^{\lambda}}\geq\frac{\sigmap^{2}I(0)^{2}}{2}\Big)+2^{p}\bE\left[(\sigmap^{2}I(0)^{2})^{-p}\right].
\end{align*}
From \eqref{eq:assu:limit:variance:proxy} and since $I(0)$ is lognormal
with uniformly bounded parameters $\mup,\sigmap$, the second term
of the previous bound is uniformly bounded in $\Delta$, and using
Markov and H\"{o}lder inequalities, we have for any $q\geq1,$ 
\[
\bP\Big(\sigmap^{2}I(0)^{2}-\gamma_{G^{\lambda}}\geq\frac{\sigmap^{2}I(0)^{2}}{2}\Big)\leq_{c}\left\Vert I(0)^{-2}\right\Vert _{2q}^{q}\left\Vert \sigmap^{2}I(0)^{2}-\gamma_{G^{\lambda}}\right\Vert _{2q}^{q}=\cO(\Delta^{2q(d_{1}\wedge\frac{d_{2}}{2})}).
\]
Combining \eqref{eq:condition:on:delta}, \eqref{eq:estimate:gammaG:sigmap2},
and choosing $q=\frac{3p}{\theta}$, \eqref{eq:estimate:gammaG:-1:norm:i:p}
is proved for $i=0$. For the cases $i\in\{1,2\}$, applying the chain
rule (see \cite[Lemma 2.1.6]{nual:06}), we have for every $t,s\in[0,T],$
\begin{align*}
D_{t}^{1}(\gamma_{G_{\delta}^{\lambda}}^{-1}) & =-\frac{D_{t}^{1}(\gamma_{G_{\delta}^{\lambda}})}{\gamma_{G_{\delta}^{\lambda}}^{2}}=-2\gamma_{G_{\delta}^{\lambda}}^{-2}\int_{0}^{T}D_{u}^{1}(G^{\lambda})D_{u,t}^{1,1}(G^{\lambda})\dd u,\\
D_{s,t}^{1,1}(\gamma_{G_{\delta}^{\lambda}}^{-1}) & =-D_{s,t}^{1,1}(\gamma_{G_{\delta}^{\lambda}})\gamma_{G_{\delta}^{\lambda}}^{-2}+2D_{t}^{1}(\gamma_{G_{\delta}^{\lambda}})D_{s}^{1}(\gamma_{G_{\delta}^{\lambda}})\gamma_{G_{\delta}^{\lambda}}^{-3},\\
D_{t}^{2}(\gamma_{G_{\delta}^{\lambda}}^{-1}) & =D_{s,t}^{2,1}(\gamma_{G_{\delta}^{\lambda}}^{-1})=D_{s,t}^{1,2}(\gamma_{G_{\delta}^{\lambda}}^{-1})=D_{s,t}^{2,2}(\gamma_{G_{\delta}^{\lambda}}^{-1})=0;
\end{align*}
we conclude thanks to \eqref{eq:estimate:I(n):norm:k:p} and \eqref{eq:estimate:gammaG:-1:norm:i:p}
with $i=0$. 
\end{proof}

 \emph{Integration-by-parts formula and conclusion}. 
We can now conclude the proof of Theorem \ref{thm:expansion:plain:model}. 
The last ingredient is  the following Malliavin integration-by-parts formula.
\begin{prop}
\label{prop:ipp:malliavin}Let $i\in\left\{ 1,2\right\} ,$ then for
any $H\in\bD^{i,\infty}$ there exist random variables $(H_{i,\lambda})_{\lambda\in[0,1]}$
in $\cap_{p\geq1}L^{p}$ such that 
\[
\bE\left[H\varphi_{\delta}^{(i)}\left(G_{\delta}^{\lambda}\right)\right]
= \bE\left[H_{i,\lambda}\, \varphi_{\delta}\left(G_{\delta}^{\lambda}\right)\right],
\]
where for any $p\geq1,$ 
\[
\sup_{\lambda\in[0,1]}\left\Vert H_{i,\lambda}\right\Vert _{p}\leq_{c}\left\Vert H\right\Vert _{i,p+\frac{1}{2}}.
\]
\end{prop}

\begin{proof}
Since $G_{\delta}^{\lambda}$ belongs to $\bD^{3,\infty}$ and is non-degenerate ($\gamma_{G_{\delta}^{\lambda}}>0$), the existence of
$H_{i,\lambda}$ follows from \cite[Proposition 2.1.4]{nual:06},
and its $L^{p}$-norm is controlled owing to \cite[Inequality (2.32), p.102]{nual:06}:
\begin{align*}
\sup_{\lambda\in\left[0,1\right]}\left\Vert H_{i,\lambda}\right\Vert _{p} & \leq_{c}\left\Vert H\right\Vert _{i,p+\frac{1}{2}}\sup_{\lambda\in\left[0,1\right]}\left\Vert \gamma_{G_{\delta}^{\lambda}}^{-1}DG_{\delta}^{\lambda}\right\Vert _{i,2^{i-1}p(2p+1)}^{i}.
\end{align*}
The second factor on the right-hand side is finite thanks to Lemmas \ref{lemma:estimate:malliavin:derivative:G}, \ref{lemma:estimate:gammaG:-1:norm:i:p}, and to  H\"{o}lder
inequality for the norms $\left\Vert \cdot\right\Vert _{k,p}$ \cite[Proposition 1.5.6]{nual:06}. 
\end{proof}
Recalling the representation \eqref{eq:error:E1:G} and applying Proposition
\ref{prop:ipp:malliavin} with $H=I^{(3)}(\ve)$, there exists
$H_{1,0}(\ve)$ such that 
\begin{align*}
E_{1}(\varphi_{\sqrt{2}\delta}^{\prime}) & =\int_{0}^{1}\frac{(1-\ve)^{2}}{2}\bE\left[\varphi_{\delta}^{\prime}\left(G_{\delta}^{0}\right)I^{(3)}(\ve)\right]\dd\ve=\int_{0}^{1}\frac{(1-\ve)^{2}}{2}\bE\left[\varphi_{\delta}\left(G_{\delta}^{0}\right)H_{1,0}(\ve)\right]\dd\ve \,.
\end{align*}
Therefore, applying \eqref{eq:estimate:I(n):norm:k:p} with $n=3$, we get
\begin{align}
\left|E_{1}(\varphi_{\sqrt{2}\delta}^{\prime})\right| & \leq_{c}\left\Vert \varphi_{\delta}\left(G_{\delta}^{0}\right)\right\Vert _{2}\sup_{\ve\in\left[0,1\right]}\left\Vert I^{(3)}(\ve)\right\Vert _{1,\frac{5}{2}}\leq_{c}\Delta^{3(d_{1}\wedge\frac{d_{2}}{2})}.\label{eq:estimate:E:1}
\end{align}
In the last step of \eqref{eq:estimate:E:1}, we have applied the easy estimate $\left\Vert \varphi_{\delta}(G_{\delta}^{\lambda})\right\Vert _{p}=\cO\left(1\right)$, uniformly in $\Delta$, which follows from \eqref{eq:estimate:phidelta:phi}
and the sub-linearity of the (H\"{o}lder continuous) function $\varphi$.

In order to estimate $E_{0}(\varphi_{\sqrt{2}\delta}^{\prime \prime})$, we start from \eqref{eq:error:E0:G} and apply Proposition
\ref{prop:ipp:malliavin} with $H=I^{(2)}(\ve_{1})I^{(2)}(\ve_{2})$, which yields
\begin{align}
\left|E_{0}(\varphi_{\sqrt{2}\delta}^{\prime \prime})\right| & \leq\int_{0}^{1}\int_{0}^{1}\int_{0}^{1}(1-\lambda)(1-\ve_{1})(1-\ve_{2})\left|\bE\left[\varphi_{\delta}\left(G_{\delta}^{\lambda}\right)H_{2,\lambda}(\ve_{1},\ve_{2})\right]\right|\dd\lambda \, \dd\ve_{1} \, \dd\ve_{2}\nonumber \\
 & \leq_{c}\sup_{\ve_{1}\in[0,1],\ve_{2}\in[0,1]}\left\Vert I^{(2)}(\ve_{1})I^{(2)}(\ve_{2})\right\Vert _{{2,\frac{5}{2}}}\leq_{c}\Delta^{4(d_{1}\wedge\frac{d_{2}}{2})} \,,\label{eq:estimate:E:0}
\end{align}
where the last inequality follows from \eqref{eq:estimate:I(n):norm:k:p} with $n=2$.
Now combining \eqref{eq:proof:expansion:error:en:delta} with
\eqref{eq:estimate:E:1} and \eqref{eq:estimate:E:0}, Theorem \ref{thm:expansion:plain:model}
is proved.

\subsection{Proof of Corollary \ref{cor:plain:model:call:put:future}}
From Proposition \ref{prop:proxies:characteristics} and Theorem \ref{thm:expansion:plain:model},
we have first 
\[
\Pcall_{0}=\bE\bigl[\varphi\bigl(\VIX_{T,{\rm P}}^{2}\bigr)\bigr]=\bE_{Z\sim\cN\left(0,1\right)}\bigl[\bigl(e^{\frac{\mup}{2}+\frac{\sigmap^{2}}{8}}e^{-\frac{\sigmap^{2}}{8}+\frac{\sigmap}{2}Z}-\kappa\bigr)_{+}\bigr]=C_{\BS}\bigl(e^{\frac{\mup}{2}+\frac{\sigmap^{2}}{8}},\kappa,\frac{\sigmap}{2}\bigr).
\]
Introducing the new function 
\[
\Psi:\ve\in\left(0,\infty\right)\mapsto\bE\bigl[\varphi\bigl(\VIX_{T,{\rm P}}^{2}e^{\ve}\bigr)\bigr]=C_{\BS}\bigl(e^{\frac{\ve}{2}+\frac{\mup}{2}+\frac{\sigmap^{2}}{8}},\kappa,\frac{\sigmap}{2}\bigr),
\]
standard computations yield 
\begin{align*}
P_{1}^{\call}=\Psi^{\prime}(0) & =\frac{1}{2}e^{\frac{\mup}{2}+\frac{\sigmap^{2}}{8}}\Delta_{\BS}\bigr(e^{\frac{\mup}{2}+\frac{\sigmap^{2}}{8}},\kappa,\frac{\sigmap}{2}\bigr),\\
P_{2}^{\call}=\Psi^{\prime \prime}(0) & =\frac{\Psi^{\prime}(0)}{2}+\frac{e^{\mup+\frac{\sigmap^{2}}{4}}}{4}\Gamma_{\BS}\bigr(e^{\frac{\mup}{2}+\frac{\sigmap^{2}}{8}},\kappa,\frac{\sigmap}{2}\bigr),\\
P_{3}^{\call}=\Psi^{(3)}(0) & =-\frac{\Psi^{\prime}(0)}{2}+\frac{3\Psi^{\prime \prime}(0)}{2}+\frac{1}{8}e^{\frac{3\mup}{2}+\frac{3\sigmap^{2}}{8}}{\rm Speed}_{\BS}\bigr(e^{\frac{\mup}{2}+\frac{\sigmap^{2}}{8}},\kappa,\frac{\sigmap}{2}\bigr).
\end{align*}
For $(P_{i}^{{\rm F}})_{i\in\{0,1,2,3\}}$, computations are the same as for $(P_{i}^{{\rm \call}})_{i\in\{0,1,2,3\}}$ but with $\kappa=0$. For put
options, use the Black--Scholes put--call parity for $P_{0}^{{\rm put}}$, the identity
$\Deltaput\left(x,y,\sigma\right)=-1+\Deltacall\left(x,y,\sigma\right)$
for $P_{1}^{{\rm put}}$, and the fact that $\Gamma_{\BS},{\rm Speed}_{\BS}$
are equal for call and put options with same characteristics. The put--call parity \eqref{eq:put:call:parity:expansion}
directly follows.

\subsection{Proof of Proposition \ref{prop:charact:onefactor:flat}}

\emph{Estimates \eqref{eq:assu:drift} and
\eqref{eq:assu:diffusion}. }Easy computations give $\nu\left(\Kdot(t)^{2}\right)=\frac{\omega^{2}}{2k\Delta}(e^{-2k(T-t)}-e^{-2k(T+\Delta-t)})=\frac{\omega^{2}}{2k\Delta}e^{-2k(T-t)}(1-e^{-2k\Delta})$
and 
\begin{equation}
\int_{0}^{T}\left[K^{u}(t)^{2}-\nu\left(\Kdot(t)^{2}\right)\right]\dd t=\frac{\omega^{2}}{4k^{2}\Delta}\left(1-e^{-2kT}\right)\left(2k\Delta e^{2k(T-u)}+e^{-2k\Delta}-1\right).\label{eq:int:0:T:K:squared:drift:onefactor}
\end{equation}
Introducing $v=\frac{u-T}{\Delta}$, we infer that 
\begin{align*}
\Gamma_{\Delta,T,p}=\frac{\omega^{2}\left(1-e^{-2kT}\right)}{4k^{2}}\Delta\left(\int_{0}^{1}\left|\frac{2k\Delta e^{-2k\Delta v}+e^{-2k\Delta}-1}{\Delta^{2}}\right|^{p}\dd v\right)^{\frac{1}{p}}.
\end{align*}
We easily check that 
\begin{align}
\lim_{\Delta\to0}\frac{(2k\Delta e^{-2k\Delta v}+e^{-2k\Delta}-1)}{\Delta^{2}}=2k^{2}(1-2v),\qquad\text{uniformly in }v\in[0,1].\label{eq:conf:uniform:integrande:bergomi}
\end{align}
Therefore, using $\int_{0}^{1}\left|1-2v\right|^{p}\dd v=\frac{1}{1+p}$,
we obtain the first asymptotics in \eqref{eq:expansion:asymptotic:bergomi}.\\
Similarly, easy computations give $\nu\left(\Kdot(t)\right)=\frac{\omega}{k\Delta}\big(e^{-k(T-t)}-e^{-k(T+\Delta-t)}\bigr)=\frac{\omega}{k\Delta}e^{-k(T-t)}\big(1-e^{-k\Delta}\bigr)$
and 
\begin{equation}
\int_{0}^{T}\left[K^{u}(t)-\nu\left(\Kdot(t)\right)\right]^{2}\dd t=\frac{\omega^{2}(1-e^{-2kT})}{2k^{3}\Delta^{2}}(e^{-k\Delta}-1+k\Delta e^{k(T-u)})^{2}.\label{eq:int:0:T:K:squared:diffusion:onefactor}
\end{equation}
The change of variables $v=\frac{u-T}{\Delta}$ leads to 
\begin{align*}
\Lambda_{\Delta,T,p}=\frac{\omega^{2}(1-e^{-2kT})}{2k^{3}}\Delta^{2}\left(\int_{0}^{1}\left|\frac{e^{-k\Delta}-1+k\Delta e^{-k\Delta v}}{\Delta^{2}}\right|^{2p}\dd v\right)^{\frac{1}{p}},
\end{align*}
and we still conclude using \eqref{eq:conf:uniform:integrande:bergomi}.
\\
\emph{Proxy's mean and variance. }From the definition \eqref{eq:mean:variance:proxy}
and the previous expression of $\nu\left(\Kdot(t)^{2}\right)$, we
easily get 
\begin{align*}
\mup & =X_{0}-\int_{0}^{T}\frac{\omega^{2}}{4k\Delta}e^{-2k(T-t)}(1-e^{-2k\Delta})\dd t
\end{align*}
in addition $\mup\xrightarrow[\Delta\to0]{}X_{0}-\frac{\omega^{2}}{4k}(1-e^{-2kT})<\infty,$
thus \eqref{eq:assu:limit:mean:proxy} holds.\\
 Similarly, for the variance, start from \eqref{eq:mean:variance:proxy}
and the previous expression of $\nu\left(\Kdot(t)\right)$, it gives
\[
\sigmap^{2}=\int_{0}^{T}\frac{\omega^{2}}{k^{2}\Delta^{2}}e^{-2k(T-t)}\big(1-e^{-k\Delta}\bigr)^{2}\dd t
\]
clearly 
 $\sigmap^{2}\xrightarrow[\Delta\to0]{}\frac{\omega^{2}}{2k}\big(1-e^{-2kT}\big)\in(0,+\infty)$,
i.e., \eqref{eq:assu:limit:variance:proxy} holds.\\
 \emph{Coefficient $\gamma_{1}$. }Write $\gamma_{1}=\gamma_{10}+\gamma_{11}$
where 
\begin{align*}
\gamma_{10} & :=\frac{1}{8\Delta}\int_{T}^{T+\Delta}\left(\int_{0}^{T}\left(K^{u}\left(t\right)^{2}-\nu_{0}\left(\Kdot\left(t\right)^{2}\right)\right)\dd t\right)^{2}\dd u,\\
\gamma_{11} & :=\frac{1}{2\Delta}\int_{T}^{T+\Delta}\left(\int_{0}^{T}\left|K^{u}\left(t\right)-\nu_{0}\left(\Kdot\left(t\right)\right)\right|^{2}\dd t\right)\dd u.
\end{align*}
From \eqref{eq:int:0:T:K:squared:drift:onefactor} and \eqref{eq:int:0:T:K:squared:diffusion:onefactor},
easy computations give 
\begin{align*}
\gamma_{10} & =\frac{\omega^{4}}{128k^{4}\Delta^{2}}\left(-1+k\Delta\frac{1+e^{-2k\Delta}}{1-e^{-2k\Delta}}\right)\left(1-e^{-2kT}\right)^{2}\left(1-e^{-2k\Delta}\right)^{2},\\
\gamma_{11} & =\frac{\omega^{4}}{8k^{3}\Delta^{2}}\left((2+k\Delta)e^{-k\Delta}-2+k\Delta\right)\left(1-e^{-2kT}\right)\left(1-e^{-k\Delta}\right).
\end{align*}
\emph{Coefficient $\gamma_{2}$. }Again, some easy computations give that
for every $u\in\cA,$ 
\[
\int_{0}^{T}\nu\left(\Kdot\left(t\right)\right)\left[K^{u}\left(t\right)-\nu\left(\Kdot\left(t\right)\right)\right]\dd t=\frac{\omega^{2}}{2k^{3}\Delta^{2}}\left(1-e^{-2kT}\right)\left(1-e^{-k\Delta}\right)\left(e^{-k\Delta}-1+k\Delta e^{k(T-u)}\right).
\]
Recalling \eqref{eq:int:0:T:K:squared:drift:onefactor} and integrating
in $u$, we then obtain $\gamma_{2}$ after some standard computations.
\emph{$\rhd$ Coefficient $\gamma_{3}$.}
Squaring the previous equality and reintegrating in $u$,
we obtain $\gamma_{3}$ after some easy computations. 

\subsection{Proof of Proposition \ref{prop:charact:rough:flat}}

\emph{Estimates \eqref{eq:assu:drift} and
\eqref{eq:assu:diffusion}. } Easy computations first give $\nu\big(\Kdot(t)^{2}\big)=\frac{\eta^{2}}{2H\Delta}\big[(T+\Delta-t)^{2H}-(T-t)^{2H}\big]$
and 
\[
\int_{0}^{T}[K^{u}(t)^{2}-\nu(\Kdot(t)^{2})]\dd t=\frac{\eta^{2}}{2H}\Bigl(u^{2H}-(u-T)^{2H}+\frac{T^{2H+1}+\Delta^{2H+1}-(T+\Delta)^{2H+1}}{\Delta(2H+1)}\Big).
\]
Introducing the new variable $y=\frac{u-T}{\Delta}$, we infer that
\[
\Gamma_{\Delta,T,p}=\frac{\eta^{2}}{2H}\Big(\int_{0}^{1}\left|\left(T+\Delta y\right)^{2H}-\Delta^{2H}y^{2H}+\frac{T^{2H+1}+\Delta^{2H+1}-(T+\Delta)^{2H+1}}{\Delta(2H+1)}\right|^{p}\dd y\Big)^{\frac{1}{p}}.
\]
Hence, as $\left(T+\Delta y\right)^{2H}-\Delta^{2H}y^{2H}+\frac{T^{2H+1}+\Delta^{2H+1}-(T+\Delta)^{2H+1}}{\Delta(2H+1)}=HT^{2H-1}\Delta\left(2y-1\right)+\Delta^{2H}\big(\frac{1}{2H+1}-y^{2H}\big)+\cO(\Delta^{2})$
with a remainder which is uniform on $y\in\left[0,1\right],$ we have
\[
\Gamma_{\Delta,T,p}\underset{\Delta\to0}{\sim}\1_{H>\frac{1}{2}}\frac{\eta^{2}T^{2H-1}}{2}\Delta\Bigl(\int_{0}^{1}\left|1-2y\right|^{p}\dd y\Bigr)^{\frac{1}{p}}+\1_{H<\frac{1}{2}}\frac{\eta^{2}}{2H}\Delta^{2H}\Bigl(\int_{0}^{1}\left|\frac{1}{2H+1}-y^{2H}\right|^{p}\dd y\Bigr)^{\frac{1}{p}}.
\]
Since  
$\int_{0}^{1}\left|1-2y\right|^{p}\dd y=\frac{1}{p+1}$, this completes
the proof of the asymptotics of $\Gamma_{\Delta,T,p}$.

Now we handle $\Lambda_{\Delta,T,p}$. Note that 
\begin{align*}
\Lambda_{\Delta,T,p} & =\Big(\frac{1}{\Delta}\int_{v=T}^{T+\Delta}\Bigl|\int_{s=0}^{T}\Bigl[\frac{1}{\Delta}\int_{\alpha=T}^{T+\Delta}\left\{ K^{v}(s)-K^{\alpha}(s)\right\} \dd\alpha\Bigr]^{2}\dd s\Bigr|^{p}\dd v\Big)^{\frac{1}{p}}=:\eta^{2}T^{2H}I(\frac{\Delta}{T},p),
\end{align*}
where we have introduced the new variables $\beta=\frac{\alpha-T}{\Delta},u=\frac{v-T}{\Delta},t=\frac{T-s}{T},$
and the function 
\[
I:(\Delta,p)\in(0,\infty)^{2}\mapsto\Big(\int_{u=0}^{1}\Bigl|\int_{t=0}^{1}\Bigl[\int_{\beta=0}^{1}\Bigl\{\left(\Delta u+t\right)^{H-\frac{1}{2}}-\left(\Delta\beta+t\right)^{H-\frac{1}{2}}\Bigr\}\dd\beta\Bigr]^{2}\dd t\Bigr|^{p}\dd u\Big)^{\frac{1}{p}}.
\]
Integrating in $\beta$ and defining the new variable $s=\frac{t}{\Delta},$
we have 
\begin{align*}
\frac{\left(H+\frac{1}{2}\right)^{2}I\left(\Delta,p\right)}{\Delta^{2H}} 
& =\Big(\int_{u=0}^{1}\Bigl|\int_{s=0}^{\frac{1}{\Delta}}\left((1+s)^{H+\frac{1}{2}}-s^{H+\frac{1}{2}}-(H+\frac{1}{2})(u+s)^{H-\frac{1}{2}}\right)^{2}\dd s\Bigr|^{p}\dd u\Big)^{\frac{1}{p}}\\
 & =\Big(\int_{u=0}^{1}\Bigl|\int_{s=0}^{\frac{1}{\Delta}}(f_{s}(1)-f_{s}(0)-f_{s}^{\prime}(u))^{2}\dd s\Bigr|^{p}\dd u\Big)^{\frac{1}{p}},
\end{align*}
where we have set $f_{s}:v\in(0,1)\mapsto(v+s)^{H+\frac{1}{2}}$ for
$s>0.$ When $\Delta\to0$, let us show that the double integral above
converges to the same integral replacing $\frac{1}{\Delta}$ by $+\infty$
and is finite. As the integral is increasing in $\Delta$, it suffices
to show that the square is bounded by a function $C(s,u)$ such that
$\int_{0}^{1}\Bigl|\int_{0}^{\infty}C(s,u)\dd s\Bigr|^{p}\dd u<\infty.$
For $s\in[0,1],$ we can take
\[
C(s,u):=3(2^{2H+1}+1+(H+\frac{1}{2})^{2}(s^{2H-1}\vee2^{2H-1})),
\]
while for $s\geq1,$ applying twice Taylor's theorem, we choose
\begin{align*}
\left|f_{s}(1)-f_{s}(0)-f_{s}^{\prime}(u)\right|^{2} & =\left|(1-u)^{2}\int_{0}^{1}f_{s}^{\prime \prime}(u+\lambda(1-u))(1-\lambda)\dd\lambda+u^{2}\int_{0}^{1}f_{s}^{\prime \prime}((1-\lambda)u)(1-\lambda)\dd\lambda\right|^{2}\\
 & \leq(H^{2}-\frac{1}{4})^{2}s^{2H-3}=:C(s,u).
\end{align*}
The function $C$ is integrable as required recalling that $H\in(0,1),$
and we conclude remembering that $\Lambda_{\Delta,T,p}=\eta^{2}T^{2H}I(\frac{\Delta}{T},p)$.
\\
\emph{Proxy's mean and variance.} From \eqref{eq:mean:variance:proxy}
and the previous formula for $\nu\big(\Kdot(t)^{2}\big)$, we have
for the mean 
\begin{align*}
\mup & =X_{0}-\frac{1}{2}\int_{0}^{T}\frac{\eta^{2}}{2H\Delta}\big[(T+\Delta-t)^{2H}-(T-t)^{2H}\big]\dd t
\end{align*}
Using that $\frac{1}{\Delta}[(T+\Delta)^{2H+1}-\Delta^{2H+1}-T^{2H+1}]\underset{\Delta\to0}{\sim}(2H+1)T^{2H}$,
we infer that $\mup\xrightarrow[\Delta\to0]{}X_{0}-\frac{\eta^{2}T^{2H}}{4H}<\infty.$
For the variance $\sigmap^{2}$, write 
\begin{align*}
\sigmap^{2} & =\frac{\eta^{2}}{\Delta^{2}\left(H+\frac{1}{2}\right)^{2}}\int_{0}^{T}\Big[\left(t+\Delta\right)^{2H+1}+t^{2H+1}-2\left(t+\Delta\right)^{H+\frac{1}{2}}t^{H+\frac{1}{2}}\Big]\dd t\\
 & =\frac{\eta^{2}}{\Delta^{2}\left(H+\frac{1}{2}\right)^{2}}\Big[\frac{\left(T+\Delta\right)^{2H+2}+T^{2H+2}-\Delta^{2H+2}}{2H+2}-2\int_{0}^{T}\left(t+\Delta\right)^{H+\frac{1}{2}}t^{H+\frac{1}{2}}\dd t\Big].
\end{align*}
We can write
\begin{align*}
\sigmap^{2} & =\frac{\eta^{2}}{\Delta^{2}(H+\frac{1}{2})^{2}}\int_{0}^{T}\big[(t+\Delta)^{H+\frac{1}{2}}-t^{H+\frac{1}{2}}\big]^{2}\dd t\\
 & =\eta^{2}\int_{0}^{T}\int_{0}^{1}\int_{0}^{1}(t+\Delta\lambda_{1})^{H-\frac{1}{2}}(t+\Delta\lambda_{2})^{H-\frac{1}{2}}\dd\lambda_{1}\dd\lambda_{2}\dd t.
\end{align*}
Since $(t+\Delta\lambda_{1})^{H-\frac{1}{2}}(t+\Delta\lambda_{2})^{H-\frac{1}{2}}$
converges to $t^{2H-1}$ as $\Delta\to0$ and is upper bounded by the integrable function $t^{2H-1}\1_{H<\frac{1}{2}}+(T+1)^{2H-1}\1_{H\geq\frac{1}{2}}$ for $\Delta$ small enough, from the dominated convergence theorem we have $\sigmap^{2}\xrightarrow[\Delta\to0]{}\frac{\eta^{2}T^{2H}}{2H}\in(0,+\infty).$
\qed

\subsection{Proof of Theorem \ref{thm:expansion:mixed:model}}
We proceed as in the proof of Theorem \ref{thm:expansion:plain:model}.
The connection between $\VIX_{T}^{2}$ and $\VIX_{T,{\rm P}}^{2}$
is made through the interpolation
\begin{equation}
	J\left(\ve\right):=\nu\left(\xidot0\right)\int_{\cA}\left[\lambda e^{\nu_{0}\left(Y_{T,1}^{\cdot}\right)+\ve\left(Y_{T,1}^{u}-\nu_{0}\left(Y_{T,1}^{\cdot}\right)\right)}+\left(1-\lambda\right)e^{\nu_{0}\left(Y_{T,2}^{\cdot}\right)+\ve\left(Y_{T,2}^{u}-\nu_{0}\left(Y_{T,2}^{\cdot}\right)\right)}\right]\nu_{0}\left(\dd u\right).\label{eq:def:J:eps}
\end{equation}
Hence $J\left(1\right)=\VIX_{T}^{2}$, $J\left(0\right)=\VIX_{T,{\rm P}}^{2}$, and the $n$th derivative is given by
\begin{align*}
	J^{\left(n\right)}\left(\ve\right) & =\nu\left(\xidot0\right)\int_{\cA}\Big[\lambda\left(Y_{T,1}^{u}-\nu_{0}\left(Y_{T,1}^{\cdot}\right)\right)^{n}e^{\nu_{0}\left(Y_{T,1}^{\cdot}\right)+\ve\left(Y_{T,1}^{u}-\nu_{0}\left(Y_{T,1}^{\cdot}\right)\right)}\\
	& \qquad\qquad\qquad+\left(1-\lambda\right)\left(Y_{T,2}^{u}-\nu_{0}\left(Y_{T,2}^{\cdot}\right)\right)^{n}e^{\nu_{0}\left(Y_{T,2}^{\cdot}\right)+\ve\left(Y_{T,2}^{u}-\nu_{0}\left(Y_{T,2}^{\cdot}\right)\right)}\Big]\nu_{0}\left(\dd u\right).
\end{align*}
Applying a second-order Taylor formula with integral remainder to
$\varphi$ at the points $\VIX_{T}^{2}$ and $\VIX_{T,{\rm P}}^{2}$, we get
\begin{equation}
\mathbb{E}[\varphi(\VIX_{T}^{2})]=\mathbb{E}[\varphi(\VIX_{T,{\rm P}}^{2})]+\mathbb{E}[\varphi^{\prime}(\VIX_{T,{\rm P}}^{2})(\VIX_{T}^{2}-\VIX_{T,{\rm P}}^{2})]+E_{0}(\varphi^{\prime \prime}),\label{eq:proof:mixed:first:taylor}
\end{equation}
where $E_{0}(\varphi^{\prime \prime})$ is an error term given by 
\[
E_{0}(\varphi^{\prime \prime}):=\int_{0}^{1}\left(1-\beta\right)\mathbb{E}[\varphi^{\prime \prime}(\beta\VIX_{T}^{2}+(1-\beta)\VIX_{T,{\rm P}}^{2})(\VIX_{T}^{2}-\VIX_{T,{\rm P}}^{2})^{2}] \dd \beta.
\]
The expression for $\mathbb{E}[\varphi(\VIX_{T,{\rm P}}^{2})]$ follows
from \eqref{eq:mixed:proxy:distribution}.
Noticing that $J'(0) = 0$, we can express $\VIX_{T}^{2}-\VIX_{T,{\rm P}}^{2} = J(1) - J(0)$ as
\begin{align*}
\frac{J^{\left(2\right)}(0)}{2}+\int_{0}^{1}\frac{(1-\ve)^{2}}{2}J^{\left(3\right)}(\ve)\dd\ve & =\lambda\VIX_{T,{\rm P},1}^{2}\int_{\cA}\frac{1}{2}(Y_{T,1}^{u}-\nu_{0}(Y_{T,1}^{\cdot}))^{2}\nu_{0}(\dd u)\\
 & \qquad+(1-\lambda)\VIX_{T,{\rm P},2}^{2}\int_{\cA}\frac{1}{2}(Y_{T,2}^{u}-\nu_{0}(Y_{T,2}^{\cdot}))^{2}\nu_{0}\left(\dd u\right)\\
 & \qquad+\int_{0}^{1}\frac{(1-\ve)^{2}}{2}J^{\left(3\right)}(\ve)\dd\ve \,,
\end{align*}
where the variables $\VIX^2_{T,{\rm P},j}$ have been defined in section \ref{sec:mixed_model}.
We introduce the error term 
\[
E_{1}(\varphi^{\prime})
:=\mathbb{E} \biggl[ \varphi^{\prime}(\VIX_{T,{\rm P}}^{2})\int_{0}^{1}\frac{(1-\ve)^{2}}{2}J^{\left(3\right)}(\ve)\dd\ve \biggr]
\]
and the functions 
\begin{align*}
\Psi_{1}\left(x\right) & :=\partial_{y}\varphi\Bigl(\nu(\xidot0)\bigl[\lambda e^{x+y}+\left(1-\lambda\right)e^{\frac{\eta_{2}}{2}\left(\eta_{1}-\eta_{2}\right)\int_{0}^{T}\nu_{0}\left(K_{0}^{\cdot}\left(t\right)^{2}\right)\dd t+\frac{\eta_{2}}{\eta_{1}}x}\bigr]\Bigr)\bigr|_{y=0},
\\
\Psi_{2}\left(x\right) & :=\partial_{y}\varphi\Bigl(\nu(\xidot0)\bigl[\lambda e^{\frac{\eta_{1}}{2}\left(\eta_{2}-\eta_{1}\right)\int_{0}^{T}\nu_{0}\left(K_{0}^{\cdot}\left(t\right)^{2}\right)\dd t+\frac{\eta_{1}}{\eta_{2}}x}+\left(1-\lambda\right)e^{x+y}\bigr]\Bigr)\bigr|_{y=0}.
\end{align*}
Observing that 
\begin{align*}
\nu_{0}(Y_{T,1}^{\cdot}) & =-\frac{\eta_{1}^{2}}{2}\int_{0}^{T}\nu_{0}(K_{0}^{\cdot}(t)^{2})\dd t+\eta_{1}\int_{0}^{T}\nu_{0}(K_{0}^{\cdot}(t))\dd W_{t}\\
 & =\frac{\eta_{1}}{2}(\eta_{2}-\eta_{1})\int_{0}^{T}\nu_{0}(K_{0}^{\cdot}(t)^{2})\dd t+\frac{\eta_{1}}{\eta_{2}}\nu_{0}(Y_{T,2}^{\cdot}),
\end{align*}
and
\[
\nu_{0}(Y_{T,2}^{\cdot})=\frac{\eta_{2}}{2}(\eta_{1}-\eta_{2})\int_{0}^{T}\nu_{0}(K_{0}^{\cdot}(t)^{2})\dd t+\frac{\eta_{2}}{\eta_{1}}\nu_{0}(Y_{T,1}^{\cdot}),
\]
we have that the second term at the right-hand side of \eqref{eq:proof:mixed:first:taylor}
can be rewritten as 
\begin{align*}
\mathbb{E}\Bigl[\varphi^{\prime}\Bigl(\VIX_{T,{\rm P}}^{2}\Bigr) & \Bigl(\VIX_{T}^{2}-\VIX_{T,{\rm P}}^{2}\Bigr)\Bigr]\\
 & =\mathbb{E}\Big[\varphi^{\prime}\Bigl(\lambda\VIX_{T,{\rm P,1}}^{2}+\left(1-\lambda\right)\VIX_{T,{\rm P,2}}^{2}\Bigr)\\
 & \qquad\times\Big(\lambda\VIX_{T,{\rm P,1}}^{2}\int_{\cA}\frac{1}{2}\left(Y_{T,1}^{u}-\nu_{0}\left(Y_{T,1}^{\cdot}\right)\right)^{2}\nu_{0}\left(\dd u\right)\\
 & \qquad+\left(1-\lambda\right)\VIX_{T,{\rm P,2}}^{2}\int_{\cA}\frac{1}{2}\left(Y_{T,2}^{u}-\nu_{0}\left(Y_{T,2}^{\cdot}\right)\right)^{2}\nu_{0}\left(\dd u\right)\Big)\Big]+E_{1}(\varphi^{\prime})\\
 & 
 =\mathbb{E}\Big[\Psi_{1}\left(\nu_{0}\left(Y_{T,1}^{\cdot}\right)\right) 
 \int_{\cA}\frac{1}{2}\left(Y_{T,1}^{u}-\nu_{0}\left(Y_{T,1}^{\cdot}\right)\right)^{2}\nu_{0}\left(\dd u\right)
\\
& \qquad+\Psi_{2}\left(\nu_{0}\left(Y_{T,2}^{\cdot}\right)\right)
 \int_{\cA}\frac{1}{2}\left(Y_{T,2}^{u}-\nu_{0}\left(Y_{T,2}^{\cdot}\right)\right)^{2}\nu_{0}\left(\dd u\right)\Big]+E_{1}(\varphi^{\prime}).
\end{align*}
This is similar to the setting of \eqref{eq:dphi:nu:moins:nuexp},
now with two different functions $\Psi_{1}$ and $\Psi_{2}$.
Applying twice Lemma \ref{lem:integration:by:parts} to 
$\mathbb{E}\Bigl[\Psi_{j}\bigl(\nu_{0}\bigl(Y_{T,j}^{\cdot}\bigr)\bigr) 
 \bigl(Y_{T,j}^{u}-\nu_{0}\bigl(Y_{T,j}^{\cdot}\bigr)\bigr)^{2}\Bigr]$ for $j=\{1,2\}$ and fixed $u$, with 
\begin{align*}
l(x) & =\Psi_{j}\left(x-\frac{1}{2}\int_{0}^{T}\nu_{0}\left(\Kdot_{j}\left(t\right)^{2}\right)\dd t\right),
\quad 
a_{t}=\nu_{0}\left(\Kdot_{j}\left(t\right)\right),
\\
f_{t}(u) & =h_{t}(u)=K_{j}^{u}\left(t\right)-\nu_{0}\left(\Kdot_{j}\left(t\right)\right),
\quad 
e_{t}(u)=g_{t}(u)=-\frac{1}{2}K_{j}^{u}\left(t\right)^{2}+\frac{1}{2}\nu_{0}\left(\Kdot_{j}\left(t\right)^{2}\right),
\end{align*}
we obtain the desired expansion.

Observing that $J^{(n)}(\ve) = \lambda \, I^{(n)}_1(\ve) + (1 - \lambda) I^{(n)}_2(\ve)$, where for each $j\in\{1,2\}$ the function
$I_j(\ve) = \nu\left(\xidot0\right)\int_{\cA} e^{\nu_{0}\left(Y_{T,j}^{\cdot}\right)+\ve\left(Y_{T,j}^{u}-\nu_{0}\left(Y_{T,j}^{\cdot}\right)\right)}\nu_{0}\left(\dd u\right)$ is defined precisely as in \eqref{eq:def:I:eps},
and conducting the same error analysis as in the proof of Theorem
\ref{thm:expansion:plain:model} (without the integration by parts of Malliavin calculus, since here $\varphi\in\mathcal{C}_{b}^{2}$ is smooth), we retrieve the error estimate $\cO\bigl ( \Delta^{3(d_{1}\wedge\frac{d_{2}}{2})}\bigr)$. 
\qed

\bibliographystyle{abbrvnat}
\bibliography{vix}

\end{document}